\documentclass[a4paper,onecolumn,11pt]{quantumarticle}
\pdfoutput=1
\usepackage[utf8]{inputenc}
\usepackage[english]{babel}
\usepackage[T1]{fontenc}
\usepackage{amsmath}
\usepackage{hyperref}

\usepackage{graphics}
\usepackage{graphicx}
\usepackage{amsthm}
\usepackage{here}
\usepackage{arydshln}
\usepackage{setspace}
\usepackage{wrapfig}
\usepackage{array, amssymb, latexsym, bm, mathtools, braket, multirow, url}
\usepackage{color}
\usepackage{cancel}
\usepackage{comment}
\usepackage{ascmac}
\usepackage{tikz}
\usepackage{paralist}

\usepackage[all]{xy}
\usepackage{xr}

\usetikzlibrary{positioning, arrows.meta, backgrounds}

\renewcommand{\tablename}{table}
\newtheorem{theorem}{Theorem}[section]
\newtheorem{definition}[theorem]{Definition}
\newtheorem{lemma}[theorem]{Lemma}

\newtheorem{corollary}[theorem]{Corollary}

\newtheorem{assumption}[theorem]{Assumption}
\newtheorem{remark}[theorem]{Remark}
\newtheorem{example}[theorem]{Example}

\newtheorem*{intro*}{\textbf{Introduction and Motivation}}
\newtheorem*{results*}{\textbf{Results}}
\newtheorem*{discussion*}{\textbf{Application and Discussion}}
\newtheorem*{relatedworks*}{\textbf{Related works}}
\newtheorem*{outline*}{\textbf{Outline}}

\newcommand{\cB}{\mathcal{B}}

\newcommand{\cE}{\mathcal{E}}

\newcommand{\cH}{\mathcal{H}}

\newcommand{\cL}{\mathcal{L}}
\newcommand{\cM}{\mathcal{M}}

\newcommand{\cQ}{\mathcal{Q}}

\newcommand{\cS}{\mathcal{S}}

\newcommand{\cV}{\mathcal{V}}

\DeclareMathOperator{\Tr}{Tr}

\hypersetup{%
 setpagesize=false,%
 bookmarks=true,%
 bookmarksdepth=tocdepth,%
 bookmarksnumbered=true,%
 colorlinks=false,%
 pdftitle={},%
 pdfsubject={},%
 pdfauthor={},%
 pdfkeywords={}}

\begin{document}

\title{Hypothesis testing and Stein's lemma in general probability theoires with Euclidean Jordan algebra and its quantum realization}

\author{Kanta Sonoda}
\email{kanta.sonoda.c8@math.nagoya-u.ac.jp}
\affiliation{Graduate School of Mathematics, Nagoya University, Furo-cho, Chikusa-ku, Nagoya, 464-8602, Japan}

\author{Hayato Arai}
\email{h.arai6626@gmail.com}
\affiliation{
Department of Basic Science, The University of Tokyo, 3-8-1 Komaba, Meguro-ku, Tokyo 153-8902, Japan
}
\affiliation{(Previous) Mathematical Quantum Information RIKEN Hakubi Research Team, 
RIKEN Cluster for Pioneering Research (CPR) and RIKEN Center for Quantum Computing (RQC), 
Wako, Saitama 351-0198, Japan.
}

\author{Masahito Hayashi}
\email{hmasahito@cuhk.edu.cn}
\affiliation{School of Data Science, The Chinese University of Hong Kong, Shenzhen, Longgang District, Shenzhen, 518172, China}
\affiliation{International Quantum Academy, Futian District, Shenzhen 518048, China}
\affiliation{Graduate School of Mathematics, Nagoya University, Furo-cho, Chikusa-ku, Nagoya, 464-8602, Japan}

\begin{abstract}
	Even though quantum information theory gives advantage over classical information theory,
	these two information theories have a structural similarity that many exponet rates of information tasks asymptotically equal to entropic quantities.
	A typical example is Stein's Lemma,
	which many researchers still keep interested in.
	In this paper, in order to analyze the mathemtaical roots of the structural similarity,
	we investigate mathematically minimum structure where Stein's Lemma holds.
	We focus on the structure of Euclidean Jordan Algebras (EJAs),
	which is a generalization of the algebraic structure in quantum theory,
	and we investigate the properties of general models of General Probabilistic Theories (GPTs) generated by EJAs.
As a result, we prove Stein's Lemma in any model of GPTs generated by EJAs
	by establishing a generalization of information theoretical tools from the mathematical properties of EJAs.
\end{abstract}

\maketitle

\section{Introduction}

\subsection{Overview}

Over the past decades, quantum information theory has emerged and flourished as an extension of classical information theory.
Even though quantum information theory has given many information protocols outperforming the bound performance in classical information theory,
these two theories have a structural similarity that many rates of information tasks asymptotically equal to entropic quantities.
One prominent example is Stein's lemma in hypothesis testing \cite {Holevo1975,Hayashi2,Nagaoka2001,HN2003,Hayashi1,HiaiPetz,OgawaNagaoka,Hayashi3}, which characterizes the optimal error exponent for state discrimination by the relative entropy in both classical and quantum theories.
This similarity can be considered as a reflection of ``classicalizations" in the proof of quantum Stein's lemma \cite{Hayashi1,HiaiPetz,OgawaNagaoka,Hayashi3},
represented by pinching.
However, as we understood the recent active works about generalized Stein's lemma \cite{G-Stein,Gap-G-Stein,G-Stein-HY,G-Stein-L}, we found it quite difficult to clarify the valid scope of such classicalizations,
which is far from fully understanding.

To explore the fundamental origins of the similarity, we start with a mathematical generalization of both classical and quantum models: General Probabilistic Theories (GPTs) \cite{Plavala,Janotta,Barret,KMI2009,Muller2013,KBBM2017,Matsumoto2018,Arai1,MAB2022,PNL2023,Short2010,Barnum2012,ALP2019,ALP2021,AH2024,BGW2020}.
GPTs provide a framework for describing general probabilistic models based only on operational axioms of states and measurements, rather than the postulates of quantum mechanics. 
This approach allows us to examine the mathematically universal structures in probabilistic models.
However, the studies of GPTs have clarified two important deficiency in general models, non-unique and non-canonical composite model \cite{Janotta,Barret} and inconsistent definitions of entropic quantities \cite{Short2010,Barnum2012}.
Because of the two deficiency of concepts, it is almost impossible to recover asymptotic rates by entropic quantities in general models, in contrast to classical and quantum theories.
The deficiency implies the additional mathematical structure for asymptotic behavior of entropic quantities.

In order to avoid the deficiency and to disucss asymptotic behavior of entropic quantities,
we focus on 
Euclidean Jordan Algebra (EJA) \cite{BGW2020,Jordan,Faraut,Olsen,Barnum1,Barnum2,Barnum3,Gowda,Mccrimmon}, which is a generalization of the algebraic structures of classical and quantum theories.
EJAs include not only classical and quantum theories but also alternative mathematical models such as quaternionic quantum systems, octonionic quantum systems,
and other type of models called Lorentz type.
Crucially, EJAs possess unique spectral decomposition and canonical composition, which enable rigorous analysis of asymptotic problems.
Therefore, we investigate hypothesis testing in GPTs associated with EJAs,
and we prove a generalized version of standard Stein's lemma in all EJAs.
Our result clarifies that EJA is the core mathematical principles underlying the relation between asymptotic exponent rates and entropic quantities.
Moreover,
our result is significant in terms of studies of EJAs because we recover the asymptotic equation between an exponent rate and an entropic quantity in quantum composite systems
in contrast to the previous studies discussing mathematical properties of a single system \cite{BGW2020,Barnum1,Barnum2,Barnum3}.

In the next section, we give a brief mathematical and technical overview of the whole discussion: definition of entropic quantities, development of information theoretical tools, and proof of Stein’s Lemma.
Roughly speaking,
we define entropic quantities, for example, relative entropy, Petz Relative R\'{e}nyi entropy, and Sandwiched Relative R\'{e}nyi entropy, through spectral decomposition
and investigate asymptotic behaviors of the spectrum of independent identical distribution (i.i.d.) states in the canonical composite system associated with EJAs.
Then,
we prove Stein's lemma even in any general models associated with EJAs,
i.e.,
the asymptotic equation between relative entropy of two states $\rho,\sigma$ and the exponent of type II error under the $\epsilon$-constraint of type I error of hypothesis testing of two i.i.d. states $\rho^{\otimes n},\sigma^{\otimes n}$.

Furthermore, we explore a more intuitive reason why EJA is the core structure of the relation.
We show that all models in GPTs associated with all EJAs can be canonically embedded into higher-dimensional quantum systems,
except for the case of Octonion, which is called exceptional because it cannot be canonically embedded into any other EJAs \cite{Jordan, Mccrimmon}.
Actually, this finding does not make the proof of Stein's lemma in EJAs trivial, but the embeddings give an alternative proof of Stein's lemma in almost all EJAs.
Moreover, the embedding clarifies the physical meaning of model of GPTs associated with EJAs.
Even though the studies of GPTs have become popular, few results \cite{AH2024} gives a rigorous physical implementation of models in GPTs,
our work is also a new result of such a direction.

In summary, we extend Stein's Lemma to a more general class of probabilistic models and provide a new proof using the structure of EJAs.
These findings deepen our understanding of the fundamental structure of the synchronized results that asymptotic information rates are given by entropic quantities.
Our results suggest that key principles of the synchronized results is the structure of EJAs, which is not only offering new directions for exploring probabilistic models in physics and information theory but also providing mathematical essence of standard quantum information theory.

\subsection{Proof Sketches and Outline of the Paper}

Now, we explain the whole organization of this paper and the sketch of the proof of Stein's Lemma in EJAs.
We draw the important implications of the proofs as Figure~\ref{figure:structure}, roughly.
Here, we remark that
all non-cited statements are proven in this paper.
However, we only write proofs of essential statements in main part of this paper.
Other proofs are written in Appendix.

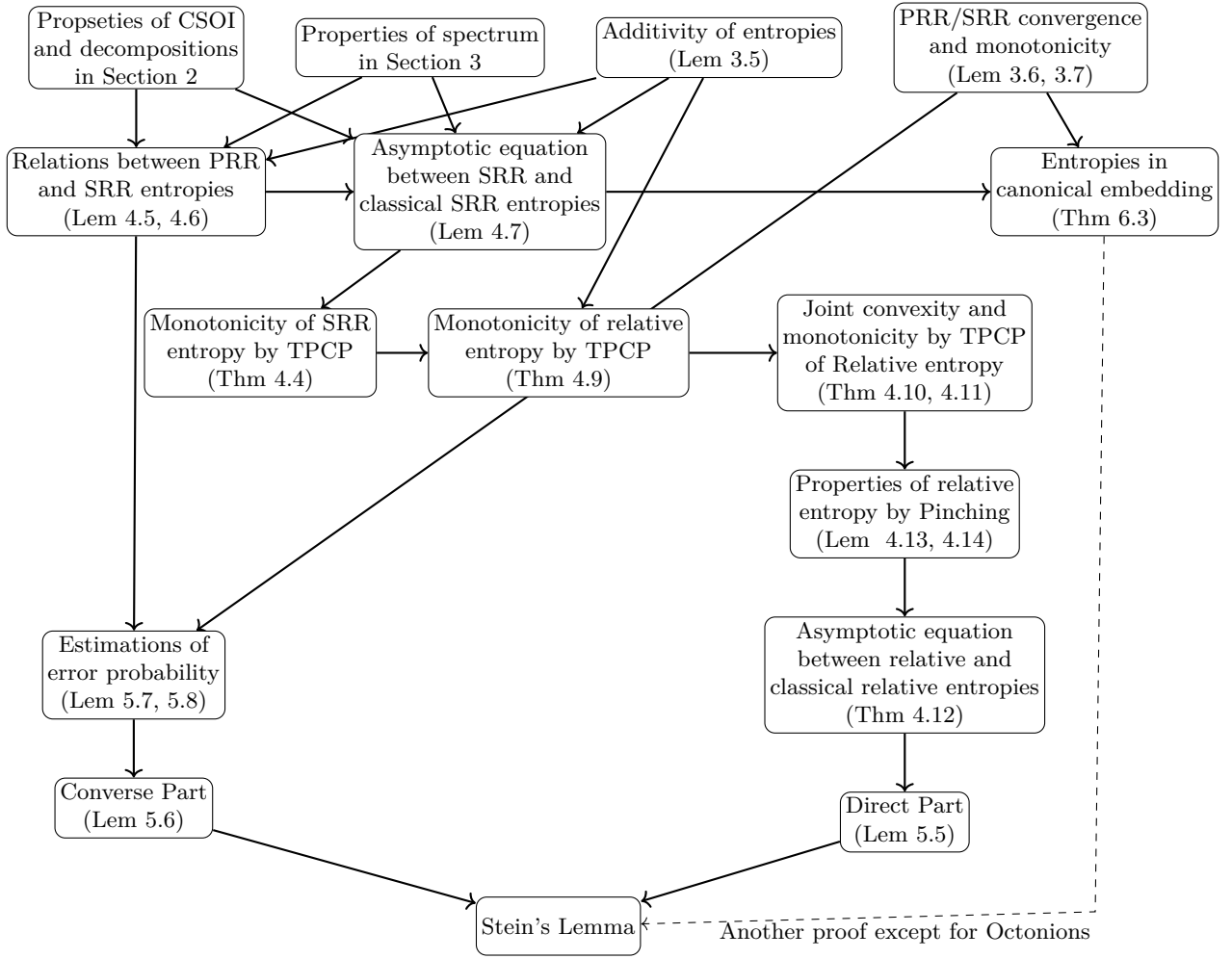
\begin{figure}[H]
\caption{The structure of the proofs.}
\label{figure:structure}
\centering
\begin{tikzpicture}[
  node distance=0.8cm and 2.2cm,
  nodestyleA/.style={
    rectangle, draw, fill=white, rounded corners,
    align=center, minimum height=0.8cm,
    inner xsep=2pt, inner ysep=2pt, font=\footnotesize
  },
textnode/.style={
    draw=none,        
    fill=none,        
    align=center,
    font=\footnotesize
  },
  arrowA/.style={->, thick},
  arrowB/.style={->, dashed}
  ]

\node[nodestyleA] (CSOI) at (0,0) {Propseties of CSOI\\ and decompositions\\ in Section~\ref{section:preliminary}};
\node[nodestyleA,right=of CSOI, xshift=-1.5cm] (Spectrum) {Properties of spectrum\\in Section~\ref{section:Information theorical tools}};
\node[nodestyleA,right=of Spectrum, xshift=-1.5cm] (Additivity) {Additivity of entropies\\(Lem~\ref{theorem:additivity})};
\node[nodestyleA,right=of Additivity, xshift=-1.5cm] (ConvPRRSRR) {PRR/SRR convergence\\
and monotonicity\\(Lem \ref{lemma:limitofRenyies}, \ref{lemma:monotonicityins})};

\node[nodestyleA,below=of CSOI] (L45L46) {Relations between
PRR\\ and SRR entropies \\ (Lem \ref{lemma:inequalityofSandwich}, \ref{lemma:observationinequalityofSandwich})};
\node[nodestyleA,right=of L45L46, xshift=-1cm] (L47) {Asymptotic equation \\
between SRR and \\classical SRR entropies\\(Lem \ref{lemma:limofSandwich})};

\node[nodestyleA,below=of L47, xshift=-3cm] (T44) {Monotonicity of SRR\\ entropy by TPCP\\(Thm \ref{theorem:MonotonicityofSandwich})};
\node[nodestyleA,right=of T44, xshift=-1.5cm] (T49) {Monotonicity of relative\\entropy by TPCP\\(Thm \ref{theorem:MonotonicityofRelative})};
\node[nodestyleA,right=of T49, xshift=-1cm] (T410) {Joint convexity and\\ monotonicity by TPCP\\ of Relative entropy\\(Thm \ref{theorem:jointconvexity}, \ref{theorem:observationinequalityofRelative})};
\node[nodestyleA,below=of T410] (L413414) {Properties of relative\\ entropy by Pinching\\ (Lem ~\ref{lemma:DirectpartofRelative1},~\ref{lemma:DirectpartofRelative2})};

\node[nodestyleA,below=of L413414] (T412) {Asymptotic equation \\
between relative and \\classical relative entropies\\(Thm~\ref{theorem:DirectpartofRelative})};
\node[nodestyleA,left=7.36cm of T412] (L57L58) {Estimations of\\error probability\\(Lem \ref{lemma:converse1}, \ref{lemma:converse2})};
\node[nodestyleA,below=of L57L58] (Converse) {Converse Part\\(Lem \ref{lemma:converse})
};

\node[nodestyleA,below=of T412] (Direct) {Direct Part\\(Lem \ref{lemma:direct})};

\node[nodestyleA,below=of Converse, xshift=5.8cm] (Main) {Stein's Lemma};

\node[nodestyleA,right=of L47, xshift=3.05cm] (QRep){Entropies in\\ canonical embedding\\(Thm \ref{TT3})};


\begin{scope}[on background layer]

\draw[arrowA] (CSOI) -- (L45L46);
\draw[arrowA] (Spectrum) -- (L45L46);
\draw[arrowA] (Additivity) -- (L45L46);

\draw[arrowA] (L45L46) -- (L47);
\draw[arrowA] (L45L46) -- (L57L58);
\draw[arrowA] (Additivity) -- (L47);
\draw[arrowA] (CSOI) -- (L47);
\draw[arrowA] (Spectrum) -- (L47);

\draw[arrowA] (L47) -- (T44);
\draw[arrowA] (T44) -- (T49);
\draw[arrowA] (T49) -- (T410);
\draw[arrowA] (Additivity) -- (T49);

\draw[arrowA] (ConvPRRSRR) -- (L57L58);
\draw[arrowA] (L57L58) -- (Converse);

\draw[arrowA] (T410) -- (L413414);
\draw[arrowA] (L413414) -- (T412);
\draw[arrowA] (T412) -- (Direct);

\draw[arrowA] (Direct) -- (Main);
\draw[arrowA] (Converse) -- (Main);

\draw[arrowA] (ConvPRRSRR) -- (QRep);
\draw[arrowA] (L47) -- (QRep);
\draw[arrowB] (QRep) -- (13.1,-11.9) -- (Main);

\node[textnode] (label) at (10.5,-12.2) {Another proof except for Octonions};

\end{scope}

\end{tikzpicture}
\end{figure}

\subsubsection{Contents in Section
~\ref{section:preliminary}
}

Section~\ref{section:preliminary} introduces mathematical frameworks of GPTs and EJAs.
Besides, we give many important prperties of EJAs in this section.

In Section~\ref{subsection:FrameworkGPTs},
we define the framework of GPTs,
which is a generalization of classical and quantum theory.
A model of GPTs is defined as a tuple of positive cone $\cQ\subset\cV$ and an unit effect $u\in\cV$ for a finitie-dimensional real vector space $\cV$ with inner product $\langle \ , \ \rangle$.
The main objects in a model of GPTs are a state $\rho$ and a measurement $\bm{M}$ defined as an element $\rho\in\cQ$ with $\langle \rho,u\rangle=1$ and a family $\bm{M}:=\{M_i\}_{i\in I}$ of the dual cone $\cQ^\ast$ satisfying $\sum_{i\in I} M_i=u$, respectively.
Also, we give many important notations, for example state space $\cS(\cQ,u)$ and measurement class $\cM(\cQ,u)$, in this section. 

In Section~\ref{subsection:EuclideanJordanalgebra},
we give the mathematical definition of EJAs and the relation between EJAs and GPTs with examples including classical theory and quantum theory.
An EJA is defined as a finite-dimensional real vector space with special type of non-associative product $\circ$, called Jordan product (Definition~\ref{definition:EuclideanJordanalgebra}).
However indeed, except for a one type called Lorentz type, all ``simple" EJAs are classified as the set of Hermitian matrices with a normed-division-algebra-valued-entries, i.e., real $\mathbb{R}$, complex $\mathbb{C}$, quaternion $\mathbb{H}$, and octonion $\mathbb{O}$ valued-entries, with the product $\frac{1}{2}\left(XY+YX\right)$ (Table~\ref{table:SimpleEJA}).
Moreover, all EJAs are written as a direct sum of simple EJAs.
In other words, the above types of simple EJAs are essential parts of EJAs.
We do not consider a concrete EJA but an abstract structure of EJAs for the proof of Stein's Lemma, but the classification is important for discussions in Section~\ref{section:quantumrealization}.
Next, we define the canonical composite systems associated with two EJAs (Definition~\ref{definition:CompositesysteminEJA}), which is important part for the $n$-shot scenario in this work.

In Section~\ref{subsection:ConceptsinEJA},
we give some important concepts and show their properties.
First, we introduce \textit{Complete System of Orthogonal Idempotents} (CSOI) and \textit{Jordan frame},
which correspond to the projections in quantum theory.
As important propositions of CSOI and Jordan frame,
we see two types of decomposition, \textit{spectral decomposition} (Theorem~\ref{theorem:Spectraldecomposition}) and \textit{Peirce decomposition} (Theorem~\ref{theorem:peircedecomposition}).
Spectral decomposition in EJAs, a decomposition on CSOI, is just a generalization of spectral decomposition of Hermitian matrices.
Peirce decomposition is a generalization of basis decomposition composed by projections and interferences of Hermitian matrices.

Second, we introduce a linear map $P_x$ called \textit{quadratic form} (Definition~\ref{definition:quadraticform}) for $x\in \cV$, which induces an important map, so-called \textit{pinching} map in quantum theory.
Then, we see some properties of the quadratic form and the above two decomposition (Theorem~\ref{theorem:simultaneousspectral} and Lemma~\ref{lemma:positivityofQradratic}, \ref{lemma:projectionPeirce}, \ref{lemma:orthogonalPeirce}, and \ref{lemma:newCSOI}), which recovers the important properties of entropic quantities for the proof of Stein's Lemma in Section~\ref{section:Information theorical tools} and \ref{section:Informationquantities}.

\subsubsection{Contents in Section
~\ref{section:Information theorical tools}
}

Section~\ref{section:Information theorical tools} develops information theorical tools as an extention of quantum information theory for the proof of Stein's Lemma.

In Section~\ref{subsection:InformationQuantities},
we define entropic quantities, including \textit{Pets Relative R\'{e}nyi} (PRR) entropy and \textit{Sandwiched Relative R\'{e}nyi} (SRR) entropy, from the spectral decomposition and the CSOI (Definition~\ref{definition:fofEJAelement} and \ref{definition:informationquantities}).
In EJAs, as the spectral decomposition, a state $\rho$ has the unique form
\begin{align}
	\rho=\sum_i\lambda_i c_i,
\end{align}
where $\lambda_i\in\mathbb{R}$ and $\{c_i\}_i$ is a CSOI.
Then, we define $f(\rho)$ as
\begin{align}
	\rho=\sum_if(\lambda_i) c_i
\end{align}
for a real function $f$, and we can define the above entropies.

Next, we prove some essential properties of PRR and SRR entropies:
additivity on tensor product (Lemma~\ref{theorem:additivity}),
\begin{align}
        D(\rho_1\otimes\rho_2 ||\sigma_1\otimes \sigma_2)                   & =D(\rho_1||\sigma_1)+D(\rho_2||\sigma_2).                                                                        \\
        D_{1+s}(\rho_1\otimes \rho_2||\sigma_1\otimes \sigma_2)             & =D_{1+s}(\rho_1||\sigma_1)+D_{1+s}(\rho_2||\sigma_2). \\
        \underline{D}_{1+s}(\rho_1\otimes \rho_2||\sigma_1\otimes \sigma_2) & =\underline{D}_{1+s}(\rho_1||\sigma_1)+\underline{D}_{1+s}(\rho_2||\sigma_2),
    \end{align}
convergence (Lemma~\ref{lemma:limitofRenyies}),
\begin{align}
        \lim_{s\to 0}D_{1+s}(\rho||\sigma)             & =\lim_{s\to 0}\frac{\phi(-s|\rho||\sigma)}{s}=D(\rho||\sigma).         \\
        \lim_{s\to 0}\underline{D}_{1+s}(\rho||\sigma) & =\lim_{s\to 0}\frac{\tilde{\phi}(-s|\rho||\sigma)}{s}=D(\rho||\sigma),
    \end{align}
and monotonicity (Lemma~\ref{lemma:monotonicityins}) on $s$,
from the properties of spectral decomposition.
Besides, we prove 
Jennsen's inequality for any convex function (Lemma~\ref{lemma:Jenseninequality})
and a bound of the number of distinct eigenvalues (Lemma~\ref{lemma:numberofeigenvalues}), i.e.,
$|C_{x^{\otimes n}}|\le (n+1)^{d-1}$
from the properties of CSOI shown in Section~\ref{subsection:ConceptsinEJA}.

In Section~\ref{subsection:pinching},
we define a generalization of a pinching map (Definition~\ref{definition:PinchingCSOI} and \ref{definition:PinchingState})
as
 \begin{align}
        \kappa_\sigma(\rho):=\sum_iP_{c_i}\rho,
\end{align}
where $P_{c_i}$ is the quadratic form of $c_i$ in Section~\ref{subsection:ConceptsinEJA}.
Next, we show that any two states are classical after pinching  (Lemma~\ref{lemma:PinchingStateisState} and \ref{lemma:classicallyofpinching}),
which prove some lemmas in the next part in this section.
Second, we define an important measurement, called \textit{pinchied measurement}, as
\begin{align}
        M^\rho_\sigma:=\{P_{c_{i,j}}M_k\}_{i,j,k},
    \end{align}
(Definition~\ref{definition:pinchedmeasurement}).
Then, we show two important properties: the following relation between relative entropy with pinching states and classical entropy with pinchied measurement (Lemma~\ref{lemma:RepresentEntropies})
\begin{align}
        \underline{D}_{1+s}(\kappa_\sigma(\rho)||\sigma)=D_{1+s}(\kappa_\sigma(\rho)||\sigma)=D_{1+s}(P_\rho^{M^\rho_\sigma}||P_\sigma^{M^\rho_\sigma})
\end{align}
and pinching inequality (Lemma~\ref{lemma:PinchingInequality}).
These properties are also shown by the properties of CSOI shown in Section~\ref{subsection:ConceptsinEJA}
and play an essential role for the proof of the direct part of Stein's Lemma.

In Section~\ref{subsection:tpcp},
we define Trace Preserving and Completely Positivity (TPCP) in EJAs (Defiition~\ref{definition:tracepreserving} to \ref{definition:tpcp})
and basic properties of TPCP map (Lemma~\ref{lemma:TPCPtensor} an d\ref{lemma:identitypreserving}).
Next, we check that partial trace and measurement are TPCP map (Lemma~\ref{lemma:TPCPofPartialtrace} and \ref{lemma:TPCPofObservation}).

\subsubsection{Contents in Section
~\ref{section:Informationquantities}
}

Section~\ref{section:Informationquantities} analyzes three information quantities, PRR entropy (in Section~\ref{subsection:relativerenyi}),
SRR entropy (in Section~\ref{subsection:sandwichedrelative}),
and Relative entropy, respectively (in Section~\ref{subsection:relative}).

The main goal is to prove Theorem~\ref{theorem:DirectpartofRelative},
i.e.,
the following relation of relative entropy with the pinchied measurement $I^{\rho^{\otimes n}}_{\sigma^{\otimes n}}$ defined in Definition~\ref{definition:pinchedmeasurement}:
\begin{align}
        \lim_{n\to\infty}\frac{1}{n}D^{I^{\rho^{\otimes n}}_{\sigma^{\otimes n}}}(\rho^{\otimes n}||\sigma^{\otimes n})=D(\rho||\sigma),
\end{align}
which shows the direct part of Stein's Lemma by combining classical Stein's Lemma.
Theorem~\ref{theorem:DirectpartofRelative} is shown by the following relations:
\begin{align}
	D(\rho||\sigma)&\ge D(P^M_\rho||P^M_\sigma) \quad\quad (\mbox{Theorem~\ref{theorem:observationinequalityofRelative}}),\\
	D(\rho||\sigma)&=D(\rho||\kappa_\sigma(\rho))+D(\kappa_\sigma(\rho)||\sigma) \quad\quad (\mbox{Lemma~\ref{lemma:DirectpartofRelative1}}),\\
	D(\rho||\kappa_C(\rho))&=H(\kappa_C(\rho))-H(\rho)\le \log|C| \quad\quad (\mbox{Lemma~\ref{lemma:DirectpartofRelative2}}).
\end{align}

Lemma~\ref{lemma:DirectpartofRelative1} is directly shown from the definition of entropy and EJAs in Appendix~\ref{appendix:proofofRelative}.
In Appendix~\ref{appendix:proofofRelative}, Lemma~\ref{lemma:DirectpartofRelative2} is shown by the joint convexity, i.e.,
\begin{align}
	D(x||y)=\sum_{i=1}^k p_iD(\rho_i||\sigma_i),\quad\quad \mbox{(Theorem~\ref{theorem:jointconvexity})}
\end{align}
and the properties of CSOI and pinching in Section~\ref{subsection:ConceptsinEJA}.
Theorem~\ref{theorem:jointconvexity} and Theorem~\ref{theorem:observationinequalityofRelative} are shown by monotonicity of relative entropy by TPCP map (Theorem~\ref{theorem:MonotonicityofRelative}), i.e., the following relation:
\begin{align}\label{intro:eq-relative}
	D(\rho||\sigma)\ge D(\kappa(\rho)||\kappa(\sigma))
\end{align}

To prove Theorem~\ref{theorem:MonotonicityofRelative} is the main aim of the first part of Section~\ref{section:Informationquantities}.
The relation \eqref{intro:eq-relative} is recovered by the convergence of SRR entropy and the same relation for SRR entropy (Theorem~\ref{theorem:MonotonicityofSandwich}), i.e., the following relation:
\begin{align}
        \underline{D}_{1+s}(\rho||\sigma)\ge \underline{D}_{1+s}(\kappa(\rho)||\kappa(\sigma)),\quad s>0.
\end{align}
This relation is proven by the fact that SRR entropy is represented by the asymptotic classical SRR entropy with the optimal measurement, i.e.,
the following relation (Lemma~\ref{lemma:limofSandwich}):
\begin{align}
        \underline{D}_{1+s}(\rho||\sigma)=\lim_{n\to\infty}\frac{1}{n}\max_{M^n}D_{1+s}(P^{M^n}_{\rho^{\otimes n}}||P^{M^n}_{\sigma^{\otimes n}}),\quad s>0.
\end{align}

Lemma~\ref{lemma:limofSandwich} is shown in Appendix~\ref{appendix:proofofSandwich} with conbining many lemmas,
Lemma~\ref{lemma:inequalityofSandwich}, Lemma~\ref{lemma:observationinequalityofSandwich},
Lemma~\ref{theorem:additivity},
properties of Pinching, and the number of spectrum in Section~\ref{subsection:pinching}.
Lemma~\ref{lemma:observationinequalityofSandwich} states the following relation:
\begin{align}
        \underline{D}_{1+s}(\rho||\sigma)\ge D_{1+s}(P^M_\rho||P^M_\sigma),\quad s>0,
\end{align}
which is also important for the proof of the converse part of Stein's Lemma.

Lemma~\ref{lemma:inequalityofSandwich} and \ref{lemma:observationinequalityofSandwich} are proven in Appendix~\ref{appendix:proofofSandwich},
but, an essential part to prove these lemmas is the same as the proof of monotonicity of PRR entropy with observation (Theorem~\ref{theorem:MonotonicityofRelativerenyi}).
We give Theorem~\ref{theorem:MonotonicityofRelativerenyi} in the main part for reader's convenience.
Theorem~\ref{theorem:MonotonicityofRelativerenyi} states the following relations:
\begin{align}
        D_{1+s}(\rho||\sigma)\ge \lim_{n\to \infty}\frac{1}{n}D_{1+s}(\kappa_{\sigma^{\otimes n}}(\rho^{\otimes n})||\sigma^{\otimes n})\ge D_{1+s}(P^M_\rho||P^M_\sigma)\quad(s>0),
\end{align}
which is the first statement in this section.
Theorem~\ref{theorem:MonotonicityofRelativerenyi} is also shown by many lemmas,
Lemma~\ref{lemma:Monotonicity of Relativerenyi1},
Lemma~\ref{lemma:Monotonicity of Relativerenyi2},
Lemma~\ref{theorem:additivity},
properties of Pinching, and the number of spectrum in Section~\ref{subsection:pinching}.

\subsubsection{Contents in Section
~\ref{section:hypothesistesting}
}

Section~\ref{section:hypothesistesting} discusses hypothesis testing in GPTs and prove Stein's Lemma.

In Section~\ref{subsection:setting},
we introduce the setting of hypothesis testing in GPTs.
Our aim is to analyze the following error probability with asymmetric setting of hypothesis testing:
\begin{align}
        \beta^n_\epsilon(\rho||\sigma):=\min_{0\le T\le u}\{\langle\sigma^{\otimes n},T\rangle|\langle\rho^{\otimes n},u-T\rangle\le \epsilon\},\quad 0<\epsilon<1,
\end{align}
We prove Stein's Lemma, i.e., the following relation:
\begin{align}
        \lim_{n\to \infty}-\frac{1}{n}\log\beta^n_\epsilon(\rho||\sigma)=D(\rho||\sigma).
\end{align}
In order to show this relation,
we introduce the following two exponents
\begin{align}
        B(\rho||\sigma)&:=\sup_{\{0\le T_n\le u\}}\left\{\varliminf_{n\to\infty}-\frac{1}{n}\log\langle\sigma^{\otimes n},T_n\rangle\mid \lim_{n\to\infty}\langle\rho^{\otimes n},u-T_n\rangle=0\right\},\\
        B^\dag(\rho||\sigma)&:=\sup_{\{0\le T_n\le u\}}\left\{\varliminf_{n\to\infty}-\frac{1}{n}\log\langle\sigma^{\otimes n},T_n\rangle\mid \varliminf_{n\to\infty}\langle\rho^{\otimes n},u-T_n\rangle<1\right\},
\end{align}
and show the direct part and converse part.
The direct part, i.e., the relation
\begin{align}
	B(\rho||\sigma)\ge D(\rho||\sigma),
\end{align}
is proven by Theorem~\ref{theorem:DirectpartofRelative} and classical Stein's Lemma.
The converse part, i.e., the relation
\begin{align}
        D(\rho||\sigma)\ge B^\dag(\rho||\sigma),
\end{align}
is proven by Lemma~\ref{lemma:converse1} and \ref{lemma:converse2},
which are shown by Lemma~\ref{lemma:limofSandwich}, 
Lemma~\ref{lemma:limitofRenyies},
and Lemma~\ref{lemma:monotonicityins}.

\subsubsection{Contents in Section
~\ref{section:quantumrealization}
}

In this section,
we give another perspective of the reason why Stein's Lemma holds even in EJAs through an embedding from some types of EJAs to quantum theory.

In Section~\ref{section:CanonicalJordanSub},
we define canonical Jordan subalgebras
and show that a corresponding state space and measurement space in canonical Jordan subalgebras can be regarded as a quotient space of the original state space and measurement space (Theorem~\ref{theorem:embed-measurement}
and \ref{theorem:embed-state}).

In Section~\ref{section:CanonicalEmbedding},
we define canonical embedding map and show that
canonical embedding map does not change SRR entropy and relative entropy (Theorem~\ref{TT3}) by applying Lemma~\ref{lemma:limitofRenyies} and Theorem~\ref{theorem:embed-measurement} in the previous sections.
As a result, we give another proof of Stein's Lemma if there exists canonical embedding map from a model into quantum theory.

In Section~\ref{section:CanonicalJordanSub} and Section~\ref{section:CanonicalEmbedding},
we see that two types of EJAs, Lorentz type and Quaternion type, satisfy the assumption of Theorem~\ref{TT3}.
As we see in Section~\ref{subsection:EuclideanJordanalgebra},
except for the octonion type,
any EJA is composed of real $\mathbb{R}$ and complex $\mathbb{C}$ types of Hermitian matrices and the above two types.
In other words, any EJA is canonically embedded into quantum theory if the EJA does not contain an octonion part, and as a result, we conclude that Stein's Lemma holds in such types of EJAs.
The existence of such canonical embedding maps for Lorentz type and Quaternion type are known in \cite{BGW2020}.
However, we give a new relation between Lorentz type and fermion annihilation and creation operators
and we recover the construction in \cite{BGW2020} by our new relation and Jordan-Wigner transformation \cite{JW1928}.

Here, we remark that we need
Lemma~\ref{lemma:limitofRenyies} for both the direct proof in Section~\ref{section:hypothesistesting} and another proof via quantum realization in Section~\ref{section:quantumrealization}.
Moreover, the direct proof in Section~\ref{section:hypothesistesting} is valid even if an EJA does not contain an octonion part.
Therefore, we need to prove Stein's Lemma directly from the definition of EJAs, as we show since Section~\ref{section:hypothesistesting},
which is the main contribution of this work.

\subsubsection{Contents in Section
~\ref{section:conclusion}
}
Finally, we conclude this paper in Section~\ref{section:conclusion}.
We give a summary of our results and open problems.

\subsubsection{Contents in Appendix}

We give the proofs of some statements in Appendix
if the statements are not co essentially related to the main structure of the whole paper.

\subsection{Abbreviations and Notations}

\begin{table}[h]
\caption{}
\centering
\begin{tabular}{cl}
	Abbreviation & Original \\
	\cline{1-2}
	\cline{1-2}
	GPTs & General Probabilistic Theories\\
	\cline{1-2}
	EJAs & Euclidean Jordan Algebras\\
	\cline{1-2}
	HT & Hypothesis Testing\\
	\cline{1-2}
	i.i.d. & independent and identical distribution \\
	\cline{1-2}
	CSOI & Complete System of Orthogonal Idempotents\\
	\cline{1-2}
	CSOPI & Complete System of Orthogonal Primitive Idempotents\\
	\cline{1-2}
	PRR entropy & Petz Relative R\'{e}nyi entropy\\
	\cline{1-2}
	SRR entropy & Sandwiched Relative R\'{e}nyi entropy\\
	\cline{1-2}
	TPCP & Trace Preserving and Completely Positive\\
	\cline{1-2}
\end{tabular}
\end{table}

\begin{table}[H]
\centering
\begin{tabular}{cll}
	Notation & Meaning & Ref\\
	\cline{1-3} \cline{1-3}
	$\cV$ & A finite-dimensional real vector space with inner product $\langle \ \rangle$ & \\
	\cline{1-3}
	$\cQ$ & A positive cone in a finite-dimensional real vector space $\cV$ & Def. \ref{definition:PositiveofGPTs}\\
	\cline{1-3}
	$\cQ^\ast$ & The dual cone of a positive cone $\cQ$ & Def.~\ref{definition:DualconeofGPTs}\\
	\cline{1-3}
	$\le_\cQ$ & The partial order defined by a positive cone $\cQ$ & Def. \ref{definition:orderinPotivecone}\\
	\cline{1-3}
	$\cS(Q,u)$ & The state space defined by a positive cone $\cQ$ and a unit $u$ & Def. \ref{definition:StateEffectMeasurement} \\
	\cline{1-3}
	$\cE(Q,u)$ & The effect space defined by the dual cone of $\cQ$ and a unit $u$ & Def. \ref{definition:StateEffectMeasurement}\\
	\cline{1-3}
	$M(Q,u)$ & The measurement space defined by the dual cone of $\cQ$ and a unit $u$ & Def. \ref{definition:StateEffectMeasurement}\\
	\cline{1-3}
	$P^{\bm{M}}_\rho$ & The probability distribution obtained by a state $\rho$ and a measurement $\bm{M}$ & Def. \ref{definition:distribution}\\
	\cline{1-3}
	$D(p||q)$ & The classical relative entropy for probability distributions $p$ and $q$  & Def. \ref{definition:classicalrelative} \\
	\cline{1-3}
	$D_{1+s}(p||q)$ & The classical relative R\'{e}nyi entropy for probability distributions $p$ and $q$  & Def. \ref{definition:classicalrelative} \\
	\cline{1-3}
	\multirow{2}{*}{$D^{\bm{M}}(\rho||\sigma)$} & The classical relative entropy associated with the probability distribution & \multirow{2}{*}{Def. \ref{definition:classicalize-relative}}\\
	& obtained by states $\rho,\sigma$ and a measurement $\bm{M}$ &\\
	\cline{1-3}
	\multirow{2}{*}{$D^{\bm{M}}_{1+s}(\rho||\sigma)$} & The classical relative R\'{e}nyi entropy associated with the probability & \multirow{2}{*}{Def.~\ref{definition:classicalize-relative}}\\
	& distribution obtained by states $\rho,\sigma$ and a measurement $\bm{M}$ &\\
	\cline{1-3}
	$\circ$ & Jordan product & Def. \ref{definition:EuclideanJordanalgebra}\\
	\cline{1-3}
	$\cQ_\cV$ & The positive cone associated with an EJA $\cV$ & Def. \ref{definition:positiveconeofEJA}\\
	\cline{1-3}
	$\mathrm{tr}x$ & The trace of an element $x$ in $\cV$ & Def. \ref{definition:trace}\\
	\cline{1-3}
	$\bm{C}_x$ & The CSOI determined by spectral decomposition of an element $x$  & Def.~\ref{definition:CSOIofx}\\
	\cline{1-3}
	$L_x$ & The linear map take the Jordan product with $x$ & Def.~\ref{definition:linearmap}\\
	\cline{1-3}
	$P_x$ & The quadratic form of $x$ &Def.~\ref{definition:quadraticform} \\
	\cline{1-3}
	$\otimes$ & The tensor prodocut in a bipartite vector space & Def.~\ref{definition:CompositesysteminEJA}\\
	\cline{1-3}
	$f(\rho)$ & The state determined by a state $\rho$ and a function $f$ & Def.~\ref{definition:fofEJAelement}\\
	\cline{1-3}
	$H(\rho)$ & von Neumann entropy of a state $\rho$ & Def.~\ref{definition:informationquantities}\\
	\cline{1-3}
	$D(\rho||\sigma)$ & Relative entropy of states $\rho$ over $\sigma$ & Def.~\ref{definition:informationquantities}\\
	\cline{1-3}
	$D_{1+s}(\rho||\sigma)$ & Petz Relative R\'{e}nyi entropy of states $\rho$ over $\sigma$ & Def.~\ref{definition:informationquantities}\\
	\cline{1-3}
	$\underline{D}_{1+s}(\rho||\sigma)$ & Sandwiched Relative R\'{e}nyi entropy of states $\rho$ over $\sigma$ & Def.~\ref{definition:informationquantities}\\
	\cline{1-3}
	$\kappa_{\bm{C}}$ & The pinching map determined by CSOI $\bm{C}$ & Def.~\ref{definition:PinchingCSOI}\\
	\cline{1-3}
	$\kappa_\sigma$ & The pinching map determined by a state $\sigma$ &Def.~\ref{definition:PinchingState} \\
	\cline{1-3}
	$M^\rho_\sigma$ & The measurement determined by pinchied state $\kappa_\sigma(\rho)$ &Def.~\ref{definition:pinchedmeasurement} \\
	\cline{1-3}
	$\mathrm{tr}_{\cV_1}$ & The partial trace map over $\cV_1$ &Def.~\ref{definition:partialtrace} \\
	\cline{1-3}
	$\kappa_{\bm{M}}$ & The observation map by a measurement $\bm{M}$ & Def.~\ref{definition:TPCPobservation}\\
	\cline{1-3}
	\multirow{2}{*}{$\beta^n_\epsilon(\rho||\sigma)$} & The optimal second type error under first type error constraint & \multirow{2}{*}{Def.~\ref{definition:beta}}\\
	& for hypothesis testing of $\rho$ and $\sigma$ & \\
	\cline{1-3}
	$B(\rho||\sigma)$ & Stein exponent with 0 error & Def.~\ref{definition:BandBdag}\\
	\cline{1-3}
	$B^\dag(\rho||\sigma)$ & Stein exponent with arbitral error &Def.~\ref{definition:BandBdag} \\
	\cline{1-3}
\end{tabular}
\end{table}

\newpage

\section{Preliminaries}\label{section:preliminary}

\subsection{Framework of GPTs}\label{subsection:FrameworkGPTs}
As a preliminary, we define some mathematical objects about GPTs.
At first, we define a positive cone and a dual cone,
which are the most basic concepts in GPTs.
Next, by using a positive cone, a dual cone and an unit effect,
we define operational concepts, i.e., a state, an effect and a measurement.
We consider these operational concepts in order to treat information theorical problems.
Next, after we define a probabilistic distribution, we prepare some well-known classical entropies.
These classical entropies will appear when we measure a state in an Euclidean Jordan algebra in later Section.
Finally, we define a composite model of GPTs.
We deal with the composite model of GPTs when we handle $n$ separate systems,
which means that we can operate information-theoritically one system repeatedly.
In this part, the space $\cV$ is denoted as a finite-dimensinal real vector space equipped with an inner product.

\begin{definition}[cone\cite{Faraut}{[Chapter1-1]}]
    A subset $\cQ\subset \cV$ is called a cone if $x\in \cQ$ and $\lambda\in\mathbb{R}_+$ imply $\lambda x\in \cQ$.
\end{definition}
We define the most basic mathematical object in GPTs as follows.
\begin{definition}[Positive cone]\label{definition:PositiveofGPTs}
    A subset $\cQ\subset \cV$ is called as a positive cone
    if $\cQ$ is a cone and holds following 3 conditions.
    \begin{itemize}
        \item[(1)]$\cQ$ has an interior point.
        \item[(2)]$\cQ\cap (-\cQ)=\{0\}$.
        \item[(3)]$\cQ$ is a closed convex set.
    \end{itemize}
\end{definition}


Now, we define another basic concept, dual cone, by using a positive cone.
\begin{definition}[Dual cone\cite{Faraut}{[Chapter1-1]}]\label{definition:DualconeofGPTs}
    A dual cone $\cQ^* \subset \cV$ of a positive cone $\cQ$ is defined as
    \begin{align}
        \cQ^*:=\{x\in \cV|\langle x,y\rangle\ge 0,\forall y\in \cQ\}.
    \end{align}
\end{definition}

The following Lemma about a dual cone holds.
\begin{lemma}[\cite{Faraut}{[Chapter1-1]}]\label{lemma:DualPositive}
    A dual cone $\cQ^*$ of a positive cone $\cQ$ is also a positive cone.
\end{lemma}

Now, we define an order in a positive cone.
This order is a convenient concept
because the dual cone satisfying Lemma~\ref{lemma:DualPositive} has a nice property of an inner product (Definition~\ref{definition:DualconeofGPTs}).
\begin{definition}[Order in Positive cone]\label{definition:orderinPotivecone}
    We define an order $\le_\cQ$ in a positive cone $\cQ$
    as $x\le_\cQ y\Leftrightarrow y-x\in \cQ$.
\end{definition}
This order in a positive cone $\cQ$ is a partial order as follows.
\begin{lemma}[Partial order]\label{lemma:partialorder}
    An order of Definition~\ref{definition:orderinPotivecone} over a positive cone $\cQ$
    is a partial order over $\cQ$.
\end{lemma}
From now on, we denote this partial order over a positive cone $\cQ$ as $\le_\cQ$.
When the positive cone $\cQ$ is given obviously, we abbreviate $\le_\cQ$ as $\le$.

Now, we can describe the set of states, measurements and effects.
\begin{definition}\label{definition:StateEffectMeasurement}
    Let $\cQ,\cQ^*\subset \cV$ be a positive cone and its dual cone, respectively.
    For a fixed inner point $u\in \cQ^*$ as an unit effect, we define the state space, the effect space and the measurement space as
    \begin{itemize}
        \item State space $\cS(\cQ,u):=\{\rho\in \cQ|\langle \rho,u\rangle=1\}$,
        \item Effect space $\cE(\cQ,u):=\{e \in \cQ^*|0\le \langle e,\rho\rangle \le 1,\forall \rho\in \cS(\cQ,u)\}$,
        \item Measurement class $\cM(\cQ,u):=\{\{M_i\}_{i=1}^d|M_i\in \cQ^*,\ d\in\mathbb{N}, \ \sum_{i=1}^dM_i=u\}$.
    \end{itemize}
\end{definition}
An element of the state space, the effect space and the measurement space are called
a state, an effect, and a measurement, respectively.

Next, we define the probability distribution
when a measurement is applied to a state as follows.
\begin{definition}\label{definition:distribution}
    For a measurement $\bm
{M}=\{M_i\}_{i=1}^d$ and a state $\rho$,
    we define the probability distribution as
    \begin{align}
        P^{\bm{M}}_\rho:=\left\{P^{\bm{M}}_\rho(i):=\langle M_i,\rho\rangle\right\}_{i=1}^d.
    \end{align}
\end{definition}
By Definition~\ref{definition:distribution},
we define the following (classical) Relative entropy and the (classical) Relative R\'{e}nyi entropy.
In later Section~\ref{section:Information theorical tools}, we extend these entropies to Euclidean Jordan algebraic entropies.
In fact, especially, classical Relative R\'{e}nyi entropy have two ways of an extension to Euclidean Jordan algebraic entropies
based on quantum information theory\cite{Hayashi1}.
These entropies are called Relative R\'{e}nyi entropy and Sandwiched Relative R\'{e}nyi entropy in an Euclidean Jordan algebra.
\begin{definition}[(Classical) Relative entropy]\label{definition:classicalrelative}
    Let $p=\{p_i\}_{i=1}^d$ and $q=\{q_i\}_{i=1}^d$ be two probability distributions.
    Then, we define (classical) Relative entropy $D(p||q)$ as
    \begin{align}
        D(p||q):=\sum_{i=1}^dp_i\log \frac{p_i}{q_i}.
    \end{align}
	Also, we define the (classical) Relative R\'{e}nyi entropy for $s\neq0$ as
    \begin{align}
        D_{1+s}(p||q):=\frac{1}{s}\log\sum_{i=1}^dp_i^{1+s}q_i^{-s}.
    \end{align}
\end{definition}
Since two states $\rho,\sigma$ and a measurement $\bm{M}$ give two probability distributions $P^{\bm{M}}_\rho,P^{\bm{M}}_\sigma$ by Definition~\ref{definition:distribution},
we denote the Relative entropy of Definition~\ref{definition:classicalrelative} as follows.
\begin{definition}\label{definition:classicalize-relative}
    For two states $\rho,\sigma$ and a measurement $\bm{M}$, using Definition~\ref{definition:distribution} and Definition~\ref{definition:classicalrelative},
    we denote as follows.
    \begin{align}
        D^{\bm{M}}(\rho||\sigma):=&D(P^{\bm{M}}_\rho||P^{\bm{M}}_\sigma),\\
		D^{\bm{M}}_{1+s}(\rho||\sigma):=&D_{1+s}(P^{\bm{M}}_\rho||P^{\bm{M}}_\sigma)
    \end{align}
\end{definition}


In GPTs, we focus on the following a model of GPT.
Simply speaking, a model of GPT is a minimal model in order to consider the flamework of GPTs.
\begin{definition}[Model of GPTs]\label{definition:modelofGPTs}
    A model of GPT is defined as a tuple $(\cV,\cQ,u)$, where $\cV$, $\cQ$ and $u$ are denoted as a finite-dimensional real vector space equipped with an inner product, a positive cone and an unit effect ,respectively.
\end{definition}
If we define a model of composite systems in GPTs, we can extend a size of systems.
It is important for us to evaluate the performance of information processing.
Therefore, using a model of GPT, we define an extension of system size as follows.
\begin{definition}[Model of Composite system in GPTs\cite{Plavala}]\label{definition:modelofCompositesystem}
    Let $(\cV,\cQ,u)$,$(\cV_1,\cQ_1,u_1)$ and $(\cV_2,\cQ_2,u_2)$ be models of GPTs.
    Then, the model $(\cV,\cQ,u)$ is called a model of a composite system of $(\cV_1,\cQ_1,u_1)$ and $(\cV_2,\cQ_2,u_2)$
    if the model $(\cV,\cQ,u)$ satisfies following conditions.
    \begin{itemize}
        \item[(1)]$\cV=\cV_1\otimes \cV_2$.
        \item[(2)]$\cQ_1\otimes \cQ_2\subset \cQ \subset (\cQ_1^*\otimes \cQ_2^*)^*$.
        \item[(3)]$u=u_1\otimes u_2$.
    \end{itemize}
    Here, the tensor product of two cones is defined as $\cQ_1\otimes \cQ_2:=\{\sum_i a_i\otimes b_i|a_i\in \cQ_1,b_i\in \cQ_2\}$.
\end{definition}

The first condition is derived from the Local tomography.
The Local tomography means the following postulates.
\begin{assumption}[Local tomography\cite{Janotta,Barret}]
    For a product effect $e_1\otimes e_2$, we apply this effect to the two states $\rho,\rho'\in \cS(\cQ,u)$.
    If the joint probabilities of two states are equivalent for any product effect, then $\rho=\rho'$.
\end{assumption}
We use the third condition when we apply the product measurement
$\{M_i^1\otimes M_j^2\}_{i,j=1}^{d_1,d_2}$
to the product state $\rho_1\otimes \rho_2$.
Also, this third condition is postulated under the Claim~\ref{postulate:AbilabilityofProduct}
in \cite{Plavala}{[Definition5.1]}.
The meaning to adopt of second condition is unclear.
However, if we postulate the following operational condition, we obtain this second condition.
\begin{assumption}[\cite{Plavala}{[Definition5.1]}]\label{postulate:AbilabilityofProduct}
    Let the $\cS(\cQ,u)$ and $\cE(\cQ,u)$ be the state space and the effect space of the model of composite system $\cV$, respectively.
    Then, for any states $\rho_1\in \cS(\cQ_1,u_1)$ and $\rho_2\in \cS(\cQ_2,u_2)$,
    the state  $\rho=\rho_1\otimes \rho_2$ belongs to $\cS(\cQ,u)$.
    In addition, for any effect $ e_1\in \cE(\cQ_1,u_1)$ and $e_2\in E(\cQ_2,u_2)$,
    the effect  $e=e_1\otimes e_2$ belongs to $\cE(\cQ,u)$.
\end{assumption}
Now we explain how to deduce the inclusion relation of the cones
from Assumption~\ref{postulate:AbilabilityofProduct}.
The condition for the states is used
when we show that $\cQ_1\otimes \cQ_2\subset \cQ$.
In addition, the condition for the effects is used
when we show that $\cQ_1^*\otimes \cQ_2^*\subset \cQ^*$.
Finally, we use the following two Lemmas.
\begin{lemma}[\cite{Boyd}{[Chapter2.6.1]}]
    If the relation $\cQ_1\subset \cQ_2$ holds for two positive cones $\cQ_1,\cQ_2$,
    then the following relation of two dual cones holds.
    \begin{align}
        \cQ_2^*\subset \cQ_1^*.
    \end{align}
\end{lemma}
\begin{lemma}[\cite{Faraut}{[Theorem1.1.1]},\cite{Boyd}{[Chapter2.6.1]}]
    For a positive cone $\cQ$,
    the following relation holds.
    \begin{align}
        \cQ^{**}=\cQ.
    \end{align}
\end{lemma}

\subsection{Euclidean Jordan algebra}\label{subsection:EuclideanJordanalgebra}

Now, we prepare an Euclidean Jordan algebra with some examples, which we use mainly in this paper.
First, in this section, we classify an Euclidean Jordan algebra.
In fact, all of Euclidean Jordan algebras can be decomposed to a direct sum of well-known Euclidean Jordan algebras.
Second, we treat an Euclidean Jordan algebra in GPTs flamework.
An Euclidean Jordan algebra contains a GPTs concepts, such as a positive cone and a dual cone.
Moreover, these cones in an Euclidean Jordan algebra has good properties.
Finally, we give two physical examples, a Quantum system and a Classical system in Euclidean Jordan algebra.
In addition, we investigate the properties of a classical system and a quantum system by using
the operational concepts in GPTs.

\begin{definition}[Euclidean Jordan algebra~\cite{Faraut}{[Chapter3-1]}]\label{definition:EuclideanJordanalgebra}
    A finite-dimensional real vector space $\cV$ equipped with an inner product is called as a Jordan algebra
    if $\cV$ has a bilinear map (called a Jordan product) $\circ:V\times V\to V$ and satisfies the following conditions.
    \begin{itemize}
        \item[(J1)] $x\circ y=y\circ x$.
        \item[(J2)] $x^2\circ (x\circ y)=x\circ(x^2\circ y)$.
    \end{itemize}
    In addition, if a Jordan algebra $\cV$ satisfies the following condition (J3),
    $\cV$ is called as an Euclidean Jordan algebra.
    \begin{itemize}
        \item[(J3)] $\langle x\circ y,z\rangle=\langle x,y\circ z\rangle$.
    \end{itemize}
\end{definition}
Note that (J2) is necessarily to decide $x^n$ for arbitrary $n$ uniquely. 
An Euclidean condition (J3) is equivalent to the following condition called formally real.
\begin{definition}[Formally real\cite{Faraut}{[Chapter3-1]}]
    A Jordan algebra $\cV$ is called formally real
    if $\cV$ satisfies the following condition.
    \begin{align}
        x^2+y^2=0\Rightarrow x=y=0.
    \end{align}
\end{definition}
From now on, we denote $\cV$ as an Euclidean Jordan algebra,
and we only consider an Euclidean Jordan algebra $\cV$ with an unit element $u$.
Now we define the following condition in order to normalize the inner product.
\begin{definition}[simple\cite{Faraut}{[Chapter3.4]}]\label{definition:simple}
    The space $\cV$ is said to be simple if $\cV$ does not contain any non-trivial ideal.
\end{definition}

Actually, all EJA are uniquely decomposed into simple EJAs.
\begin{lemma}[\cite{Faraut}{[Proposition3.4.4]}]\label{lemma:decompositionofV}
    The space $\cV$ is written as a direct sum of simple EJAs uniquely.
\end{lemma}

Lemma~\ref{lemma:decompositionofV} implies that
simple Euclidean Jordan algebras are essential objects in the studies of EJAs.
In fact, a simple Euclidean Jordan algebra is completely classified as follows \cite{Jordan} (Table~\ref{table:SimpleEJA}).


\begin{table}[H]
	\caption{List about Simple Euclidean Jordan Algebras}
	\label{table:SimpleEJA}
	\centering
	\begin{tabular}{cccc}
	\cline{1-4}
	vector space $\cV$ & Jordan product & inner product & unit \\ \cline{1-4} \cline{1-4}
	$\mathrm{Sym}(m,\mathbb{R})$ & $\frac{1}{2}\left(XY+YX\right)$  & $\Tr xy$ & $I$ \\ \cline{1-4}
	$\mathrm{Herm}(m,\mathbb{C})$ & $\frac{1}{2}\left(XY+YX\right)$  & $\Tr xy$ & $I$ \\ \cline{1-4}
	$\mathrm{Herm}(m,\mathbb{H})$ & $\frac{1}{2}\left(XY+YX\right)$ & $\Tr xy$ & $I$ \\ \cline{1-4}
	$\mathbb{R}\times\mathbb{R}^d$ & $\circ_\cL$  & canonical & $(1,\cdots,1)$ \\ \cline{1-4}
	$\mathrm{Herm}(3,\mathbb{O})$ & $\frac{1}{2}\left(XY+YX\right)$ & $\Tr xy$ & $I$ \\ \cline{1-4}
	\end{tabular}
\end{table}

Now, we explain the above simple EJAs: $\mathrm{Sym}(m,\mathbb{R})$, $\mathrm{Herm}(m,\mathbb{C})$, $\mathrm{Herm}(m,\mathbb{H})$, $\mathbb{R}\times \mathbb{R}^{d}$, and $\mathrm{Herm}(3,\mathbb{O})$.
The first $\mathrm{Sym}(m,\mathbb{R})$ is a real vector space of $m\times m$ size symmetric matrices.
We will investigate the direct sum of $\mathrm{Sym}(1,\mathbb{R})$ corresponding to a classical system later in this part.
The second $\mathrm{Herm}(m,\mathbb{C})$ is a real vector space of $m\times m$ size Hermitian matrices in $\mathbb{C}$.
We will investigate this second example corresponding to a quantum system later in this part.
The third $\mathrm{Herm}(m,\mathbb{H})$ is a real vector space of $m\times m$ size Hermitian matrices in $\mathbb{H}$.
The fourth $\mathbb{R}\times \mathbb{R}^{n-1}$ is called a Lorenz cone with dimension $d$.
The fifth $\mathrm{Herm}(3,\mathbb{O})$ is a real vector space of $3\times 3$ size Hermitian matrices in $\mathbb{O}$.
We define the detailed of the third, the fourth, and the fifth types of EJAs in Section~\ref{section:quantumrealization}.

Next, we explain the relation between these simple EJAs and a second example a, quantum system.
From the first to fourth ones are said to be special and the fifth one is said to be exceptional.
The special EJA can be canonically embedded into a higher-dimensional quantum system.
We will discuss the relation between this embedding and sone information quantities in Section~\ref{section:quantumrealization}.
On the other hand, it is unknown the embedding of an exceptional EJA to Quantum system.
Our one of main result imply
the possibility of an embedding of an exceptional EJA in a Quantum system.

Next, we define a model of GPTs associated with an EJA.
From Section~\ref{subsection:FrameworkGPTs}, firstly we prepare a positive cone and its dual cone in an EJA.
Secondly, we obtain a State space, an Effect space and a Measurement class in an EJA by Definition~\ref{definition:StateEffectMeasurement}.

\begin{definition}[Positive cone in Euclidean Jordan algebra~\cite{Faraut}{[Chapter3-2]}]\label{definition:positiveconeofEJA}
    We define a canonical positive cone $\cQ_\cV$ over an EJA $\cV$ by the cone $\cQ_\cV:=\{x^2|x\in \cV\}$.
\end{definition}


\begin{lemma}[\cite{Faraut}{[Chapter3-2.1]}]\label{lemma:Qispositivecone}
    The cone $\cQ_\cV$ by Definition~\ref{definition:positiveconeofEJA} satisfies the conditions of a positive cone in GPTs.(Definition~\ref{definition:PositiveofGPTs})
\end{lemma}
To prove this Lemma~\ref{lemma:Qispositivecone}, we need some additional concepts of an EJA.
Therefore, we will show in the later Section~\ref{subsection:ConceptsinEJA}.


Next, we see the self-duality of $\cQ_\cV$, i.e., $\cQ^\ast=\cQ$.
\begin{lemma}[\cite{Faraut}{[Theorem3.2.1]}]\label{lemma:selfduality}
    For an EJA $\cV$, the dual cone $\cQ^*_\cV$ of the positive cone $\cQ_\cV$
satisfies $\cQ_\cV=\cQ^*_\cV$.
\end{lemma}
We will show this Lemma~\ref{lemma:selfduality} in the Section~\ref{subsection:ConceptsinEJA}.

Recall Definition~\ref{definition:orderinPotivecone} and a self-duality of $\cQ_\cV$.
The partial orders $\le_{\cQ_\cV}$ and $\le_{\cQ^\ast_\cV}$ are equivalent.
Therefore, we denote this order as $\le$ simply.

Because of the definition of $\cQ_\cV$ and $\cQ^*_\cV$,
we obtain a state space, a effect space and a measurement class
from Definition~\ref{definition:StateEffectMeasurement},
where the unit effect $u$ is chosen as an unit element of $\cV$.


Now, we can investigate two physical examples in an EJA,
a classical system and a quantum system.
A classical system is defined as follows\cite{Janotta}.

\begin{example}[Classical system]\label{example:classicalsystem}
    We call $\cV$ is a Classical system if a real vector space $\cV=\mathbb{R}^d$
    with a canonical inner product has the following Jordan product:
    \begin{align}
        u_i\circ u_j :=\delta_{i,j}u_j,
    \end{align}
    where $u_i$ takes $1$ in $i$th element and $0$ in others,
    and where $\delta_{i,j}$ is a Kronecker delta.
    Because $\{u_i\}$ is a basis of $\cV$,
    the product of two elements $x$ and $y$ written as $x=\sum_{i=1}^d \lambda_i u_i,y=\sum_{i=1}^d \mu_i u_i$
    are given as follows.
    \begin{align}
        x\circ y=\sum_{i=1}^d \lambda_i\mu_i u_i.
    \end{align}
\end{example}
Here, we remark that the classical system is written as the direct sum of EJA of symmetric matrices.

At first, we examine the positive and the dual cones in a classical system.
For the positive cone $\cQ_\cV$ in a classical system, we obtain
\begin{align}
    x\circ x=\sum_{i=1}^dx_i^2u_i \in \cQ_\cV,\label{equation:positiveconeofClassical}
\end{align}
where $x$ is decomposed to $x=\sum_{i=1}^dx_i u_i(x_i\in\mathbb{R})$.
Because a positive cone holds a self-duality(Lemma~\ref{lemma:selfduality}),
a relation $\cQ_\cV=\cQ^*_\cV$ holds.

Secondly, we examine a state, an effect and a measurement in classical system.
The unit element is chosen as an identity element $u=\sum_{i=1}^d u_i$ in $\cV$.
Then, we see the two of the properties of a classical system, a perfect distinguishability\cite{Janotta}
and simultaneous spectrality of all elements as follows.
Any state $\rho\in \cQ_\cV$ satisfies the following relation:
\begin{align}
    1\stackrel{(a)}{=}\langle\rho, u\rangle\stackrel{(b)}{=}\langle \sum_{i=1}^d\lambda_i u_i,u\rangle=\sum_{i=1}^d\lambda_i.\label{equation:stateofclassicalsystem}
\end{align}
Here, (a) is the condition of a state (Definition~\ref{definition:StateEffectMeasurement}).
in (a), we consider the decomposition of $\rho$ as $\rho=\sum_{i=1}^d\lambda_i u_i(\lambda_i\ge 0)$ by \eqref{equation:positiveconeofClassical}.
Therefore, a state corresponds to a probability distribution $\{\lambda_i\}_{i=1}^d$.
From this result, the state space $\cS(\cQ,u)$ is the set of probability distributions with $d$-elements,
that is, $\cS(\cQ,u)$ is the convex set of pure states $u_i(i=1,\ldots,d)$.
Here, a pure state corresponds to an extremal point of the convex set in a state space.

Finally, we consider two properties of a classical system.
We characterize a classical system by a simultaneous spectral decomposition
in Appendix~\ref{appendix:conceptsinEJA} (Lemma~\ref{lemma:characterizationofClassicalsystem}).
Now we investigate a perfect distinguishability.
A perfect distinguishability of $n$ pure states $\{\rho_i\}_{i=1}^n$ means that the
exteremal effects $\{e_j\}_{j=1}^n$ single out $n$ pure states, that is,
$\langle e_j,\rho_i\rangle=\delta_{i,j}$ holds,
where $\delta_{i,j}$ is a Kronecker delta.
An extremal effect means the extremal point of the effect space $\cE(\cQ_\cV,u)$.
In a classical system, there exists $d$ pure states $u_i(i=1,\ldots,d)$.
Now we take the $d$ exteremal effects $\{e_j=u_j\}_{j=1}^d$.
Then $\langle e_j,u_i\rangle=\langle u_j,u_i\rangle=\delta_{i,j}(i,j=1,\ldots,d)$.
Therefore, in a classical system of $d$ dimension, $d$ pure
states are perfectly distinguishable.

Next, we see that a quantum system is regarded as a model of GPTs associated with the EJA $\mathrm{Herm}(m,\mathbb{C})$ \cite{Janotta,Plavala}
\begin{example}[Quantum system]\label{example:quantumsystem}
    We call $\cV$ is a Quantum system if a real vector space of complex Hermitian matrices
    with a Hilbert-Schmidt inner product has the following Jordan product:
    \begin{align}
        x\circ y:=\frac{1}{2}(xy+yx)\quad x,y\in \cV.
    \end{align}
    Here, $xy$ and $yx$ are multiplied by a matrix product.
\end{example}
We investigate the quantum system can be treated in GPTs framework.
In addition, we examine the state, the effect, and the measurement are the canonical ones in Quantum system.

At first, we examine a positive cone and a dual cone in a Quantum system.
For a positive cone $\cQ_\cV$ in a quantum system, the relation $x\circ x=\frac{1}{2}(xx+xx)=x^2(=xx)\in \cQ_\cV$ holds.
The element $x^2=xx$ has $0$ or positive eigenvalues.
Therefore, $\cQ_\cV$ is equal to the set of positive semi-definite matrices.
Besides, Lemma~\ref{lemma:selfduality} implies that the dual $\cQ_\cV^\ast$ is equivalent to $\cQ_\cV$.

Secondly, we examine the state space, the effect space, and the Measurement class.
By choosing of $u$ as an identity matrix $I$ over $\cV$,
the state space, the effect space and the Measurement class are determined as follows.
Recall of Definition~\ref{definition:StateEffectMeasurement}, a state $\rho$ satisfies the following relation:
\begin{align}
    \langle\rho,u\rangle\stackrel{(a)}{=}\mathrm{tr}\rho=1.
\end{align}
Because we choose the Hilbert-Schmidt inner product $\langle\cdot,\cdot\rangle$,
the equality (a) holds for the identity matrix $u=I$.
Therefore, a state corresponds to a density matrix, i.e.,
a positive semi-definite matrix satisfying $\mathrm{tr}\rho=1$.

Next, we examine an effect.
Recall Definition~\ref{definition:StateEffectMeasurement}, an effect $e$ satisfies the following relation:
\begin{align}
    0\le\langle\rho,e\rangle\le 1\quad\forall\rho\in \cS(\cQ_\cV,I).\label{equation:effectinQs}
\end{align}
In addition, the element $e\in \cQ_\cV^*=\cQ_\cV$ is a positive semi-definite matrix.
Therefore, an effect $e\ge 0$ holds in a matrix inequality.
On the other hand, we show the element $I-e$ is also effect as follows.
We calculate the following quantities for any $y\in\cQ_\cV$.
\begin{align}
    \langle I-e, y\rangle=\langle I,y\rangle-\langle e,y\rangle=\mathrm{tr}y-\langle e,y\rangle.\label{equation:effectinQs2}
\end{align}
Also, any $y\in\cQ_\cV$ satisfies $y/ \mathrm{tr} y\in\cS(\cQ_\cV,I)$,
and therefore,
we obtain
\begin{align}
	0\le \langle e,\frac{y}{\mathrm{tr}y}\rangle=\frac{\langle e,y\rangle}{\mathrm{tr}y}\le 1.\label{equation:effectinQs3}
\end{align} 
By combining \eqref{equation:effectinQs2} and \eqref{equation:effectinQs3}, we obtain
\begin{align}
    \langle u-e,y\rangle=\mathrm{tr}y-\langle e,y\rangle \stackrel{(a)}{\ge} \mathrm{tr}y-\mathrm{tr}y=0.
\end{align}
Now we apply \eqref{equation:effectinQs3} to (a).
Therefore, we obtain $I-e\in Q^*=Q$.
This means that a matrix $I-e$ is positive semidefinite,
which implies $I-e\ge 0$.
As a result, we obtain $0\le e\le u$.
This means $e$ is a Test (POVM element) in a Quantum system.

Finally,we examine a measurement.
Recall Definition~\ref{definition:StateEffectMeasurement}.
A measurement $\bm{M}:=\{M_i\}_{i=1}^d\in\cM(\cQ_\cV,I)$ satisfies $M_i\in \cQ^*\quad (i=1,\ldots,d)$ and $\sum_{i=1}^d M_i=I$.
The self-duality $\cQ^*=\cQ\quad(i=1,\ldots,d)$ implies $0\le M_i$.
Therefore, the family $\bm{M}$ is a POVM in a Quantum system.

\begin{remark}
	Here, we remark that EJAs give more non-trivial models of GPTs except for Classical and Quantum systems.
	A typical example of such models is given by Lorentz type, which is known as a special restriction of Quantum system in \cite{BGW2020}.
	Moreover, we show that this model is also regarded as a model determined by real and complex parts of creation and annihilation operators of Fermion in Section~\ref{section:quantumrealization}.
\end{remark}

\subsection{Concepts in Euclidean Jordan algebra}\label{subsection:ConceptsinEJA}
In this section, we introduce some concepts of an Euclidean Jordan algebra.
First, we introduce a special type of complete systems called \textit{Completely System of Orthogonal Idempotents} (CSOI), which is regarded as a generalization of projections in Quantum system.
CSOI is directly connected to two important decompositions in EJAs, Spectral decomposition and Peirce decomposition.
Thanks to these decompositions, we can analyze an EJA
in detail by applying information theoretical tools.
In addition, we will introduce the most important concept, a Quadratic form,
which is important for the definition of pinching map.
Finally, we define the canonical composite systems of EJAs.
After Section~\ref{section:Informationquantities},
we analyze asymptotic behaviors of information quantities.
Therefore, we mainly consider $n$-composite system of a single EJA.
We introduce the essential part of these concepts in this section and explain the rest part of concepts and proofs
in Appendix~\ref{appendix:conceptsinEJA}


We define special types of complete systems.
\begin{definition}[Complete system of orthogonal (primitive) idempotents\cite{Faraut}{[Chapter3-1]}]\label{definition:completesystem}
    Let $\bm{C}$ be a subset with $d$ elements in $\cV$.
    The elements in $\bm{C}=\{c_i\}_{i=1}^d$ are said to be orthogonal, idempotent ,complete, primitive
    if the elements in $\bm{C}$ satisfy the following conditions.
    \begin{itemize}
        \item[(1)]Different two elements $c_i,c_j$ are said to be orthogonal if these two elements satisfy $c_i\circ c_j =0$.
        \item[(2)]An element $c_i$ is said to be idempotent if this element satisfy $c_i^2=c_i$.
        \item[(3)]The elements $c_i(i=1,\ldots,d)$ are said to be complete system if its elements satisfy $c_1+\cdots +c_d=u$.
        \item[(4)]An element $c_i$ is said to be primitive when this element cannot be written as the sum of two non zero idempotents which is each orthogonal.
    \end{itemize}
    A family $\bm{C}$ is called Complete System of Orthogonal Idempotents (CSOI)
    if all elements in $\bm{C}$ satisfy (1)-(3) conditions.
    In addition, a family $\bm{C}$ is called Complete System of Orthogonal Primitive Idempotents (and sometimes called Jordan frame)
    if all elements in $\bm{C}$ satisfy (1)-(4) conditions.
\end{definition}

Two concepts in Definition~\ref{definition:completesystem}
are related to the important Theorems both Spectral theorem and Pierce decomposition.
Moreover, the complete system of orthogonal idempotents
mainly appear in information theorical objects
in later than Section~\ref{section:Information theorical tools}.
The following Lemma implies the concepts in Definition~\ref{definition:completesystem}
are related to operational objects in GPTs.
\begin{lemma}\label{lemma:CSOIisMeasurement}
    Let $\bm{C}=\{c_i\}$ be a complete system of orthogonal idempotents.
    Then this family $\bm{C}$ is a measurement.
	In particular, each $c_i$ is an effect.
\end{lemma}
In this setting, the following Spectral theorem holds.
\begin{theorem}[Spectral theorem\cite{Faraut}{[Theorem 3.1.1]}]\label{theorem:Spectraldecomposition}
    For $x\in \cV$, there exist unique distinct real numbers $\lambda_1,\ldots,\lambda_d$ and a unique CSOI $c_1,\ldots ,c_d$ such that
    \begin{align}
        x=\lambda_1c_1+\cdots+\lambda_d c_d.
    \end{align}
    The numbers $\lambda_i$ are said to be the eigenvalues,
    and this decomposition of $x$ is called as spectral decomposition of $x$.
    Here, the number $d$ depends on the element of $\cV$.
\end{theorem}

Similarly to Spectral theorem (\ref{theorem:Spectraldecomposition}), 
the following Spectral theorem holds for a Jordan frame.
\begin{theorem}[Spectral theorem for Jordan frame\cite{Faraut}{[Theorem 3.1.2]}]
For an element $x\in \cV$, there exists Jordan frame $\{c_i\}_{i=1}^r$ and real numbers $\{\lambda_i\}_{i=1}^r$ such that
\begin{align}
  x=\sum_{i=1}^r \lambda_i c_i.  
\end{align}
Moreover, the number $r$ is common for any $x\in\cV$.
\end{theorem}

Due to Theorem~\ref{theorem:Spectraldecomposition},
we choose the number $r_\cV$ as the number $r$ in Theorem~\ref{theorem:Spectraldecomposition} for each EJA $\cV$.
The number is called \textit{rank} of $\cV$.

However, we basically don't use this spectral theorem for primitive one
because the elements have some ways to spectral decompositions for primitive ones,
not unique similarly to Theorem~\ref{theorem:Spectraldecomposition}.
We use spectral decomposition of primitive one in Appendix~\ref{appendix:conceptsinEJA} with
the characterization of a classical system (Lemma~\ref{lemma:characterizationofClassicalsystem})

By Theorem~\ref{theorem:Spectraldecomposition},
we introduce the following notations for the future convenience.
\begin{definition}\label{definition:CSOIofx}
    For a CSOI $\bm{C}=\{c_i\}$,
    we denote $|\bm{C}|$ as the numbers of the elements in $\bm{C}$.
    In particular,
	by Definition~\ref{theorem:Spectraldecomposition},
    there exists unique Spectral decomposition for $x\in \cV$ as $x=\sum_{i=1}^d\lambda_i c_i$.
    Then, the CSOI of $x$ is denoted as $\bm{C}_x=\{c_i\}_{i=1}^d$,
	and the numbers of the elements in $\bm{C}_x$ is $|\bm{C}_x|=d$.
\end{definition}

Next, we introduce two maps including a Quadratic form.
\begin{definition}[\cite{Faraut}{[Chapter2-1]}]\label{definition:linearmap}
    We define a linear map $L_x:\cV\to \cV$ for $x\in \cV$ if $L_x$ satisfies the relation $L_x(y)=x\circ y$ for $y\in \cV$.
\end{definition}
The following Lemma is important to show the Peirce decomposition of idempotents and self-duality of the positive cone of an Euclidean Jordan algebra.
\begin{lemma}[\cite{Faraut}{[Chapter2-1]}]\label{lemma:eigenvaluesofCompletesystem}
    For an element $c_i$ in a CSOI $\bm{C}$, $L_{c_i}$ takes an eigenvalue of $0$ ,$\frac{1}{2}$ or $1$.
\end{lemma}

\begin{definition}[Quadratic form\cite{Faraut}{[Chapter2-3]}]\label{definition:quadraticform}
    The linear map $P_x(\cdot):\cV\to \cV$ for $x\in \cV$ is called as a Quadratic form if the map $P_x$ is defined as $P_x(\cdot):=2L_x(L_x(\cdot))-L_{x^2}(\cdot)$.
\end{definition}

\begin{lemma}[\cite{Olsen}{[Proposition3.3.6]}~\cite{Faraut}{[Proposition3.2.2]}]\label{lemma:positivityofQradratic}
    Let $\cQ\subset \cV$ be a positive cone.Then, for $x\in \cV,y\in \cQ$, $P_x(y)\in \cQ$ holds.
\end{lemma}
Here, we remark that $P_x$ does not equal to $L_{x^2}$ because Jordan product is non-associative.
For example, in the case of Quantum system,
the quadratic form $P_x(y)$ of $y$ is calculated as follows:
\begin{align}
	P_x(y)=&2L_x(L_x(y))-L_{x^2}(y)
	=2L_x\left(\frac{xy+yx}{2}\right)-\frac{x^2y+yx^2}{2}\\
	=&\frac{x^2y+2xyx+yx^2}{2}-\frac{x^2y+yx^2}{2}=xyx.
\end{align}

Now, we prepare some additional preparations,
which imply the decomposition of $\cV$ by a complete system of orthogonal primitive idempotents.
We use the following Theorem to prove a simultaneous spectrality
and the condition that $\cV$ is isomorphic to a classical system.

\begin{theorem}[Peirce decomposition\cite{Faraut}{[Theorem4.2.1]}]\label{theorem:peircedecomposition}
    Let $\bm{C}=\{c_i\}_{i=1}^d$ be a complete system of orthogonal idempotents.
    Then, The space $\cV$ is decomposed in the following direct sum.
    \begin{align}
        \cV=\oplus_{i=1}^d \cV(i,1)\oplus_{i<j}^d \cV(i,\frac{1}{2})\cap \cV(j,\frac{1}{2}).
    \end{align}
    Here, $\cV(i,1)$, $\cV(i,\frac{1}{2})$ are eigenspaces of eigenvalues $1$, $\frac{1}{2}$ of $c_i$ respectively.

    In addition, let $\{e_i\}_{i=1}^n$ be a complete system of orthogonal primitive idempotents.
    Then, $\cV$ is decomposed as
    \begin{align}
        \cV=\oplus_{i=1}^n \cV(i,1)\oplus_{i<j} \cV(i,\frac{1}{2})\cap \cV(j,\frac{1}{2}).
    \end{align}
    Here, $\cV(i,1)=\mathbb{R}e_i$.
\end{theorem}

\begin{theorem}[simultaneous spectral decomposition\cite{Gowda}{[Theorem3.1]}]\label{theorem:simultaneousspectral}
    For two elements $x,y\in \cV$, the following two conditions are equivalent.
    \begin{itemize}
        \item[(1)] The linear maps of $x,y$  defined by $\ref{definition:linearmap}$ are commute. i.e. the relation $L_xL_y=L_yL_x$ holds.
        \item[(2)] Two elements $x,y$ have a simultaneous spectral decomposition. i.e. for the spectral decomposition of $x$ as $\sum_i \lambda_i c_i$, there exists the spectral decomposition of $y$ as $y=\sum_i \mu_i d_i$ such that $d_i\in \oplus_j \cV(c_j,1)$.
    \end{itemize}
\end{theorem}

By Theorem~\ref{theorem:simultaneousspectral}, we define the concept said to behave classically as follows.
\begin{definition}[Classically]\label{definition:classically}
    The elements $x,y\in \cV$ are said to behave classically
    if the relation $L_xL_y=L_yL_x$ holds.
\end{definition}

\begin{lemma}\label{lemma:projectionPeirce}
    Let $\{c_i\}$ be a CSOI in $\cV$.
    Also, $x\in \cV$ has a Peirce decomposition with $\{c_i\}$ $x=\sum_i x_i+\sum_{i<j}x_{i,j}$.
    Then, the quadratic form $P_{c_i}$ maps $x$ to $x_i\in \cV(c_i,1)$.
\end{lemma}

Next, we introduce some lemmas for the further discussion.

\begin{lemma}[\cite{Faraut}4.1.1]\label{lemma:orthogonalPeirce}
    Let $\{c_i\}$ be a CSOI in $\cV$.
    Then, the relation $\cV(c_i,1)\circ \cV(c_j,1)=\{0\}$ holds for $i\neq j$, 
    where $\cV_1\circ \cV_2:=\{x\circ y\mid x\in \cV_1,y\in \cV_2\}$ for Jordan algebras $\cV_1,\cV_2$. 
\end{lemma}

\begin{lemma}\label{lemma:newCSOI}
    Let $\{c_i\}$ and $x$ be a CSOI and an element in $\cV$, respectively.
    Let $P_{c_i}x=\sum_j \lambda_{i,j}c_{i,j}$ be a spectral decomposition.
    Then, $\sum_jc_{i,j}=c_i$ holds.
\end{lemma}

Next, we define the trace as follows by using an inner product of $\cV$.
\begin{definition}[Trace\cite{Faraut}{[Chapter3-1]}]\label{definition:trace}
    We define a trace of $x\in \cV$ as
    \begin{align}
        \mathrm{tr}x:=\langle x,u\rangle.
    \end{align}
\end{definition}

However, in order to ensure that the trace $\mathrm{tr}$ is the generalization of matrix trace $\Tr$,
we need to normalize the trace and the inner product.
From the definition of quadratic form (Definition~\ref{definition:quadraticform}), we obtain following lemma.

\begin{lemma}[\cite{Faraut}{Proposition4.2.4(ii)}]\label{lemma:transitivity}
    Let $\cV$ and $x,y$ be a simple EJAs and primitive idempotents.
    Then, there exists the element $w$ satisfying $P_w(x)=y$ and $w^2=u$.
\end{lemma}

By applying Lemma~\ref{lemma:transitivity}, we obtain $\langle x,u\rangle=\langle y,u\rangle$ for a primitive idempotent $x,y$ on a simple EJAs $\cV$ by following way:
\begin{align}
    \langle u,x\rangle&=\langle u,P(w)y\rangle\\
    &\stackrel{(a)}{=}\langle P(w)u,y\rangle=\langle u,y\rangle,
\end{align}
where $P(w)$ maps $x$ to $y$.
The equality (a) is shown by Euclidean condition.

We normalize a norm $\mathrm{tr}x\circ y:=\langle x,y\rangle$ on an EJAs $\cV$ by following way:
Firstly, when an EJAs $\cV=\oplus_{i=1}^n \cV_i$ is decomposed to simple EJAs $\cV_i$, we set a norm $\langle x,y\rangle=\langle x_1,y_1\rangle_1+\cdots +\langle x_n,y_n\rangle_n$, 
where $\langle x_i,y_i\rangle_i$ is the inner product of element $x_i,y_i\in \cV_i$.
In these settings, We set a new inner product in $\cV$ as $\langle x,u\rangle=1$ for all simple EJAs.
Next, applying this normalization to an EJA decomposed by $\cV=\oplus_{i=1}^n\cV_i$,
We obtain $V=a_1\langle x_1,u\rangle_1+\cdots +a_n\langle x_n,u\rangle_n$, where $a_i(i=1,\ldots,n)$ are
constant in order to normalize to 1 for each elements.
Here in after, we only consider an EJA with the above normalized inner product.

Next, we introduce a composite system of an Euclidean Jordan algebra.
For general models of GPTs, we can not canonically define unique composite model of given models.
In contrast, we give a canonical definition of composite model for two models associated with two EJAs.

\begin{definition}[Composite system in an Euclidean Jordan algebra\cite{Barnum3}]\label{definition:CompositesysteminEJA}
    Let $\cV_1,\cV_2$ be Euclidean Jordan algebras.
    Let $\cV=\cV_1\otimes \cV_2$ be the tensor product vector space.
    Let the Jordan products in $\cV_1,\cV_2$ be $\circ_1,\circ_2$ ,respectively.
    We define the Jordan product $\circ$ of $\cV$ as $a_1\otimes b_1\circ a_2\otimes b_2=(a_1\circ_1 a_2)\otimes (b_1\circ_2 b_2)$.
    Moreover, we define the inner product $\langle\cdot,\cdot\rangle$ of $\cV$
    as $\langle a_1\otimes b_1,a_2\otimes b_2\rangle=\langle a_1,b_1\rangle_1\langle a_2,b_2\rangle_2$,
    where $\langle\cdot,\cdot\rangle_1,\langle\cdot,\cdot\rangle_2$ are inner products of $\cV_1,\cV_2$, respectively.
    Then $\cV$ become an Euclidean Jordan algebra.
    Here, by Definition~\ref{definition:positiveconeofEJA},
    we give the canonical positive cone $\cQ_\cV$ and we call $(\cV,\cQ_\cV,u_{1,2})$
    as the composite system of an Euclidean Jordan algebra,
    where $u_{1,2}=u_1\otimes u_2$ for the unit elements of $\cV_1,\cV_2$.
\end{definition}

\begin{lemma}
    The space $\cV$ defined by Definition~\ref{definition:CompositesysteminEJA} is an Euclidean Jordan algebra.
\end{lemma}
\begin{proof}
    For $x=a_1\otimes b_1,y=a_2\otimes b_2$, the relations $x\circ y=y\circ x$ and $x^2\circ(x\circ y)=x\circ (x^2\circ y)$
    are shown by the definition of the Jordan algebra $V_1,V_2$.
    The Euclidean condition is from the Euclidean conditions of $V_1,V_2$, that is,
    \begin{align}
         & \langle x\circ y,z\rangle=\langle a_1\circ_1 b_1\otimes a_2\circ_2 b_2, c_1\otimes c_2\rangle =\langle a_1\circ_1 b_1,c_1\rangle_1\langle a_2\circ_2 b_2,c_2\rangle_2 \\
         & =\langle a_1,b_1\circ_1 c_1\rangle_1\langle a_2,b_2\circ_2 c_2\rangle_2=\langle a_1\otimes a_2,b_1\circ_1 c_1\otimes b_2\circ_2 c_2\rangle =\langle x,y\circ z\rangle
    \end{align}
    ,where $z=c_1\otimes c_2$.
\end{proof}

\section{Information theorical tools}\label{section:Information theorical tools}

In this section, we define the information quantities in an EJA and investigate their properties.
In addition, we introduce some useful lemmas for latter discussions.
Next, we introduce an information theoretical tool, pinching, and show so-called pinching inequality and a lemma which states corresponding to measurement with the pinching states.
We apply them in order to show the inequalities of the information quantities such as Petz Relative R\'{e}nyi (PRR) entropy and Sandwiched Relative R\'{e}nyi (SRR) entropy.
Finally, we define TPCP map over an EJA and we check some examples and its properties.
From now on, we consider over an EJA $\cV$ with its canonical positive cone $\cQ$ unless explicitly stated.

\subsection{Information quantities in Euclidean Jordan algebra}\label{subsection:InformationQuantities}

At first, we introduce $\log \rho$ or $\exp \rho$ for the state $\rho$ in $\cV$.
\begin{definition}\label{definition:fofEJAelement}
    If the state $\rho$ has spectral decomposition as $\rho=\sum_i\lambda_i c_i$, we define $f(\rho)$ by the function $f:\mathbb{R}\to\mathbb{R}$ as
    \begin{align}
        f(\rho):=\sum_i f(\lambda_i)c_i.
    \end{align}
    Here, all of $\lambda_i$ are in the domain of definition of the function $f$.
\end{definition}

\begin{definition}
    If the state $\rho$ has a spectral decomposition as $\rho=\sum_i\lambda_ic_i$, we define $\rho^s\quad s\in(-\infty,0),(0,\infty)$ and $\log \rho$ as
    \begin{align}
        \rho^s    & =\sum_i\lambda^s c_i\quad s\in(-\infty,0),(0,\infty). \\
        \log \rho & =\sum_i \log \lambda_i c_i.
    \end{align}
    Here, all of $\lambda_i$ are in the domain of definition of the function $x^s,\log x$.
\end{definition}
These are an extension of a quantum state $\rho^s,\log \rho$.
By these Definitions, we extend the quantum information quantities to that of an EJA as follows.

\begin{definition}[Information Quantities in Euclidean Jordan algebra]\label{definition:informationquantities}
    For the states $\rho,\sigma$, we define the information quantities as
    \begin{itemize}
        \item[(1)]von Neumann entropy: $H(\rho):=-\mathrm{tr}\rho\circ\log\rho$.
        \item[(2)]Relative entropy: $D(\rho||\sigma):=\mathrm{tr}\left(\rho\circ\log\rho-\rho\circ\log\sigma\right)$.
        \item[(3)]Petz Relative R\'{e}nyi (PRR) entropy: $D_{1+s}(\rho||\sigma):=\frac{\phi(-s|\rho||\sigma)}{s}=\frac{1}{s}\log \mathrm{tr}\rho^{1+s}\circ \sigma^{-s}$.
        \item[(4)]Sandwiched Relative R\'{e}nyi (SRR) entropy: $\underline{D}_{1+s}(\rho||\sigma):=\frac{\tilde{\phi}(-s|\rho||\sigma)}{s}=\frac{1}{s}\log\mathrm{tr}\left(P_{\sigma^{\frac{-s}{2(1+s)}}}(\rho)\right)^{1+s}$.
    \end{itemize}
\end{definition}

Now, we give some statements of information quantities for the latter discussions.
All of them are known in Quantum system as the same way.
In other words, we generalize such statements to the case of EJAs.
We prove them in Appendix~\ref{appendix:fundamentalproperties},
and the structure of proofs is based on \cite{Hayashi1}{[Chapter3.1]}.

\begin{lemma}\label{lemma:classicallyRenyies}
    If the states $\rho,\sigma$ are classically (Definition~\ref{definition:classically}), PRR entropy is corresponding to SRR entropy,
	that is,
    \begin{align}
        D_{1+s}(\rho||\sigma)=\underline{D}_{1+s}(\rho||\sigma).
    \end{align}
\end{lemma}

\begin{lemma}[Additivity]\label{theorem:additivity}
    For the states $\rho_1,\rho_2,\sigma_1,\sigma_2$, the following relations hold.
    \begin{align}
        D(\rho_1\otimes\rho_2 ||\sigma_1\otimes \sigma_2)                   & =D(\rho_1||\sigma_1)+D(\rho_2||\sigma_2).                                                                        \\
        D_{1+s}(\rho_1\otimes \rho_2||\sigma_1\otimes \sigma_2)             & =D_{1+s}(\rho_1||\sigma_1)+D_{1+s}(\rho_2||\sigma_2).\label{equation:renyiadditivity}                            \\
        \underline{D}_{1+s}(\rho_1\otimes \rho_2||\sigma_1\otimes \sigma_2) & =\underline{D}_{1+s}(\rho_1||\sigma_1)+\underline{D}_{1+s}(\rho_2||\sigma_2).\label{equation:sandwichadditivity}
    \end{align}
\end{lemma}

\begin{lemma}\label{lemma:limitofRenyies}
    For the states $\rho,\sigma$, PRR entropy and SRR entropy holds following relations.
    \begin{align}
        \lim_{s\to 0}D_{1+s}(\rho||\sigma)             & =\lim_{s\to 0}\frac{\phi(-s|\rho||\sigma)}{s}=D(\rho||\sigma).         \\
        \lim_{s\to 0}\underline{D}_{1+s}(\rho||\sigma) & =\lim_{s\to 0}\frac{\tilde{\phi}(-s|\rho||\sigma)}{s}=D(\rho||\sigma).
    \end{align}
\end{lemma}

\begin{lemma}\label{lemma:monotonicityins}
    Let $\rho,\sigma$ be states in $\cV$.
    Then, the functions $s\to \underline{D}_{1+s}(\rho||\sigma)$ and $s\to D_{1+s}(\rho||\sigma)$ are monotone increasing.
\end{lemma}

\begin{lemma}[Jensen's inequality in Euclidean Jordan algebra]\label{lemma:Jenseninequality}
    Let $\rho$ be a state in $\cV$, f be a convex function. Then, the following inequality holds for $x\in \cV$.
    \begin{align}
        \mathrm{tr}\rho\circ f(x)\ge f(\mathrm{tr}\rho\circ x).
    \end{align}
\end{lemma}

\begin{lemma}\label{lemma:numberofeigenvalues}
    Let $x=\sum_{i=1}^d \lambda_i c_i$ be a spectral decomposition of $x\in \cV$.Then, $x^{\otimes n}\in \cV^{\otimes n}$ has at most $(n+1)^{d-1}$ all distinct eigenvalues.Then, $|\bm{C}_{x^{\otimes n}}|\le (n+1)^{d-1}$ holds.
\end{lemma}

\subsection{Pinching map and pinching inequality}\label{subsection:pinching}
In this part, we define an important information theoretical tool, pinching map.
Moreover, there are two important lemmas related to pinching map,
Lemma~\ref{lemma:RepresentEntropies} and Lemma~\ref{lemma:PinchingInequality}.
We use both of two statements in order to evaluate the information quantities in Section~\ref{section:Informationquantities}.

Now, we define the two kind of the pinching maps.
At first, the pinching of a state by CSOI is defined as follows.
\begin{definition}[Pinching by CSOI]\label{definition:PinchingCSOI}
    Let $\rho$ be a state over $\cV$.
    Also, let $\bm{C}=\{c_i\}$ be a CSOI in $\cV$.
    Then, we define the pinching of the state $\rho$ by CSOI $\bm{C}$ as follows.
    \begin{align}
        \kappa_{\bm{C}}(\rho):=\sum_iP_{c_i}(\rho).
    \end{align}
\end{definition}

Next, we prepare the pinching of a state by a state as follows.
\begin{definition}[Pinching by State]\label{definition:PinchingState}
    Let $\rho,\sigma$ be states in $\cV$.
    Also, we decompose $\sigma$ to $\sigma=\sum_i\mu_ic_i$ by Spectral theorem (Theorem~\ref{theorem:Spectraldecomposition}).
    Then, we define the pinching of the state $\rho$ by the state $\sigma$ as follows.
    \begin{align}
        \kappa_\sigma(\rho):=\sum_iP_{c_i}(\rho).
    \end{align}
\end{definition}

We check the elementary properties of pinchings as following lemmas.
\begin{lemma}[Pinching of State is State]\label{lemma:PinchingStateisState}
    Let $\rho$ be a state in $\cV$.
    Also, let $\bm{C}=\{c_i\}$ be CSOI in $\cV$.
    Then, the pinching of the state $\rho$ by CSOI $C$ is also a state.
\end{lemma}

\begin{proof}
    Now, we check the condition of a state (Definition~\ref{definition:StateEffectMeasurement}).
    At first, for a CSOI $\bm{C}=\{c_i\}$ and a state $\rho$, we obtain $P_{c_i}(\rho)\ge0$ for all $c_i\in \bm{C}$ by Lemma~\ref{lemma:positivityofQradratic}.
    Therefore, $\kappa_{\bm{C}}(\rho)=\sum_i P_{c_i}(\rho)\ge 0$ from the property of a convex cone
    of the positive cone $\cQ_\cV \subset \cV$.

    Next, we check the condition of the normalization as follows.
    \begin{align}
        \langle\kappa_{\bm{C}}(\rho),u\rangle=\langle\sum_iP_{c_i}(\rho),u\rangle=\sum_i\langle P_{c_i}(\rho),u\rangle\stackrel{(a)}{=}\sum_i\langle\rho,P_{c_i}(u)\rangle=\sum_i\langle\rho,c_i\rangle=\langle\rho,\sum_i c_i\rangle=\langle\rho,u\rangle=1.
    \end{align}
    The equality (a) is given by the Euclidean condition (J3) of Definition~\ref{definition:EuclideanJordanalgebra}.
    The conditions of the positivity $\kappa_{\bm{C}}(\rho)\ge 0$ and the normalization $\langle \kappa_{\bm{C}}(\rho),u\rangle=1$ imply that $\kappa_{\bm{C}}(\rho)$ is a state.
\end{proof}

\begin{lemma}\label{lemma:classicallyofpinching}
    Let $\rho,\sigma$ be states over $\cV$.
    Then, the pinching of $\rho$ by $\sigma$ and $\sigma$ are classically(Definition~\ref{definition:classically}).
\end{lemma}
\begin{proof}
    Now, we show
    \begin{align}
        L_{\kappa_\sigma(\rho)}L_{\sigma}=L_{\sigma}L_{\kappa_\sigma(\rho)}.
    \end{align}
    Let the spectral decomposition of $\sigma$ be $\sigma=\sum_i \mu_i c_i$.
    Also, let the spectral decomposition (Lemma~\ref{lemma:newCSOI}) of $\kappa_\sigma(\rho)$ be
    $\kappa_\sigma(\rho)=\sum_iP_{c_i}\rho=\sum_{i,j}\lambda_{i,j}c_{i,j}$.
    Here, the relation $c_{i,j}\in \cV(c_i,1)$ holds.
    Then, by using this relation of $c_{i,j}$,
    we show $L_{c_{i,j}}L_{c_k}=L_{c_k}L_{c_{i,j}}$ because of the linearity of $L$.

    First, we consider Pierce decomposition of $z\in \cV$ by the CSOI $\{c_i\}$ as follows.
    \begin{align}
        z=\sum_iz_i+\sum_{i<j}z_{i,j}.\label{equation:classicallyofpinching}
    \end{align}
    Here, $z_i$ belongs to $\cV(c_i,1)$ and $z_{i,j}$ belongs to $\cV(c_i,\frac{1}{2})\cap \cV(c_j,\frac{1}{2})$.
    Next, we apply \eqref{equation:classicallyofpinching} to $L_{c_{i,j}}L_{c_k}$.
    If $k<i$ holds, we obtain
    \begin{align}
        L_{c_{i,j}}L_{c_k}(z)=L_{c_{i,j}}\left(z_k+\frac{1}{2}\sum_{l:k<l}z_{k,l}+\frac{1}{2}\sum_{l:l<k}z_{l,k}\right)=\frac{1}{4}z_{k,i}.\label{equation:classicallyofpinching2}
    \end{align}
    Then, we calculate the following relations.
    \begin{align}
        L(c_k)L(c_{i,j})z=L(c_k)\left(z_i+\frac{1}{2}\sum_{l:i<l}z_{i,l}+\frac{1}{2}\sum_{l:l<i}z_{l,i}\right)=\frac{1}{4}z_{k,i}.\label{equation:classicallyofpinching3}
    \end{align}
    For all $z\in \cV$, the equations \eqref{equation:classicallyofpinching2},\eqref{equation:classicallyofpinching3} hold.
    On the other hand, if $k>i$ holds, we obtain $L_{c_{i,j}}L_{c_k}(z)=L_{c_k}L_{c_{i,j}}(z)$ similarly to \eqref{equation:classicallyofpinching2} and \eqref{equation:classicallyofpinching3}.
    Moreover, if $k=i$ holds, we obtain
    \begin{align}
        L_{c_{i,j}}L_{c_i}(z)=L_{c_{i,j}}\left(z_i+\frac{1}{2}\sum_{l:i<l}z_{i,l}+\frac{1}{2}\sum_{l:i>l}z_{l,i}\right)=z_i.
    \end{align}
    \begin{align}
        L_{c_i}L_{c_{i,j}}z=L_{c_i}\left(z_i+\frac{1}{2}\sum_{l:i<l}z_{i,l}+\frac{1}{2}\sum_{l:l<i}z_{l,i}\right)=z_i.
    \end{align}
    Combining the case of $k>i$, $k<i$ and $k=i$, we obtain $L_{c_{i,j}}L_{c_k}=L_{c_k}L_{c_{i,j}}$ for all $i,j,k$.
    Therefore, we obtain the conclusion.
\end{proof}

First, we define Pinched Measurement, which plays an important role 
in the proof of the main results.
\begin{definition}\label{definition:pinchedmeasurement}
    Let $\rho,\sigma$ be states in $\cV$.
    Also, let $\bm{M}=\{M_k\}_k$ be a measurement in $\cV$.
    Then, we define the following family:
    \begin{align}
        \bm{M}^\rho_\sigma:=\{P_{c_{i,j}}(M_k)\}_{i,j,k}.
    \end{align}
    Here, $\{c_{i,j}\}$ is given the spectral decomposition (Lemma~\ref{lemma:newCSOI}) $\kappa_\sigma(\rho)=\sum_{i,j} \lambda'_{i,j}c_{i,j}$ i.e.,
    Especially, in the case of the obvious measurement $\bm{M}=\{u\}$, 
    we denote $I^\rho_\sigma:=\bm{M}^\rho_\sigma=\{c_{i,j}\}_{i,j}$.
\end{definition}

\begin{lemma}
    The family defined in Definition~\ref{definition:pinchedmeasurement}
    is a measurement.
\end{lemma}

\begin{proof}
    Let the spectral decomposition of $\sigma$ be $\sigma=\sum_i \mu_i e_i$.
    Also, by Lemma~\ref{lemma:newCSOI}, let the spectral decomposition of $P_{e_i}(\rho)$ be $P_{e_i}(\rho)=\sum_{j}\lambda_{i,j}c_{i,j}$.
    From Pierce decomposition, the state $\kappa_\sigma(\rho)$ belongs to $\oplus_i \cV(e_i,1)$ and the element $\{c_{i,j}\}_{i}$ belongs to the space $\cV(e_i,1)$.
    At first, we show that $\bm{M}^\rho_\sigma=\{P_{c_{i,j}}(M_k)\}_{i,j,k}$ is a measurement. 
\end{proof}

Then, the first main lemma gives the relation between entropies with pinching.

\begin{lemma}[Represent Entropies with pinching state by Classical Entropies with Measurement]\label{lemma:RepresentEntropies}
    Let $\rho,\sigma$ be states in $\cV$.
    Also, let $M=\{M_k\}_k$ be a measurement in $\cV$.
    Then, the following relations hold.
    \begin{align}
        D_{1+s}(\kappa_\sigma(\rho)||\sigma) & =D_{1+s}(P_\rho^{M^\rho_\sigma}||P_\sigma^{M^\rho_\sigma})\quad(s\neq0). \label{equation:pinchingRelativerenyi1} \\
        D(\kappa_\sigma(\rho)||\sigma)       & =D(P_\rho^{M^\rho_\sigma}||P_\sigma^{M^\rho_\sigma}).\label{equation:pinchingRelative}
    \end{align}
    Moreover, $\kappa_\sigma(\rho)$ and $\sigma$ are classically by Lemma~\ref{lemma:classicallyofpinching}.
    Therefore, the following relation holds from Lemma~\ref{lemma:classicallyRenyies} and \eqref{equation:pinchingRelativerenyi1}.
    \begin{align}
        \underline{D}_{1+s}(\kappa_\sigma(\rho)||\sigma)=D_{1+s}(\kappa_\sigma(\rho)||\sigma)=D_{1+s}(P_\rho^{M^\rho_\sigma}||P_\sigma^{M^\rho_\sigma})\label{equation:pinchingsandwiched}.
    \end{align}
\end{lemma}

\begin{proof}
    The spectral decomposition of $\sigma$ is given as $\sigma=\sum_i \mu_i e_i$.
    Also, the spectral decomposition of $\kappa_\sigma(\rho)$ is given by $\{c_{i,j}\} $similarly to Lemma~\ref{lemma:newCSOI}:
    Then, we have the following relations:
    \begin{align}
        P_{e_i}(\rho)&=\sum_j\lambda_{i,j}c_{i,j}.\label{equation:pinchingRelativerenyispectral2}\\
        \kappa_\sigma(\rho)&=\sum_{i}P_{e_i}(\rho)=\sum_{i,j}\lambda_{i,j}c_{i,j}.\label{equation:pinchingRelativerenyispectral}\\
        \sum_{j}c_{i,j}&=e_i.\label{equation:pinchingRelativerenyispectral1}
    \end{align}
    Here we remark that the equation \eqref{equation:pinchingRelativerenyispectral} is a spectral decomposition of $\kappa_\sigma(\rho)$.

    First we show \eqref{equation:pinchingRelativerenyi1}.
    In order to show this equality, we calculate LHS of \eqref{equation:pinchingRelativerenyi1} as follows:
    \begin{align}
        D_{1+s}(\kappa_\sigma(\rho)||\sigma) & =\frac{1}{s}\log\mathrm{tr}\kappa_\sigma(\rho)^{1+s}\circ\sigma^{-s}                  \\
                                             & \stackrel{(a)}{=}\frac{1}{s}\log\mathrm{tr}\sum_{i,j}\lambda_{i,j}^{1+s}\mu_i^{-s}c_{i,j}             \\
                                             & =\frac{1}{s}\log \mathrm{tr}\sum_{i,j,k}\lambda_{i,j}^{1+s}\mu_i^{-s}c_{i,j}\circ M_k.\label{equation:pinchingRelativerenyi}
                                            \end{align}
The equality (a) is given by applying 
the relation \eqref{equation:pinchingRelativerenyispectral1} to $\sigma$ and orthogonality of $\{c_{i,j}\}$.

                                            Next, we will show the following relations:
    \begin{align}
        \mathrm{tr}c_{i,j}\circ M_k          & =\mathrm{tr}P_{c_{i,j}}(M_k).\label{equation:pinchingMeasurement0}              \\
        \mathrm{tr}\rho\circ P_{c_{i,j}}(M_k)   & =\lambda_{i,j}\mathrm{tr}P_{c_{i,j}}(M_k).\label{equation:pinchingMeasurement1}\\
        \mathrm{tr}\sigma\circ P_{c_{i,j}}(M_k) & =\mu_i\mathrm{tr}P_{c_{i,j}}(M_k).\label{equation:pinchingMeasurement2}
    \end{align}
    The equation \eqref{equation:pinchingMeasurement0} is shown by the relation $\langle P_{c_{i,j}}(M_k),u\rangle=\langle M_k,P_{c_{i,j}}(u)\rangle$,
    which is derived from Euclidean condition (J3) of Definition~\ref{definition:EuclideanJordanalgebra}.
    The equation \eqref{equation:pinchingMeasurement1} is shown as follows: 
    \begin{align}
        \mathrm{tr}\rho\circ P_{c_{i,j}}(M_k)&=\langle\rho\circ P_{c_{i,j}}(M_k),u\rangle\\
        &=\langle P_{c_{i,j}}(\rho),M_k\rangle\\
        &\stackrel{(a)}{=}\langle P_{c_{i,j}}P_{e_i}(\rho),M_k \rangle\\
        &\stackrel{(b)}{=}\langle \lambda_{i,j}c_{i,j},M_k\rangle\\
        &=\lambda_{i,j}\mathrm{tr}P_{c_{i,j}}(M_k).
    \end{align}
    The equality (a) is shown by the relation $\cV(c_{i,j},1)\subset \cV(e_i,1)$, which is derived from \eqref{equation:pinchingRelativerenyispectral1}. 
    The equality (b) is shown by the relation \eqref{equation:pinchingRelativerenyispectral2}.
    The equation \eqref{equation:pinchingMeasurement2} is shown by \eqref{equation:pinchingRelativerenyispectral1}, 
    similarly to \eqref{equation:pinchingMeasurement1}.
    
    

    Combining the equation \eqref{equation:pinchingMeasurement1} and \eqref{equation:pinchingMeasurement2},
    we organize the relation \eqref{equation:pinchingRelativerenyi} as follows:
    \begin{align}
        \frac{1}{s}\log\mathrm{tr}\sum_{i,j,k}\lambda_{i,j}^{1+s}\mu_i^{-s}c_{i,j}\circ M_k & =\frac{1}{s}\log\left(\mathrm{tr}\rho\circ P_{c_{i,j}}(M_k)\right)^{1+s}\left(\mathrm{tr}\sigma\circ P_{c_{i,j}}(M_k)\right)^{-s} \\
                                                                                            & =D_{1+s}(P^{M^\rho_\sigma}_\rho||P^{M^\rho_\sigma}_\sigma).\label{equation:MeasurementRelativerenyi}
    \end{align}
    As a result, we obtain \eqref{equation:pinchingRelativerenyi1}.

    Besides, the equation \eqref{equation:pinchingRelative} is given by the fact that the parameter $s$ of $D_{1+s}$ in \eqref{equation:pinchingRelativerenyi1}
    does not depend on the choice of a measurement $\bm{M}$.
    Therefore, we apply Lemma~\ref{lemma:limitofRenyies} to \eqref{equation:pinchingRelativerenyi1}.
    Then, we obtain the equality \eqref{equation:pinchingRelative}.
    
    Finally, the equation \eqref{equation:pinchingsandwiched} is given by Lemma~\ref{lemma:classicallyofpinching} and Lemma~\ref{lemma:classicallyRenyies} as follows:
    \begin{align}
     \underline{D}_{1+s}(\kappa_\sigma(\rho)||\sigma) & =D_{1+s}(\kappa_\sigma(\rho)||\sigma)=D_{1+s}(P^{\bm{M}^\rho_\sigma}_\rho||P^{\bm{M}^\rho_\sigma}_\sigma) .
    \end{align}
\end{proof}

Finally, we give the following lemma, which is called pinching inequality in Quantum system.
\begin{lemma}[Pinching inequality]\label{lemma:PinchingInequality}
    Let $\bm{C}=\{c_i\}$ be COSI in $\cV$.
    Also, let $\rho$ be a state in $\cV$.
    Then, the following relation holds:
    \begin{align}
        |\bm{C}|\kappa_{\bm{C}}(\rho)\ge\rho.
    \end{align}
\end{lemma}
\begin{proof}
   Denote $|\bm{C}|=n$, and we obtain the conclusion as follows:
    \begin{align}
          & P_{c_1+\cdots +c_n}(\rho)+\sum_{1\le i<j\le n}P_{c_i-c_j}(\rho)                                                                                                                           \\
        \stackrel{(a)}{=} & 2\left(L_{c_1+\cdots +c_n}^2+\sum_{1\le i<j\le n}L_{c_i-c_j}^2\right)(\rho)-\left(L_{(c_1+\cdots+c_n)^2}+\sum_{1\le i<j\le n}L_{(c_i-c_j)^2}\right)(\rho)                                                         \\
        \stackrel{(b)}{=}& 2(\sum_{i=1}^n n L_{c_i^2})(\rho)-(\sum_{i=1}^n nL_{c_i})(\rho)                                                                    =n\sum_{i=1}^n(2L_{c_i}^2-L_{c_i})(\rho)=|\bm{C}|\kappa_{\bm{C}}(\rho).
    \end{align}
    The equality (a) is given by the definition of a quadratic form and organization of the equation.
    The equality (b) is implied as follows.
    The first term is reduced by the linearity of $L$ and simple calculation.
    The second term is reduced by orthogonality and idempotency of $\{c_i\}$.
    On the other hand,
    \begin{align}
        P_{c_1+\cdots+c_n}(\rho) & =P_u(\rho)=\rho\ge 0.              \label{equation:conditionofidentity3} \\
        P_{c_i-c_j}(\rho)        & \ge 0\quad(i\neq j,i,j=1,\ldots,n).\label{equation:conditionofpositive3}
    \end{align}
    Combining \eqref{equation:conditionofidentity3} and \eqref{equation:conditionofpositive3}, we obtain
    \begin{align}
        |\bm{C}|\kappa_{\bm{C}}(\rho)\ge \rho.
    \end{align}
\end{proof}
\begin{remark}
Definition~\ref{definition:PinchingCSOI} and Definition~\ref{definition:PinchingState} are generalizations
of standard definitions in quantum theory with PVM\cite{Hayashi1}[Chapter3.8].
Also, Lemma~\ref{lemma:RepresentEntropies} and Lemma~\ref{lemma:PinchingInequality}
are the corresponding important properties by the generalization.
However, due to the structure of EJA, 
we need to define $M^\rho_\sigma$ for Lemma~\ref{lemma:RepresentEntropies}.
Also, We need to prove Lemma~\ref{lemma:PinchingInequality} 
by an indirect generalization of the proof in quantum
theory with the properties of quadratic form 
as an analogy from \cite{Hayashi1}[Lemma3.10] and \cite{Hayashi3}[Chapter3 Lemma5].
\end{remark}

\subsection{TPCP map in Euclidean Jordan algebra}\label{subsection:tpcp}
In this part, we define the TPCP map in EJAs similarly to quantum theory.
Moreover, we check the properties of a TPCP map.
Finally, we prepare a concrete example of TPCP maps applied in Section~\ref{section:Informationquantities}.
Only in this part, we denote $\cV_1,\cV_2,\cV'$ as EJAs.
In addition, we denote $\cQ_1,\cQ_2$ as positive cones associated with $\cV_1,\cV_2$, respectively.

At first we define the TPCP map as follows.

\begin{definition}[Trace Preserving]\label{definition:tracepreserving}
    We call the linear map $\kappa:\cV_1\to \cV_2$ a Trace Preserving (TP) map 
    if the map $\kappa$ satisfies $\mathrm{tr}x=\mathrm{tr}\kappa(x)$ 
    for any element $x\in \cV_1$.
\end{definition}

\begin{definition}[Positive map]
    We call the linear map $\kappa:\cV_1\to \cV_2$ a Positive map 
    if the map $\kappa$ satisfies $\kappa(x)\in \cQ_2$ for any $x\in \cQ_1$.
\end{definition}

\begin{definition}[Completely Positivity]
    We call the linear map $\kappa:\cV_1\to \cV_2$ a Completely Positive (CP) map
    if the map $\kappa$ satisfies the following condition:
    For any space $V'$, the map $\kappa\otimes \iota:\cV_1\otimes \cV'\to \cV_2\otimes \cV'$ 
    is a positive map, where $\iota:\cV'\to \cV'$ is an identity map.
\end{definition}

\begin{definition}[TPCP map]\label{definition:tpcp}
    We call the linear map $\kappa:\cV_1\to \cV_2$ a TPCP map
    if the map $\kappa$ is trace preserving and complete positive.
\end{definition}


\begin{lemma}\label{lemma:TPCPtensor}
    Let $\kappa:\cV_1\to \cV_2$ be a TPCP map.
    Then, the map $\kappa^{\otimes n}:\cV_1^{\otimes n}\to \cV_2^{\otimes n}$ is a TPCP map.
\end{lemma}

\begin{proof}
From the Trace Preservity and Completely Positivity of $\kappa$, we obtain $\kappa\otimes \kappa=(\kappa\otimes \iota)(\iota\otimes \kappa):\cV_1\otimes \cV_1\to \cV_2\otimes \cV_2$
is the TPCP map.
In addition, the Trace Preservity and Completely Positivity of $\kappa\otimes \kappa$, $\kappa\otimes \kappa\otimes \kappa$ is the TPCP map.
Inductively, the map $\kappa^{\otimes n}$ is the TPCP map for an arbitrary number $n$.
\end{proof}
\color{black}

Now, we give the following Lemma in order to prove 
the monotonicity of SRR entropy by a TPCP map
in Section~\ref{subsection:sandwichedrelative}.

\begin{lemma}[Identity preservation of adjoint map]\label{lemma:identitypreserving}
    Let $\kappa:\cV\to \cV$ and $\kappa^*:\cV\to \cV$ be a TPCP map 
    and the adjoint map of $\kappa$, respectively.
    Then, the following relation holds:
    \begin{align}
        \kappa^*(u)=u.
    \end{align}
\end{lemma}

\begin{proof}
    From the definition of adjoint map, 
    the following relation holds for any $x\in \cV$.
    \begin{align}
        \langle \kappa(x),u\rangle=\langle x,\kappa^*(u)\rangle.
    \end{align}
    From the condition of trace preserving, we obtain
    \begin{align}
        1=\left\langle\frac{x}{\mathrm{tr}x},\kappa^*(u)\right\rangle.\label{equation:identitypreserving1}
    \end{align}
    Now, we consider the spectral decomposition of $\kappa^*(u)$ as $\kappa^*(u)=\sum_i \lambda_ic_i$.
    We substitute $c_i$ for $x$ in equation \eqref{equation:identitypreserving1}.
    Then, we obtain the following equation for any $i$:
    \begin{align}
        1=\left\langle\frac{x}{\mathrm{tr}x},\kappa^*(u)\right\rangle=\lambda_i\left\langle\frac{c_i}{\mathrm{tr}c_i},u\right\rangle=\lambda_i.
    \end{align}
    As a result, we obtain
    \begin{align}
        \kappa^*(u)=\sum_i c_i=u.
    \end{align}
\end{proof}

Finally, we investigate two concrete examples.
We will apply these two TPCP maps to the proof of the information processing inequality in Section~\ref{section:Informationquantities}.
\begin{definition}[Partial trace]\label{definition:partialtrace}
    We call the linear map $\mathrm{tr}_{\cV_1}:\cV_1\otimes \cV_2\to \cV_2$ a partial trace for $V_1$
    if the map $\mathrm{tr}_{\cV_1}$ satisfies the following condition:
    For the element $z=\sum_{i,j}\lambda_{i,j}x_i\otimes y_j \in \cV_1\otimes \cV_2$, the map $\mathrm{tr}_{\cV_1}$ satisfies
    \begin{align}
        \mathrm{tr}_{\cV_1}z=\sum_{i,j}\lambda_{i,j}(\mathrm{tr}x_i)y_j=\sum_{i,j}\lambda_{i,j}\langle x_i,u_1\rangle_1 y_j,
    \end{align}
    where $\langle\cdot,\cdot\rangle_1$ is an inner product in $\cV_1$ and $u_1$ is an unit effect in $\cV_1$.
\end{definition}
\begin{lemma}[Partial trace is TPCP map]\label{lemma:TPCPofPartialtrace}
    The partial trace $\mathrm{tr}_{\cV_1}:\cV_1\otimes \cV_2\to \cV_2$ is a TPCP map.
\end{lemma}
\begin{proof}
    At first, trace preservation of $\mathrm{tr}_{\cV_1}$ is shown as follows:
    \begin{align}
        \mathrm{tr}z & =\mathrm{tr}\sum_{i,j}\lambda_{i,j}x_i\otimes y_j=\sum_{i,j}\lambda_{i,j}(\mathrm{tr}_{\cV_1}x_i)(\mathrm{tr}y_j) \\
                     & =\mathrm{tr}\sum_{i,j}\lambda_{i,j}(\mathrm{tr}_{\cV_1}x_i)y_j=\mathrm{tr}(\mathrm{tr}_{\cV_1}z),
    \end{align}
   where $z=\sum_{i,j}\lambda_{i,j}x_i\otimes y_j\in \cV_1\otimes \cV_2$.

    Next, we will show the completely positivity of $\mathrm{tr}_{\cV_1}$.
    For any space $\cV'$, we consider the space $\cV_1\otimes \cV_2\otimes \cV'$.
    Now, we take the element $x$ of the positive cone associated with $\cV_1\otimes \cV_2\otimes \cV'$.
    Here, we consider the spectral decomposition $x=\sum_{i,j,k}\lambda_{i,j,k}c_i\otimes d_j\otimes e_k$,
    where $\{c_i\},\{d_k\},\{e_j\}$ are the COSI of $\cV_1,\cV_2,\cV'$, respectively.
    Moreover, the coefficiences satisfy $\lambda_{i,j,k}\ge 0$.
    We apply the map $\mathrm{tr}_{\cV_1}\otimes \iota:\cV_1\otimes \cV_2\otimes \cV'\to \cV_2\otimes \cV'$ to the element $x$, and we obtain
    \begin{align}
        \mathrm{tr}_{\cV_1}\otimes \iota(x) & =\mathrm{tr}_{\cV_1}\otimes \iota(\sum_{i,j,k}\lambda_{i,j,k}c_i\otimes d_j\otimes e_k)=\sum_{i,j,k}\lambda_{i,j,k}(\mathrm{tr}c_i)d_j\otimes c_k.
    \end{align}
    Here, the coefficiences $\lambda_{i,j,k}\ge 0$ and $\mathrm{tr}c_i\ge 0$ by idempotency of $c_i$.
    Hence, the element $\mathrm{tr}_{\cV_1}\otimes \iota (x)$ is the element of positive cone associated with $\cV_2\otimes \cV'$.
    Therefore, we complete the proof of a completely positivity of a partial trace.
\end{proof}

Next, we introduce the following new TPCP map.
We will apply this TPCP map in order to show that the observing is one of the TPCP map.
\begin{definition}[TPCP map of Observation]\label{definition:TPCPobservation}
    Let $\bm{M}=\{M_i\}_{i=1}^d$ be a measurement in $\cV$.
    Let $\mathbb{R}^d$ be a classical system (Example~\ref{example:classicalsystem}).
    Also, let $u_i\in\mathbb{R}^d(i=1,\ldots,d)$ be the element which takes $1$ in $i$th element and $0$ in others.
    Now, we define a linear map $\kappa_{\bm{M}}:\cV\to\mathbb{R}^d$ as
    \begin{align}
        \kappa_{\b{M}}(x):=\sum_{i=1}^d\mathrm{tr}(M_i\circ x)u_i,\forall x\in \cV.\label{equation:TPCPobservation}
    \end{align}
\end{definition}
\begin{lemma}\label{lemma:TPCPofObservation}
    The map Definition~\ref{definition:TPCPobservation} is a TPCP map.
\end{lemma}
\begin{proof}
    At first, we obtain the trace preservation of the map $\kappa_{\bm{M}}$ as follows:
    \begin{align}
        \mathrm{tr}\kappa_{\bm{M}}(x)=\mathrm{tr}\sum_{i=1}^d \mathrm{tr}(M_i\circ x)u_i=\sum_{i=1}^d\mathrm{tr}(M_i\circ x)=\mathrm{tr}x \quad\forall x\in \cV.
    \end{align}

    Next, we examine the completely positivity of $\kappa_{\bm{M}}$.
    For any space $\cV'$, we take an arbitrary element $x$ in the positive cone associated with $\cV\otimes \cV'$.
    Then, we consider the spectral decomposition $x=\sum_{j,k}\lambda_{j,k}c_j\otimes d_k$,
    where the coefficiences $\lambda_{j,k}\ge 0$ hold for all $j,k$ and $\{c_j\},\{d_k\}$ are COSI in $\cV,\cV'$, respectively.
    Now, we apply the map $\kappa_{\bm{M}}\otimes \iota$ to the element $x$, and we obtain
    \begin{align}
        \kappa_{\bm{M}}\otimes \iota(x)=\kappa_{\bm{M}}\otimes \iota (\sum_{j,k}\lambda_{j,k}c_j\otimes d_k)=\sum_{i=1}^d\sum_{j,k}\lambda_{j,k}\mathrm{tr}(M_i\circ c_j)u_i\otimes d_k,
    \end{align}
    where $\mathrm{tr}M_i\circ c_j=\langle M_i, P_{c_j}(u)\rangle=\langle P_{c_j}(M_i),u\rangle=\mathrm{tr}P_{c_j}(M_i)\ge 0$
    from Lemma~\ref{lemma:positivityofQradratic}.
    Moreover, $\{u_i\otimes d_k\}$ is a CSOI in $\mathbb{R}^d\otimes V'$.
    Therefore, we complete to prove the completely positivity of $\kappa_{\bm{M}}$.
\end{proof}

Here we remark that the above map \eqref{equation:TPCPobservation} corresponds to 
the observation for a state $x$ with $\bm{M}$.

\section{The relation of Information quantities}\label{section:Informationquantities}

In  this section, we investigate the three information quantities, PRR entropy,
SRR entropy and Relative entropy in order to prove Stein's lemma 
with EJAs in Section~\ref{section:hypothesistesting}.
At first, we examine a property of PRR entropy, 
monotonicity of an observation.
Secondly, we investigate the property of SRR entropy, 
monotonicity of a TPCP map.
Finally, conbining these monotonicities of PRR entropy and SRR entropy,
we investigate the property of Relative entropy with the monotonicity under a TPCP map and
show some theorems.

We note that all lemmas and theorems in Section~\ref{section:Informationquantities} 
are directly generalized from known results in quantum information theory.
Some statements are derived by the same way 
as that of quantum theory through the properties in Section~\ref{section:Information theorical tools}.
However, due to the lack of operator monotonicity in EJAs,
we need to prove other statements by indirect generalizations of the proofs in quantum information theory.

\subsection{Petz Relative R\'{e}nyi entropy}\label{subsection:relativerenyi}
In this part, we give a relation among $D_{1+s}(\rho||\sigma)$, $D_{1+s}(P^{\bm{M}}_\rho||P^{\bm{M}}_\sigma)$, and $\frac{1}{n}D_{1+s}(\kappa_{\sigma^{\otimes n}}(\rho^{\otimes n})||\sigma^{\otimes n})$
as the following theorem.
The convergency of $\frac{1}{n}D_{1+s}(\kappa_{\sigma^{\otimes n}}(\rho^{\otimes n})||\sigma^{\otimes n})$ is discussed in Appendix~\ref{appendix:proofofRelativeRenyi}.

\begin{theorem}[Monotonicity of PRR entropy by an observation]\label{theorem:MonotonicityofRelativerenyi}
    Let $\rho,\sigma$ be states in $\cV$.
    Also, let $\bm{M}=\{M_i\}$ be a measurement in $\cV$.
    Then, the following inequality holds:
    \begin{align}
        D_{1+s}(\rho||\sigma)\ge \lim_{n\to \infty}\frac{1}{n}D_{1+s}(\kappa_{\sigma^{\otimes n}}(\rho^{\otimes n})||\sigma^{\otimes n})\ge D_{1+s}(P^{\bm{M}}_\rho||P^{\bm{M}}_\sigma)\quad(s>0).
    \end{align}
\end{theorem}

This Theorem~\ref{theorem:MonotonicityofRelativerenyi} is proven by the following two Lemmas.
\begin{lemma}\label{lemma:Monotonicity of Relativerenyi1}
    Let $\rho,\sigma$ be states in $\cV$.
    Then, the following inequality holds:
    \begin{align}
        D_{1+s}(\rho||\sigma)\ge D_{1+s}(\kappa_\sigma(\rho)||\sigma),\quad s>0.
    \end{align}
\end{lemma}
\begin{lemma}\label{lemma:Monotonicity of Relativerenyi2}
    Let $\rho,\sigma$ be states in $\cV$.
    Also, let $\bm{M}=\{M_i\}$ be a measurement in $\cV$.
    Then, the following inequality holds:
    \begin{align}
        D_{1+s}(\kappa_\sigma(\rho)||\sigma)\ge D_{1+s}(P^{\bm{M}}_\rho||P^{\bm{M}}_\sigma)-\frac{1+s}{s}\log|\bm{C}_\sigma|\quad(s>0).
    \end{align}
\end{lemma}
These two Lemmas are proven in Appendix~\ref{appendix:proofofRelativeRenyi}.
Here we prove Theorem~\ref{theorem:MonotonicityofRelativerenyi} by assuming Lemma~\ref{lemma:Monotonicity of Relativerenyi1} and Lemma~\ref{lemma:Monotonicity of Relativerenyi2}.

\begin{proof}[proof of Theorem~\ref{theorem:MonotonicityofRelativerenyi}]
    Combining Lemma~\ref{lemma:Monotonicity of Relativerenyi1} and Lemma~\ref{lemma:Monotonicity of Relativerenyi2},
    we obtain
    \begin{align}
        D_{1+s}(\rho||\sigma)\ge D_{1+s}(\kappa_\sigma(\rho)||\sigma)\ge D_{1+s}(P^{\bm{M}}_\rho||P^{\bm{M}}_\sigma)-\frac{1+s}{s}\log|\bm{C}_\sigma|\quad (s>0).\label{equation:MonotonicityofRelativerenyi1}
    \end{align}
    Now, we apply the inequality \eqref{equation:MonotonicityofRelativerenyi1} to the states $\rho^{\otimes n},\sigma^{\otimes n}$ 
    and an arbitrary measurement $\bm{M}^n=\{M^n_i\}_i$ in $\cV^{\otimes n}$.
    Then, we obtain
    \begin{align}
        D_{1+s}(\rho^{\otimes n}||\sigma^{\otimes n})\ge D_{1+s}(\kappa_{\sigma^{\otimes n}}(\rho^{\otimes n})||\sigma^{\otimes n}) & \ge D_{1+s}(P^{\bm{M}^n}_{\rho^{\otimes n}}||P^{\bm{M}^n}_{\sigma^{\otimes n}})-\frac{1+s}{s}\log|\bm{C}_{\sigma^{\otimes n}}|     \\
                                                                                                                                    & \stackrel{(a)}{\ge} D_{1+s}(P^{\bm{M}^n}_{\rho^{\otimes n}}||P^{\bm{M}^n}_{\sigma^{\otimes n}})-\frac{1+s}{s}\log(n+1)^{d-1},\label{equation:MonotonicityofRelativerenyi_1}
    \end{align}
    where $d:=|\bm{C}_\sigma|$.
    The equation (a) is given by Lemma~\ref{lemma:numberofeigenvalues}.
    Therefore, the measurement $\{M_{i_1}\otimes \cdots \otimes M_{i_n}\}$ of $n$-shot use of $\bm{M}=\{M_i\}$ 
    satisfies \eqref{equation:MonotonicityofRelativerenyi_1} instead of $\bm{M}^n$.
    Then, we obtain
    \begin{align}
        nD_{1+s}(\rho||\sigma)\stackrel{(a)}{=}D_{1+s}(\rho^{\otimes n}||\sigma^{\otimes n})\ge D_{1+s}(\kappa_{\sigma^{\otimes n}}(\rho^{\otimes n})||\sigma^{\otimes n}) & \ge D_{1+s}(P^{\bm{M}^n}_{\rho^{\otimes n}}||P^{\bm{M}^n}_{\sigma^{\otimes n}})-\frac{1+s}{s}\log(n+1)^{d-1} \\
                                                                                                                                                                           & \stackrel{(b)}{\ge} nD_{1+s}(P^{\bm{M}}_\rho||P^{\bm{M}}_\sigma)-\frac{1+s}{s}\log(n+1)^{d-1}.\label{equation:MonotonicityofRelativerenyi_2}
    \end{align}
    The equation (a) is given by the additivity Lemma~\ref{theorem:additivity}.
    Also, the equation (b) is given by the additivity of $D_{1+s}(P^{\bm{M}^n}_{\rho^{\otimes n}}||P^{\bm{M}^n}_{\sigma^{\otimes n}})$.

    By deviding the inequality \eqref{equation:MonotonicityofRelativerenyi_2} by $n$, we obtain
    \begin{align}
        D_{1+s}(\rho||\sigma)\ge \frac{1}{n}D_{1+s}(\kappa_{\sigma^{\otimes n}}(\rho^{\otimes n})||\sigma^{\otimes n})\ge D_{1+s}(P^{\bm{M}}_\rho||P^{\bm{M}}_\sigma)-\frac{1+s}{ns}\log(n+1)^{d-1}\label{equation:MonotonicityofRelativerenyi_3}.
    \end{align}
    The final term $\frac{1}{n}\log(n+1)^{d-1}$ converges to $0$ by taking $n\to \infty$.
   As a result, the desired inequality is proven.
\end{proof}

\subsection{Sandwiched Relative R\'{e}nyi entropy}\label{subsection:sandwichedrelative}
In this part, we mention about information inequalities of SRR entropy.
In particular, we give the monotonicity of SRR entropy in TPCP map as follows:
\begin{theorem}[Monotonicity of SRR entropy by TPCP map]\label{theorem:MonotonicityofSandwich}
    Let $\rho,\sigma$ be the states in $\cV$.
    Also, let $\kappa:\cV\to \cV$ be the TPCP map.
    Then, the following inequality holds.
    \begin{align}
        \underline{D}_{1+s}(\rho||\sigma)\ge \underline{D}_{1+s}(\kappa(\rho)||\kappa(\sigma)),\quad s>0.\label{equation:monotonicityinTPCPmap}
    \end{align}
\end{theorem}

In order to prove Theorem~\ref{theorem:MonotonicityofSandwich}, we organize the following three lemmas.
The first and second lemmas show the third lemma.
The third lemma shows Theorem~\ref{theorem:MonotonicityofSandwich}.
\begin{lemma}\label{lemma:inequalityofSandwich}
    Let $\rho,\sigma$ be the states in $\cV$.
    Then, the following inequality holds.
    \begin{align}
        D_{1+s}(\kappa_\sigma(\rho)||\sigma)+\frac{1+s}{s}\log|\bm{C}_\sigma|\ge \underline{D}_{1+s}(\rho||\sigma)\ge D_{1+s}(\kappa_\sigma(\rho)||\sigma),\quad s>0.
    \end{align}
\end{lemma}

\begin{lemma}\label{lemma:observationinequalityofSandwich}
    Let $\rho,\sigma$ be the states in $\cV$.
    Also, let $\bm{M}$ be the measurement in $\cV$.
    Then, the following inequality holds.
    \begin{align}
        \underline{D}_{1+s}(\rho||\sigma)\ge D_{1+s}(P^{\bm{M}}_\rho||P^{\bm{M}}_\sigma),\quad s>0.
    \end{align}
\end{lemma}
The proof of Lemma~\ref{lemma:inequalityofSandwich} and Lemma~\ref{lemma:observationinequalityofSandwich}
are provided in Appendix~\ref{appendix:proofofSandwich}.
From these two lemmas, we obtain the following lemma.
\begin{lemma}\label{lemma:limofSandwich}
    Let $\rho^{\otimes n},\sigma^{\otimes n}$ be states in $\cV^{\otimes n}$.
    Also, let $\bm{M}^n$ be a measurement in $\cV^{\otimes n}$.
    Then, the following equality holds.
    \begin{align}
        \underline{D}_{1+s}(\rho||\sigma)=\lim_{n\to\infty}\frac{1}{n}\max_{\bm{M}^n}D_{1+s}(P^{\bm{M}^n}_{\rho^{\otimes n}}||P^{\bm{M}^n}_{\sigma^{\otimes n}}),\quad s>0.\label{equation:limofSandwich}
    \end{align}
    In addition, the following equality holds.
    \begin{align}
        \underline{D}_{1+s}(\rho||\sigma)&=\lim_{n\to\infty}\frac{1}{n}D_{1+s}(\kappa_{\sigma^{\otimes n}}(\rho^{\otimes n})||\sigma^{\otimes n})\label{equation:limofSandwich2}\\
        &=\lim_{n\to\infty}\frac{1}{n}D_{1+s}(P^{I^{\rho^{\otimes n}}_{\sigma^{\otimes n}}}_{\rho^{\otimes n}}||P^{I^{\rho^{\otimes n}}_{\sigma^{\otimes n}}}_{\sigma^{\otimes n}}).\label{equation:limofSandwich3}
    \end{align}
    Therefore, the family $\{I^{\rho^{\otimes n}}_{\sigma^{\otimes n}}\}$
    can be selected as a measurement $\bm{M}^n$ in \eqref{equation:limofSandwich}.
\end{lemma}

\begin{proof}[Proof of Lemma~\ref{lemma:limofSandwich} by assuming Lemma~\ref{lemma:inequalityofSandwich} and Lemma~\ref{lemma:observationinequalityofSandwich}]
    We apply Lemma~\ref{lemma:observationinequalityofSandwich} to the states $\rho^{\otimes n},\sigma^{\otimes n}$ and the measurement $\bm{M}^n$ .
    Then, we obtain
    \begin{align}
        n\underline{D}_{1+s}(\rho||\sigma)\stackrel{(a)}{=}\underline{D}_{1+s}(\rho^{\otimes n}||\sigma^{\otimes n})\ge D_{1+s}(P^{\bm{M}^n}_{\rho^{\otimes n}}||P^{\bm{M}^n}_{\sigma^{\otimes n}}).\label{equation:inequalityofSandwich1}
    \end{align}
    The equation (a) is shown by additivity of SRR entropy (Lemma~\ref{theorem:additivity}).
    On the other hand, we apply Lemma~\ref{lemma:inequalityofSandwich} to the states $\rho^{\otimes n},\sigma^{\otimes n}$.
    Then, we obtain
    \begin{align}
        \max_{\bm{M}^n}D_{1+s}(P^{\bm{M}^n}_{\rho^{\otimes n}}||P^{\bm{M}^n}_{\sigma^{\otimes n}})+\frac{1+s}{s}\log|\bm{C}_{\sigma^{\otimes n}}|&\stackrel{(a)}{\ge} D_{1+s}(\kappa_{\sigma^{\otimes n}}(\rho^{\otimes n})||\sigma^{\otimes n})+\frac{1+s}{s}\log|\bm{C}
_{\sigma^{\otimes n}}|\\
        &\ge n\underline{D}_{1+s}(\rho||\sigma) \stackrel{(b)}{=} \underline{D}_{1+s}(\rho^{\otimes n}||\sigma^{\otimes n})\ge D_{1+s}(\kappa_{\sigma^{\otimes n}}(\rho^{\otimes n})||\sigma^{\otimes n}).\label{equation:inequalityofSandwich3}
    \end{align}
    The inequality (a) is shown by compering $D_{1+s}(\kappa_{\sigma^{\otimes n}}(\rho^{\otimes n})||\sigma^{\otimes n})$ to the maximam value \\
    $\max_{\bm{M}^n}D_{1+s}(P^{\bm{M}^n}_{\rho^{\otimes n}}||P^{\bm{M}^n}_{\sigma^{\otimes n}})$
    from Lemma~\ref{lemma:RepresentEntropies}.
    The equality (b) is shown by additivity Lemma~\ref{theorem:additivity}.
    Now, we divide \eqref{equation:inequalityofSandwich3} by $n$, then we obtain
    \begin{align}
        \frac{1}{n}\max_{\bm{M}^n}D_{1+s}(P^{\bm{M}^n}_{\rho^{\otimes n}}||P^{\bm{M}^n}_{\sigma^{\otimes n}})+\frac{1+s}{ns}\log(1+n)^{d-1}
        &\stackrel{(a)}{\ge}\frac{1}{n}\max_{\bm{M}^n}D_{1+s}(P^{\bm{M}^n}_{\rho^{\otimes n}}||P^{\bm{M}^n}_{\sigma^{\otimes n}})+\frac{1+s}{ns}\log|\bm{C}_{\sigma^{\otimes n}}|\\
        &\ge \underline{D}_{1+s}(\rho||\sigma)
        \stackrel{(b)}{\ge} \frac{1}{n}\max_{\bm{M}^n}D_{1+s}(P^{\bm{M}^n}_{\rho^{\otimes n}}||P^{\bm{M}^n}_{\sigma^{\otimes n}}),\label{equation:inequalityofSandwich4}
    \end{align}
    where $d:=|\bm{C}_\sigma|$.
    The inequality (a) is shown by Lemma~\ref{lemma:numberofeigenvalues}, i.e., $|\bm{C}_{\sigma^{\otimes n}}|\le (1+n)^{d-1}$
    The inequality (b) is shown by \eqref{equation:inequalityofSandwich1} with taking the maximum of $\bm{M}^n$.
    The term $\frac{1+s}{ns}\log(1+n)^{d-1}$ converges to $0$ by taking $n\to\infty$.
    As a result, we obtain the conclusion. 
\end{proof}

Now, we prove Theorem~\ref{theorem:MonotonicityofSandwich} under Lemma~\ref{lemma:limofSandwich}.
\begin{proof}[Proof of Theorem~\ref{theorem:MonotonicityofSandwich} by assuming Lemma~\ref{lemma:limofSandwich}]
    For a measurement $\bm{M}^n=\{M^n_i\}$ and a TPCP map $\kappa:\cV\to \cV$, 
    we consider the family $\{\kappa^{\otimes n *}(M^n_i)\}$, 
    where $\kappa^{\otimes n*}$ denotes the adjoint map of $\kappa^{\otimes n}$. 
    By applying Lemma~\ref{lemma:TPCPtensor} and Lemma~\ref{lemma:identitypreserving}, 
    the family $\{\kappa^{\otimes n *}(M^n_i)\}$ is a measurement.
    Now we apply Lemma~\ref{lemma:limofSandwich}, we obtain the desired inequality as follows:
    \begin{align}
        \underline{D}_{1+s}(\rho||\sigma)&=\lim_{n\to\infty}\frac{1}{n}\max_{\bm{M}^n}D_{1+s}(P^{\bm{M}^n}_{\rho^{\otimes n}}||P^{\bm{M}^n}_{\sigma^{\otimes n}})\\
        &\ge \lim_{n\to \infty}\frac{1}{n}\max_{\bm{M}^n}D_{1+s}(P^{\kappa^{\otimes n *}(\bm{M}^n)}_{\rho^{\otimes n}}||P^{\kappa^{\otimes n *}(\bm{M}^n)}_{\sigma^{\otimes n}})\\
        &=\lim_{n\to \infty}\frac{1}{n}\max_{\bm{M}^n}D_{1+s}(P^{\bm{M}^n}_{\kappa(\rho)^{\otimes n}}||P^{\bm{M}^n}_{\kappa(\sigma)^{\otimes n}})=D_{1+s}(\kappa(\rho)||\kappa(\sigma)),\quad s>0.
    \end{align}
\end{proof}

As a corollary of Theorem~\ref{theorem:MonotonicityofRelativerenyi} and the equation \ref{equation:limofSandwich2} in the proof of Lemma~\ref{lemma:limofSandwich},
we obtain the following relation between PRR entropy and SRR entropy,
but the corollary is not directly related to the main topic.
\begin{corollary}\label{theorem:RelationbetweenRelativeandSandwich}
    Let $\rho,\sigma$ be the states in $\cV$.
    Then, the following inequality holds.
    \begin{align}
        D_{1+s}(\rho||\sigma)\ge \underline{D}_{1+s}(\rho||\sigma),\quad s>0.
    \end{align}    
\end{corollary}



\subsection{Relative entropy}\label{subsection:relative}
In this part, we investigate some relations of Relative entropy
from the relations given in Section~\ref{subsection:relativerenyi} and Section~\ref{subsection:sandwichedrelative}.
At first, we give monotonicity of Relative entropy with a TPCP map
from monotonicity of SRR entropy with a TPCP map(Theorem~\ref{theorem:MonotonicityofSandwich}).
Secondly, we give joint convexity of Relative entropy 
and monotonicity of Relative entropy with an observation. 
Finally, we show the asymptotic equivalence between single shot Relative entropy
and $n$-shot Relative entropy with an observation.

The monotonicity of Relative entropy with a TPCP map is given as follows.
\begin{theorem}[Monotonicity of relative entropy by TPCP map]\label{theorem:MonotonicityofRelative}
    Let $\rho,\sigma$ be states in $\cV$.
    Also, let $\kappa:\cV\to \cV$ be a TPCP map.
    Then, the following inequality holds:
    \begin{align}
        D(\rho||\sigma)\ge D(\kappa(\rho)||\kappa(\sigma)).\label{equation:MonotonicityofRelative1}
    \end{align}
\end{theorem}
\begin{proof}[Proof of Theorem~\ref{theorem:MonotonicityofRelative}]
    In Theorem~\ref{theorem:MonotonicityofSandwich}, we take the limit $s$ to $0$ in \eqref{equation:monotonicityinTPCPmap}.
    Then, we obtain the desired inequality \eqref{equation:MonotonicityofRelative1} from Theorem~\ref{lemma:limitofRenyies}.
\end{proof}
From Theorem~\ref{theorem:MonotonicityofRelative} and Lemma~\ref{lemma:TPCPofPartialtrace},
we obtain joint convexity of Relative entropy as follows:
\begin{theorem}[Joint convexity of Relative entropy]\label{theorem:jointconvexity}
    Let $\rho_i,\sigma_i,\quad(i=1,\ldots,k)$ be states in $\cV$.
    Also, let $\{p_i\}_{i=1}^k$ be a probability distribution.
    Then, the following inequality holds:
    \begin{align}
        \sum_{i=1}^k p_i D(\rho_i||\sigma_i)\ge D(\sum_{i=1}^k p_i \rho_i||\sum_{i=1}^k p_i\sigma_i).
    \end{align} 
\end{theorem}
\begin{proof}[Proof of Theorem~\ref{theorem:jointconvexity}]
    Let $x,y$ be the states $x=\sum_{i=1}^k p_i u_i\otimes \rho_i$, $y=\sum_{i=1}^k p_i u_i\otimes \sigma_i$
    in $\mathbb{R}^k\otimes \cV$, 
    where $u_i\in\mathbb{R}^d(i=1,\ldots,d)$ are the element which takes $1$ in $i$th element and $0$ in others.
    Then, from Theorem~\ref{theorem:MonotonicityofRelative}, we obtain
    \begin{align}
        D(x||y)\ge D(\mathrm{tr}_{\mathbb{R}^k}x||\mathrm{tr}_{\mathbb{R}^k}y),\label{equation:jointconvexity}
    \end{align}
where $\mathrm{tr}_{\mathbb{R}^k}$ is the partial trace onto $\mathbb{R}^k$(Defininition~\ref{definition:partialtrace}).
Here, we calculate LHS of \eqref{equation:jointconvexity} by definition, and we obtain
\begin{align}
    D(x||y)=\sum_{i=1}^k p_iD(\rho_i||\sigma_i).\label{equation:jointconvexity1}
\end{align}
On the other hand, we calculate 
$\mathrm{tr}_{\mathbb{R}^k}x$ and $\mathrm{tr}_{\mathbb{R}^k}y$ in RHS of \eqref{equation:jointconvexity}, and we obtain
\begin{align}
    \mathrm{tr}_{\mathbb{R}^k}x=\sum_{i=1}^kp_i\rho_i,\mathrm{tr}_{\mathbb{R}^k}y=\sum_{i=1}^kp_i\sigma_i.\label{equation:jointconvexity2}
\end{align}
As a result, we obtain the conclusion.
\end{proof}

Next, monotonicity of Relative entropy with an observation also holds as a corollary of Theorem~\ref{theorem:MonotonicityofRelative}.
\begin{theorem}\label{theorem:observationinequalityofRelative}
    Let $\rho,\sigma$ be states in $\cV$.
    Also, let $\bm{M}=\{M_i\}_{i=1}^k$ be a measurement in $\cV$.
    Then, the following inequality holds:
    \begin{align}
        D(\rho||\sigma)\ge D(P^{\bm{M}}_\rho||P^{\bm{M}}_\sigma).
    \end{align}
\end{theorem}
\begin{proof}
    We apply Theorem~\ref{theorem:MonotonicityofRelative} 
    to the TPCP map $\kappa_{\bm{M}}$ defined in Definition~\ref{definition:TPCPobservation}
    for the measurement $M$.
\end{proof}

From Theorem~\ref{theorem:jointconvexity}, we prove the following theorem, 
which is essential to show direct part of Stein's theorem with EJAs.
\begin{theorem}\label{theorem:DirectpartofRelative}
    Let $\rho,\sigma$ be states in $\cV$.
    Then, for the measurement $I^{\rho^{\otimes n}}_{\sigma^{\otimes n}}$ in Definition~\ref{definition:pinchedmeasurement},
    the following relation holds:
    \begin{align}
        \lim_{n\to\infty}\frac{1}{n}D^{I^{\rho^{\otimes n}}_{\sigma^{\otimes n}}}(\rho^{\otimes n}||\sigma^{\otimes n})=D(\rho||\sigma)\label{equation:DirectpartofRelative6}.
    \end{align} 
\end{theorem}
Theorem~\ref{theorem:DirectpartofRelative} is shown from the following two lemmas.
\begin{lemma}\label{lemma:DirectpartofRelative1}
    Let $\rho,\sigma$ be states in $\cV$.
    Then, the following relation holds:
    \begin{align}
        D(\rho||\sigma)=D(\rho||\kappa_\sigma(\rho))+D(\kappa_\sigma(\rho)||\sigma).
    \end{align}
\end{lemma}
\begin{lemma}\label{lemma:DirectpartofRelative2}
    Let $\bm{C}=\{e_i\}$ be a CSOI.
    Also, let $\rho$ be a state in $\cV$.
    Then, the following relation holds:
    \begin{align}
        D(\rho||\kappa_{\bm{C}}(\rho))=H(\kappa_{\bm{C}}(\rho))-H(\rho)\le \log|\bm{C}|.
    \end{align}
\end{lemma}
Lemma~\ref{lemma:DirectpartofRelative1} and Lemma~\ref{lemma:DirectpartofRelative2} are provided in Appendix~\ref{appendix:proofofRelative}.
Here, we prove Theorem~\ref{theorem:DirectpartofRelative} under Lemma~\ref{lemma:DirectpartofRelative1} and Lemma~\ref{lemma:DirectpartofRelative2}.
\begin{proof}[Proof of Theorem~\ref{theorem:DirectpartofRelative} assuming Lemma~\ref{lemma:DirectpartofRelative1} and Lemma~\ref{lemma:DirectpartofRelative2}]
    Applying Lemma~\ref{lemma:DirectpartofRelative1} to the states $\rho^{\otimes n},\sigma^{\otimes n}$ in $\cV^{\otimes n}$.
    We obtain the following equation:
    \begin{align}
        D(\rho^{\otimes n}||\sigma^{\otimes n})=D(\rho^{\otimes n}||\kappa_{\sigma^{\otimes n}}(\rho^{\otimes n}))+D(\kappa_{\sigma^{\otimes n}}(\rho^{\otimes n})||\sigma^{\otimes n}).\label{equation:DirectpartofRelative1}
    \end{align}
    First, we estimate the first term of RHS \eqref{equation:DirectpartofRelative1} as follows:
    \begin{align}
        D(\rho^{\otimes n}||\kappa_{\sigma^{\otimes n}}(\rho^{\otimes n}))\stackrel{(a)}{\le} \log|\bm{C}_{\sigma^{\otimes n}}|\stackrel{(b)}{\le} \log(n+1)^{d-1},\label{equation:DirectpartofRelative2}
    \end{align}
    where $d:=|\bm{C}_\sigma|$.
    The equation (a) is shown by Lemma~\ref{lemma:DirectpartofRelative2}.
    The equation (b) is shown by Lemma~\ref{lemma:numberofeigenvalues}.
    Second, from Lemma~\ref{lemma:RepresentEntropies}, we rewrite the second term of RHS \eqref{equation:DirectpartofRelative1} as follows:
    \begin{align}
        D(\kappa_{\sigma^{\otimes n}}(\rho^{\otimes n})||\sigma^{\otimes n})=D^{I^{\rho^{\otimes n}}_{\sigma^{\otimes n}}}(\rho^{\otimes n}||\sigma^{\otimes n}).\label{equation:DirectpartofRelative3}
    \end{align}
    Applying \eqref{equation:DirectpartofRelative2} and \eqref{equation:DirectpartofRelative3} to \eqref{equation:DirectpartofRelative1}, we obtain the following upper bound of $nD(\rho||\sigma)$:
    \begin{align}
        nD(\rho||\sigma)\stackrel{(a)}{=}D(\rho^{\otimes n}||\sigma^{\otimes n})\le \log(n+1)^{d-1}+D^{I^{\rho^{\otimes n}}_{\sigma^{\otimes n}}}(\rho^{\otimes n}||\sigma^{\otimes n}).\label{equation:DirectpartofRelative4}
    \end{align}
    The equation (a) is shown by additivity of Relative entropy Lemma~\ref{theorem:additivity}.
    On the other hand, by Theorem~\ref{theorem:observationinequalityofRelative}, the following lower bound of $nD(\rho||\sigma)$ holds:
    \begin{align}
        D^{I^{\rho^{\otimes n}}_{\sigma^{\otimes n}}}(\rho^{\otimes n}||\sigma^{\otimes n})\le D(\rho^{\otimes n}||\sigma^{\otimes n})=nD(\rho||\sigma).\label{equation:DirectpartofRelative5}
    \end{align}
    Finally, combining \eqref{equation:DirectpartofRelative4} and \eqref{equation:DirectpartofRelative5}, we obtain
    \begin{align}
        D(\rho||\sigma)-\frac{\log(n+1)^{d-1}}{n}\le \frac{1}{n}D^{I^{\rho^{\otimes n}}_{\sigma^{\otimes n}}}(\rho^{\otimes n}||\sigma^{\otimes n})\le D(\rho||\sigma).
    \end{align}
    The term $\frac{\log(n+1)^{d-1}}{n}$ converges to $0$ when $n\to\infty$.
    Therefore, we obtain the conclusion.
\end{proof}

\begin{remark}
    We can not directly show Theorem~\ref{theorem:DirectpartofRelative} 
    by taking limit $s$ to $0$ in \eqref{equation:limofSandwich3}
    because $\frac{1+s}{ns}\log|\bm{C}_{\sigma^{\otimes n}}|$ diverges to infinity 
    as $s$ approaches $0$ in \eqref{equation:inequalityofSandwich4}.
\end{remark}

\section{Hypothesis testing and Stein's Lemma in Euclidean Jordan algebra}\label{section:hypothesistesting}

In this section, we prepare the setting of hypothesis testing and
prove a generalization of Stein's Lemma with EJAs.
In order to prove Stein's Lemma with EJAs, we separate the problem into two parts,
the direct part(Section~\ref{subsection:directpart}) and the converse part(Section~\ref{subsection:conversepart}).

\subsection{Settings and Stein's Lemma}\label{subsection:setting}

Hypothesis testing is an information task, which determines whether
we support alternative hypothesis with rejecting null hypothesis or
we support null hypothesis with rejecting alternative hypothesis.
Similarly to the setting of quantum Stein's Lemma, 
we have an i.i.d. source of an unknown state.
Now, we consider Null hypothesis: the unknown state is given as $\rho$ and 
Alternative hypothesis: the unknown state is given as $\sigma$.
By applying the i.i.d. source $n$-times and a global measurement $\{T,u-T\}$ one time,
we determine the hypothesis as the measurement outcome.
In this case, there are two types of errors.
The type I error, where we support the alternative hypothesis 
but the null hypothesis is correct, occurs with probability $\langle \rho^{\otimes n}, u-T\rangle$.
The type II error, where we support the null hypothesis 
but the alternative hypothesis is correct, occurs with probability $\langle\sigma^{\otimes n},T\rangle$.
We aim to minimize the two types of error probabilities, 
but they are related to each other as trade-off.
Then, we consider the case that we minimize the type II error 
under a bound of the type I error, and we introduce the following quantity.
\begin{definition}\label{definition:beta}
    For states $\rho,\sigma\in \cV$, we define the following quantity:
    \begin{align}
        \beta^n_\epsilon(\rho||\sigma):=\min_{0\le T\le u}\{\langle\sigma^{\otimes n},T\rangle|\langle\rho^{\otimes n},u-T\rangle\le \epsilon\},\quad 0<\epsilon<1,
    \end{align} 
    where the condition $0\le T\le u$ in the minimization is considered in the space $V^{\otimes n}$.
\end{definition}
In quantum theory, the references \cite{HiaiPetz} and \cite{OgawaNagaoka} have proved that
the exponent of $\beta^n_\epsilon$ is asymptotically equivalent to the relative entropy.
In this paper, we prove the statement even in EJAs, i.e., we prove the following theorem:
\begin{theorem}\label{theorem:Stein}
    For states $\rho,\sigma$ and any $0<\epsilon<1$, the following relation holds:
    \begin{align}
        \lim_{n\to \infty}-\frac{1}{n}\log\beta^n_\epsilon(\rho||\sigma)=D(\rho||\sigma).
    \end{align}
\end{theorem}

Similarly to quantum Stein's Lemma, for simplicity of the proof,
we introduce the following two quantities.
\begin{definition}\label{definition:BandBdag}
    For states $\rho,\sigma\in \cV$, we define the following quantities:
    \begin{align}
        B(\rho||\sigma)&:=\sup_{\{0\le T_n\le u\}}\left\{\varliminf_{n\to\infty}-\frac{1}{n}\log\langle\sigma^{\otimes n},T_n\rangle\mid \lim_{n\to\infty}\langle\rho^{\otimes n},u-T_n\rangle=0\right\},\\
        B^\dag(\rho||\sigma)&:=\sup_{\{0\le T_n\le u\}}\left\{\varliminf_{n\to\infty}-\frac{1}{n}\log\langle\sigma^{\otimes n},T_n\rangle\mid \varliminf_{n\to\infty}\langle\rho^{\otimes n},u-T_n\rangle<1\right\},
    \end{align}
    where the condition of supremum $\{0\le T_n\le u\}$ is the family of the inequalities 
    and each inequality $0\le T_n\le u$ is considered in the space $\cV^{\otimes n}$.
\end{definition}

Similarly to quantum Stein's Lemma, we prove the following theorem with $B,B^\dag$ at first,
and we prove Theorem~\ref{theorem:Stein} by applying the following theorem.
\begin{theorem}\label{theorem:Stein2}
    For states $\rho,\sigma\in \cV$, the following relations hold.
    \begin{align}
        B^\dag(\rho||\sigma)=B(\rho||\sigma)=D(\rho||\sigma).
    \end{align}
\end{theorem}

In the following sections, we prove Theorem~\ref{theorem:Stein2}.
The implication of Theorem~\ref{theorem:Stein2} to Theorem~\ref{theorem:Stein} is shown in Appendix~\ref{appendix:proofofStein}.
Because of the relation $B^\dag\ge B$ by Definition~\ref{definition:BandBdag}, 
we divide Theorem~\ref{theorem:Stein2} into two parts, the direct part and the converse part.
\begin{lemma}[Direct part]\label{lemma:direct}
    For states $\rho,\sigma\in \cV$, the following inequality holds:
    \begin{align}
        B(\rho||\sigma)\ge D(\rho||\sigma).\label{equation:direct3}
    \end{align}
\end{lemma}
\begin{lemma}[Converse part]\label{lemma:converse}
    For states $\rho,\sigma\in \cV$, the following inequality holds:
    \begin{align}
        D(\rho||\sigma)\ge B^\dag(\rho||\sigma).
    \end{align}
\end{lemma}
We prove Lemma~\ref{lemma:direct} in Section~\ref{subsection:directpart}
and Lemma~\ref{lemma:converse} in Section~\ref{subsection:conversepart}.

\subsection{Direct part}\label{subsection:directpart}

In this subsection, we prove Direct part (Lemma~\ref{lemma:direct}).
By applying Theorem~\ref{theorem:DirectpartofRelative}, we prove Direct part as follows:
\begin{proof}[Proof of Lemma~\ref{lemma:direct}]
At first, we take the family of measurement $\{I^{\rho^{\otimes n}}_{\sigma^{\otimes n}}\}_n$ in Definition~\ref{definition:pinchedmeasurement}.
From Theorem~\ref{theorem:DirectpartofRelative}, for each $\epsilon>0$, there exists $N_\epsilon$ such that 
\begin{align}
    \frac{1}{n}D^{I^{\rho^{\otimes n}}_{\sigma^{\otimes n}}}(\rho^{\otimes n}||\sigma^{\otimes n})\ge D(\rho||\sigma)-\epsilon,\quad n\ge N_\epsilon.\label{equation:direct1}
\end{align}
Here, We take $k\ge N_\epsilon$ and a set $A_l\subset \{1,\ldots,m\}^l$ for arbitrary $l\in \mathbb{N}$,
where $m$ is the number of element in the measurement $I^{\rho^{\otimes k}}_{\sigma^{\otimes k}}$.

Then, we obtain
\begin{align}
    B(\rho||\sigma)\stackrel{(a)}{\ge} \frac{1}{k}B(P^{I^{\rho^{\otimes k}}_{\sigma^{\otimes k}}}_{\rho^{\otimes k}}||P^{I^{\rho^{\otimes k}}_{\sigma^{\otimes k}}}_{\sigma^{\otimes k}})\stackrel{(b)}{=}\frac{1}{k}D^{I^{\rho^{\otimes k}}_{\sigma^{\otimes k}}}(\rho^{\otimes k}||\sigma^{\otimes k}).\label{equation:direct2}
\end{align}
The inequality (a) is shown by the definition of $B(\rho||\sigma)$.
The equation (b) is shown by applying the set $A_l$ to definition of $\beta_\epsilon^n(p||q)$ and classical Stein's Lemma(Theorem~\ref{theorem:classicalstein}).
Combining \eqref{equation:direct1} and \eqref{equation:direct2}, we obtain
\begin{align}
    B(\rho||\sigma)\ge D(\rho||\sigma)-\epsilon.
\end{align}
The parameter $\epsilon>0$ is chosen arbitrary, and therefore, we conclude the inequality \eqref{equation:direct3}.
\end{proof}

\subsection{Converse part}\label{subsection:conversepart}

In this subsection, we show Converse part (Lemma~\ref{lemma:converse}).
At first, we estimate the type I error by SRR entropy as follows.
\begin{lemma}\label{lemma:converse1}
    Let $\rho,\sigma$ be states in $\cV$.
    An effect $T_n$ in $V^{\otimes n}$ satisfies the following inequality holds for arbitrary $n\in \mathbb{N}$ and $s<0$:
    \begin{align}
        -\frac{1}{n}\log\langle\rho^{\otimes n},T_n\rangle\ge\frac{-\tilde{\phi}(s|\rho||\sigma)-s(-\frac{1}{n}\log\langle\sigma^{\otimes n},T_n\rangle)}{1-s},\label{equation:converse1}
    \end{align} 
    where $\tilde{\phi}(s|\rho||\sigma)$ is defined in Definition~\ref{definition:informationquantities}.
\end{lemma}
\begin{proof}
    At first, the following relation holds for $s<0$:
\begin{align}
    &(\langle \rho^{\otimes n},T_n\rangle)^{1-s}(\langle\sigma^{\otimes n},T_n\rangle)^s \nonumber\\
\le &(\langle \rho^{\otimes n},T_n\rangle)^{1-s}(\langle\sigma^{\otimes n},T_n\rangle)^s+(\langle \rho^{\otimes n},u-T_n\rangle)^{1-s}(\langle\sigma^{\otimes n},u-T_n\rangle)^s\\
=&e^{\phi(s|P^{T_n}_{\rho^{\otimes n}}||P^{T_n}_{\sigma^{\otimes n}})}.\label{equation:converse2}
\end{align}
In \eqref{equation:converse2}, we take logarithm and divide the equation by $n$, we organize
the equation as follows:
\begin{align}
    (1-s)\frac{1}{n}\log\langle\rho^{\otimes n},T_n\rangle+s\frac{1}{n}\log\langle\sigma^{\otimes n},T_n\rangle\le \frac{1}{n}\phi(s|P^{T_n}_{\rho^{\otimes n}}||P^{T_n}_{\sigma^{\otimes n}})\stackrel{(a)}{\le} \tilde{\phi}(s|\rho||\sigma).\label{equation:converse3}
\end{align}
The inequality (a) is shown as Lemma~\ref{lemma:observationinequalityofSandwich}.
Finally, we divide \eqref{equation:converse3} by $1-s>0$ and organize the inequality, we obtain \eqref{equation:converse1}.
\end{proof}

Next, under the condition about type II error, the limitation of type I error is bounded with $r$ as follows.
\begin{lemma}\label{lemma:converse2}
Let $\rho,\sigma$ be states in $\cV$.
We take an arbitrary effect $T_n$ in $V^{\otimes n}$ and a number $r:=\varliminf_{n\to\infty}-\frac{1}{n}\log \langle\sigma^{\otimes n},T_n\rangle$.
If $r>D(\rho||\sigma)$, there exist $s_0<0$ such that
\begin{align}
    \varliminf_{n\to\infty}-\frac{1}{n}\log\langle\rho^{\otimes n},T_n\rangle\ge \frac{-\tilde{\phi}(s_0|\rho||\sigma)-s_0r}{1-s_0}>0.\label{equation:converse4}
\end{align}    

\end{lemma}
\begin{proof}
    The First inequality of \eqref{equation:converse4} is shown by taking the inferior limit in \eqref{equation:converse1}.
    The equation $\lim_{s\to 0}\frac{\tilde{\phi}(s|\rho||\sigma)}{-s}=D(\rho||\sigma)$ holds by Lemma~\ref{lemma:limitofRenyies}.
     The monotonicity of $\underline{D}_{1+s}$(Lemma~\ref{lemma:monotonicityins}) ensures the existence of $s_0<0$ satisfying $r>\frac{\tilde{\phi}(s_0|\rho||\sigma)}{-s_0}=\frac{\tilde{\phi}(s_0|\rho||\sigma)-\tilde{\phi}(0|\rho||\sigma)}{-s_0}$.
    Therefore, the following relation holds:
    \begin{align}
        \frac{\tilde{\phi}(s_0|\rho||\sigma)-s_0r}{1-s_0}=\frac{s_0}{1-s_0}(\frac{\tilde{\phi}(s_0|\rho||\sigma)}{-s_0}-r)>0.
    \end{align}
	As a result, the inequality \eqref{equation:converse4} holds.
\end{proof}

Applying Lemma~\ref{lemma:converse1} and Lemma~\ref{lemma:converse2}, we obtain Converse part (Lemma~\ref{lemma:converse}).
\begin{proof}[Proof of Lemma~\ref{lemma:converse}]
    We take a number $r$ as
    \begin{align}
        r:=\varliminf_{n\to \infty}-\frac{1}{n}\log\langle\sigma^{\otimes n},T_n\rangle.
    \end{align}
    In the case of $r>D(\rho||\sigma)$, from Lemma~\ref{lemma:converse2}, we obtain
    \begin{align}
        \varlimsup_{n\to\infty}\langle\rho^{\otimes n},T_n\rangle=0
    \end{align}
    In this case, 
    \begin{align}
        \varliminf_{n\to\infty}\langle\rho^{\otimes n},u-T_n\rangle=1.\label{equation:converse5}
    \end{align}
    holds.
    Therefore, in this case, the family of effects $\{T_n\}$ does not satisfy the condition of $B^\dag(\rho||\sigma)$.
    If we take effects which don't satisfy the condition \eqref{equation:converse5}, the relation
    \begin{align}
        D(\rho||\sigma)\ge \varliminf_{n\to\infty}-\frac{1}{n}\log\langle\sigma^{\otimes n},T_n\rangle.\label{equation:converse6}
    \end{align}
    holds.
    We take supremum in \eqref{equation:converse6} with effects which don't satisfy \eqref{equation:converse5}, we obtain Converse part.
\end{proof}

\section{Quantum Realization of EJAs}\label{section:quantumrealization}

In this section, we discuss how we realize the model associated with EJAs in quantum theory.
First, we define canonical Jordan subalgebras and canonical embedding map.
Then, we show that canonical embedding map preserve SRR entropy for any $s>0$, and as a result, we give another proof of Stein's Lemma if there exists a canonical embedding map into quantum theory.
Finally, we see that Lorentz type and Quaternion type, which are the remaining type of simple EJA except for Octonion type, are canonically embedded into quantum theory.
In other words, we conclude another proof of Stein's Lemma if the single system does not contain any Octonion part.

\subsection{Canonical Jordan subalgebra}
\label{section:CanonicalJordanSub}

First, we define the canonical Jordan subalgebras and see their properties.

We consider a Jordan algebra $\cV$.
A strictly positive definite inner product $\langle ~,~\rangle$
is called canonical when $\cQ_\cV=\cQ_\cV^\ast$, i.e., 
\begin{align}
\{X\circ X\mid X\in\cV\}=\{X \in \cV \mid \langle X,Y\rangle\ge 0 ,\forall Y \in \cQ_\cV\}.
\end{align}
A subspace $\cV_1$ of a Jordan algebra $\cV$ with the unit $u$
is called a Jordan subalgebra of $\cV$
when $\cV_1$ contains $u$ and is closed for the Jordan product
of $\cV$.

A Jordan subalgebra $\cV_1$ of 
${\cal V}$ with a canonical inner product $\langle ~,~\rangle$
is called a canonical Jordan subalgebra of 
${\cal V}$ with a canonical inner product $\langle ~,~\rangle$
when the inner product $\langle ~,~\rangle$ is 
canonical even for the Jordan subalgebra ${\cal V}_1$.

Now, we choose a canonical Jordan subalgebra ${\cal V}_1$ of 
${\cal V}$ with a canonical inner product $\langle ~,~\rangle$.
We choose 
two cones $\cQ_\cV$ and $\cQ_{\cV_1}$.
Also, we consider their state spaces $\cS(\cQ_\cV,u)$, $\cS(\cQ_{\cV_1},u)$ and their measurement spaces $\cM(\cQ_\cV,u)$, $\cM(\cQ_{\cV_1},u)$.

For two states $\rho,\rho'\in \cS(\cQ_\cV,u)$, 
we denote $\rho\sim_{\cM(\cQ_{\cV_1},u)} \rho' $ when
\begin{align}
\langle \rho, M_j\rangle =
\langle \rho', M_j\rangle , \quad \forall j ,
\quad \forall \{ M_j\}_j \in \cM(\cQ_{\cV_1},u).
\end{align}

For two measurements $\{ M_j\}_j,\{ M_j'\}_j \in \cM(\cQ_\cV,u)$, 
we denote $\{ M_j\}_j\sim_{\cS(\cQ_{\cV_1},u)} \{ M_j'\}_j$ when
\begin{align}
\langle \rho, M_j\rangle =
\langle \rho, M_j'\rangle , \quad \forall j ,
\quad \forall \rho \in {\cS(\cQ_{\cV_1},u)}_1.
\end{align}

When our state is limited into $\cS(\cQ_{\cV_1},u)$,
any measurement can be written as an element of 
$\cM(\cQ_{\cV_1},u)$.
Hence, we have the following theorem.

\begin{theorem}\label{theorem:embed-measurement}
For any measurement $\{ M_j\}_j \in \cM(\cQ_\cV,u)$,
there exists a measurement $\{ M_j'\}_j \in \cM(\cQ_{\cV_1},u)$
such that
$\{ M_j\}_j\sim_{\cS(\cQ_{\cV_1},u)} \{ M_j'\}_j$.
\end{theorem}
Therefore, 
when our states are limited into $\cM(\cQ_{\cV_1},u)$,
we can restrict our measurements into elements of $\cM(\cQ_{\cV_1},u)$.

When our measurement is limited into $\cM(\cQ_{\cV_1},u)$,
any state can be written as an element of 
$\cS(\cQ_{\cV_1},u)$.
Hence, we have the following theorem.

\begin{theorem}\label{theorem:embed-state}
For any state $\rho \in \cS(\cQ_{\cV},u)$,
there exists a state $\rho' \in \cS(\cQ_{\cV_1},u)$
such that $\rho\sim_{\cM(\cQ_{\cV_1},u)} \rho'$.
\end{theorem}
Therefore, 
when our measurements are limited into $\cM(\cQ_{\cV_1},u)$,
we can restrict our states into elements of $\cS(\cQ_{\cV_1},u)$.

\subsection{Canonical embedding map}
\label{section:CanonicalEmbedding}

Next, we define the canonical embedding map
and see that the SRR entropy is preserved by canonical embedding maps.
As a result, we give another proof of Stein's Lemma if there exists a canonical embedding map into quantum theory (Theorem~\ref{TT3}).

We say that a linear mapt $\phi$ from an Jordan algebra $\cV_1$ to another Jordan algebra $\cV_2$ is a Jordan homomorphsm
when $\phi(x)\circ_2\phi(y)=\phi(x\circ_1 y)$ holds for any $x,y\in\cV_1$ with Jordan products $\circ_1$ on $\cV_1$ and $\circ_2$ on $\cV_2$.

Given a finite-dimensional Hilbert space ${\cal H}$,
we denote the set of Hermitian matrices by
${\cal B}_H({\cal H})$.
Given a Jordan algebra ${\cal V}$ with a canonical inner product
$\langle ~,~\rangle_{{\cal V}}$,
an embedding map $\phi$ from
${\cal V}$ to ${\cal B}_H({\cal H})$ is called 
a canonical embedding map
when $\phi$ is a Jordan homomorphsm
and
the Jordan subalgebra $\phi({\cal V})$ 
is a canonical Jordan subalgebra of 
${\cal B}_H({\cal H})$ with the inner product defined by the trace.
We define the dual map $\phi^*:
{\cal B}_H({\cal H})\to {\cal V}$ as
\begin{align}
\langle \phi^*(X),Y\rangle_{{\cal V}}
=\Tr X \phi(Y)
\end{align}
for $X \in {\cal B}_H({\cal H})$ and $Y\in {\cal V}$.
We consider
the following sets of states $\cS(\cQ_\cV,u)$ and $\cS(\cQ_{\phi(\cV)},I)$ of Jordan subalgebras ${\cal V}$ and
$\phi({\cal V})$ with the inner product defined by the trace.

Then, we obtain the following theorem about the equivalence of SRR entropy and relative entropy by cannonical embedding
by applying Lemma~\ref{lemma:limofSandwich}

\begin{theorem}\label{TT3}
Given a canonical embedding map $\phi$ from
${\cal V}$ to ${\cal B}_H({\cal H})$,
two states $\rho,\sigma \in \cS(\cQ_{\phi(\cV)},I)$ satisfy
\begin{align}
D(\rho\|\sigma)&=D(\phi^*(\rho)\|\phi^*(\sigma)) \label{eq:TT3-1}\\
\underline{D}_{1+s}(\rho\|\sigma)
&=\underline{D}_{1+s}(\phi^*(\rho)\|\phi^*(\sigma)) \label{eq:TT3-2}
\end{align}
for $s>0$.
Also, the map $\phi^*$ gives one-to-one relation between
$\cS(\cQ_{\phi(\cV)},I)$ and $\cS(\cQ_\cV,u)$.
That is, there is a map $\psi:\cS(\cQ_\cV,u) \to \cS(\cQ_{\phi(\cV)},I)$ such that
$ \phi^* \circ \psi $ is the identity map.
Hence, two states $\rho',\sigma' \in \cS(\cQ_\cV,u)$ satisfy
\begin{align}
D(\rho'\|\sigma')
&=D(\psi(\rho')\|\psi(\sigma')) \label{eq:TT3-3}\\
\underline{D}_{1+s}(\rho'\|\sigma')
&=\underline{D}_{1+s}(\psi(\rho')\|\psi(\sigma')) \label{eq:TT3-4}
\end{align}
for $s>0$.
\end{theorem}

In order to apply Lemma~\ref{lemma:limofSandwich} for cannonical embedding $\phi$,
we need to define the cannonical embedding $\phi_n$ from $n$-composite system $\cV^{\otimes n}$ to $\cB_H(\cH^{\otimes n})$ induced by $\phi$.
For $x_n:=\bigotimes_{i=1}^n x_i\in\cV^{\otimes n}$, we define $\phi_n(x_n):=\bigotimes_{i=1}^n \phi(x_i)$.
Because of our choice of composite Jordan algebra (Definition~\ref{definition:CompositesysteminEJA}),
the map $\phi_n$ is cannonical embedding from $\cV^{\otimes n}$ to $\cB_H(\cH^{\otimes n})$.

\begin{proof}[Proof of Theorem~\ref{TT3}]
	First, we prove the equations \eqref{eq:TT3-1} and \eqref{eq:TT3-2}.
	Lemma~\ref{lemma:limitofRenyies}, i.e., the following relation,
	implies that we only have to prove the case of SRR entropy for any $s>0$:
	\begin{align}
		\lim_{s\to 0}\underline{D}_{1+s}(\rho||\sigma) &=D(\rho||\sigma).
    	\end{align}
	Now, we show the relation \eqref{eq:TT3-2} by applying Lemma~\ref{lemma:limofSandwich}, i.e., the following relation:
	\begin{align}
		\underline{D}_{1+s}(\rho||\sigma)=\lim_{n\to\infty}\frac{1}{n}\max_{M^n}D_{1+s}(P^{M^n}_{\rho^{\otimes n}}||P^{M^n}_{\sigma^{\otimes n}}),\quad s>0.
	\end{align}
	Then, we obtain the following relation:
	\begin{align}
		\underline{D}_{1+s}(\phi^*(\rho)\|\phi^*(\sigma))
		=&\lim_{n\to\infty}\frac{1}{n}\max_{M^n\in \cM(\cQ_{\cV^{\otimes n}},I)}D_{1+s}(P^{M^n}_{\phi^*(\rho)^{\otimes n}}||P^{M^n}_{\phi^*(\sigma)^{\otimes n}})\\
		=&\lim_{n\to\infty}\frac{1}{n}\max_{M^n\in \cM(\cQ_{\cV^{\otimes n}},I)}D_{1+s}(P^{M^n}_{\phi^*_n(\rho^{\otimes n})}||P^{M^n}_{\phi^*_n(\sigma^{\otimes n})})\\
		=&\lim_{n\to\infty}\frac{1}{n}\max_{M^n\in \cM(\cQ_{\cV^{\otimes n}},I)}D_{1+s}(P^{\phi_n(M^n)}_{\rho^{\otimes n}}||P^{\phi_n(M^n)}_{\sigma^{\otimes n}})\\
		=&\lim_{n\to\infty}\frac{1}{n}\max_{M^n\in \cM(\cQ_{\phi_n(\cV^{\otimes n})},I)}D_{1+s}(P^{M^n}_{\rho^{\otimes n}}||P^{M^n}_{\sigma^{\otimes n}})\label{eq:TT3-proof-1}
	\end{align}
	Now, we apply Theorem~\ref{theorem:embed-measurement} for the case $\cV=\cV^{\otimes n}$ and $\cV_1=\phi_n(\cV^{\otimes n})$,
	and therefore, we can replace $\cM(\cQ_{\phi_n(\cV^{\otimes n})},I)$ with $\cM(\cQ_{\cV^{\otimes n}},I)$ in the maxmization in \eqref{eq:TT3-proof-1}.
	Finally, we apply Lemma~\ref{lemma:limitofRenyies} again, and as a result, we obtain the following desired relation:
	\begin{align}
		\underline{D}_{1+s}(\phi^*(\rho)\|\phi^*(\sigma))
		=&\lim_{n\to\infty}\frac{1}{n}\max_{M^n\in \cM(\cQ_{\cV^{\otimes n}},I)}D_{1+s}(P^{M^n}_{\rho^{\otimes n}}||P^{M^n}_{\sigma^{\otimes n}})\label{eq:TT3-proof-1}\\
		=&\underline{D}_{1+s}(\rho\|\sigma).
	\end{align}

	Next, we prove the equations \eqref{eq:TT3-3} and \eqref{eq:TT3-4}.
	We choose $\psi$ as 
	\begin{align}
		\Tr \psi(x)\phi(y)=\langle x,y\rangle_\cV,
	\end{align}
	for any $x,y\in\cV$, and therefore, we obtain the following relation for any $x,y\in \cV$
	\begin{align}
		\langle \phi^\ast\circ\psi(x),y\rangle_\cV=\Tr \psi(x),\phi(y)=\langle x,y\rangle_\cV,
	\end{align}
	which implies the map $\phi^\ast\circ\psi$ is the identity map.
	Then, the equations \eqref{eq:TT3-3} and \eqref{eq:TT3-4} are shown by the equations \eqref{eq:TT3-1} and \eqref{eq:TT3-2} as follows:
	\begin{align}
		D(\psi(\rho')\|\psi(\sigma'))=D(\phi^\ast\circ\psi(\rho')\|\phi^\ast\circ\psi(\sigma'))
		=D(\rho'\|\sigma'),\\
		\underline{D}_{1+s}(\psi(\rho')\|\psi(\sigma'))=\underline{D}_{1+s}(\phi^\ast\circ\psi(\rho')\|\phi^\ast\circ\psi(\sigma'))
		=\underline{D}_{1+s}(\rho'\|\sigma').
	\end{align}
\end{proof}

Theorem~\ref{TT3} gives another proof of Stein's lemma in EJAs
through a canonical embedding map $\phi$ from
${\cal V}$ to ${\cal B}_H({\cal H})$.

\color{black}

\begin{theorem}\label{theorem:stein-another}
When an EJA ${\cal V}$ 
satisfies the conditions of Theorem \ref{TT3},
two states $\rho',\sigma' \in {\cal S}$ satisfy
\begin{align}
\lim_{n\to \infty}-\frac{1}{n}\log \beta_\epsilon^n({\rho'}^{\otimes n}\|
{\sigma'}^{\otimes n})
=D(\rho'\|\sigma').\label{eq:stein-another}
\end{align}
\end{theorem}

\begin{proof}
	By applying Theorem \ref{TT3},
	we obtain a map $\psi:\cS(\cQ_\cV,u) \to \cS(\cQ_{\phi(\cV)},I)$ satisfies \eqref{eq:TT3-3}.
	Also, simiarly to the proof of Theorem~\ref{TT3},
	we can conclude $\beta_\epsilon^n(\rho'\|\sigma')=\beta_\epsilon^n(\psi(\rho')\|\psi(\sigma'))$ as follows:
	\begin{align}
		\beta^n_\epsilon(\psi(\rho')||\psi(\sigma'))=&\min_{T_n\in\cM(\cQ_{\phi(\cV),I})}\{\langle\psi(\sigma')^{\otimes n},T_n\rangle|\langle\psi(\rho')^{\otimes n},u-T_n\rangle\le \epsilon\}\\
		=&\min_{T_n\in\cM(\cQ_{\phi(\cV),I})}\{\langle\psi_n(\sigma'^{\otimes n}),T_n\rangle|\langle\psi_n(\rho'^{\otimes n}),u-T_n\rangle\le \epsilon\}\\
		=&\min_{T_n\in\cM(\cQ_{\phi(\cV),I})}\{\langle \sigma'^{\otimes n},\psi^\ast_n(T_n)\rangle|\langle\rho'^{\otimes n},u-\psi_n^\ast(T_n)\rangle\le \epsilon\}\\
		\stackrel{(a)}{=}&\min_{T_n\in\cM(\cQ_{\cV,u})}\{\langle \sigma'^{\otimes n},T_n\rangle|\langle\rho'^{\otimes n},u-T_n\rangle\le \epsilon\}\\
		=&\beta_\epsilon^n(\rho'\|\sigma'), \label{eq:proof-stein-another}
	\end{align}
	where the map $\psi_n:\cV^{\otimes n}\to \cB_H(\cH^{\otimes n})$ and $\psi_n^\ast:\cB_H(\cH^{\otimes n})\to\cV^{\otimes n}$ are the composite map and its dual map induced by $\psi$.
	The equation $(a)$ holds because of Theorem~\ref{theorem:embed-measurement} and the fact that $\psi_n^\ast\circ \phi$ is a canonical embedding map.
	By combinating \eqref{eq:TT3-3}, \eqref{eq:proof-stein-another}, and Stein's Lemma in quantum theory $\cB_H(\cH)$,
	we conclude the equation \eqref{eq:stein-another}.
\end{proof}

\subsection{Lorentz Type}
\label{section:LorentzType}

Next, we show that 
Lorentz type, i.e.,
Jordan algebra with Lorentz cone, satisfies the conditions of Theorem \ref{TT3}. 
Actually, it has already known in \cite{BGW2020}.
However, we give a new relation between Lorentz type and fermion annihilation and creation operators
and recover the construction in \cite{BGW2020} by our new relation and Jordan-Wigner transformation \cite{JW1928}.

\subsubsection{Formulation}

We consider $d+1$-dimensional vector space $\cV_{d+1}$.
Its element $v$ has the form $v=(v_0,v_1, \ldots, v_d)$.
The Jordan product $v \circ v'$ between $v,v'$ is given as
$(\sum_{j=0}^d v_j v_j', v_0v_1'+v_1v_0', \ldots, v_0v_d'+v_dv_0')$.
When we denote $(0,\underbrace{0,\ldots, 0}_{j-1},1,0,\ldots,0)$
by $e_j$.
Hence, it is sufficient to check the following condition for Lorentz type.
\begin{align}
e_0\circ e_0= e_0,~
e_j\circ e_0= e_j,~
e_j\circ e_{j'}= \delta_{j,j'} e_0 \label{CCM}
\end{align}
for $1 \le j,j'\le d$.
We also consider the inner product
$\langle e_j,e_{j'}\rangle:=\delta_{j,j'}$.

We denote Lorentz cone of $d+1$-dimensional vector space $\cV_{d+1}$
by 
${\cal Q}(\cV_{d+1})$, which is written as
\begin{align}
{\cal Q}(\cV_{d+1})= \left\{c_{0}e_0+ \sum_{j=0}^{2n}c_j e_j\middle|
c_{-1}\ge \sqrt{\sum_{j=0}^{2n}c_j^2}\right\}.
\end{align}
We denote its state space by 
${\cal S}(\cV_{d+1})$.

\subsubsection{Relation with fermion}
We consider fermion annihilation and creation operators
$a_k$ and $a_k^\dagger$ with $k=1, \ldots, n$ with the following commutation relations.
\begin{align}
2 a_k\circ a_{k'}^\dagger &=\delta_{k,k'} \label{XB1}\\
 a_k\circ a_{k'}&=0,\quad
 a_k^\dagger \circ a_{k'}^\dagger=0 \label{XB2}
\end{align}

We define $\lambda_{2k-1}:= a_k+a_k^\dagger$ and 
$\lambda_{2k}:= i(a_k-a_k^\dagger)$.
Then, we have
\begin{align}
\lambda_{2k-1}\circ \lambda_{2k}=&-i a_k\circ a_k^\dagger+ i a_k^\dagger\circ a_k=0 \\
\lambda_{2k-1}\circ \lambda_{2k-1}=&a_k\circ a_k^\dagger+ a_k^\dagger\circ a_k=I \\
\lambda_{2k}\circ \lambda_{2k}
=&-i^2 a_k\circ a_k^\dagger-i^2 a_k^\dagger\circ a_k \\
=&a_k\circ a_k^\dagger+ a_k^\dagger\circ a_k=I .
\end{align}
Also, for $k \neq k'$, we have
\begin{align}
\lambda_{2k-1}\circ \lambda_{2k'-1}&=0 \\
\lambda_{2k-1}\circ \lambda_{2k'}&=0 \\
\lambda_{2k-1}\circ \lambda_{2k'-1}&=0.
\end{align}
Therefore, 
the operators
$I, \lambda_1, \ldots, \lambda_{2n}$ generate
a Clifford algebra, i.e., a Lorentz type
$\cV_{2n+1}$.

However, in this system, we have other observables
$a_k a_{k'}+a_{k'}^\dagger a_{k}^\dagger$,
$i(a_k a_{k'}-a_{k'}^\dagger a_{k}^\dagger)$
 for any $k, k'$.
They cannot be written as linear combination of 
$I, \lambda_1, \ldots, \lambda_{2n}$.
That is,
when we are interested in 
the real and imaginary parts of the fermion annihilation and creation 
operators,
our system is written by Lorentz type
$\cV_{2n+1}$.

\subsubsection{Canonical embedding map with $d=2n$}\label{S3}
To find an canonical embedding map of a Lorentz type,
we employ
the above relation between 
fermion and $V_{2n+1}$
and
Jordan–Wigner transformation \cite{JW1928}, which show how to describe
$n$-mode fermion in $n$ qubits. 

We set ${\cal H}:= (\mathbb{C}^2)^{\otimes n}$.
We prepare the following notations.
\begin{align}
\sigma_0&:=
\left(
\begin{array}{cc}
1 & 0 \\
0 & 1
\end{array}
\right),~
\sigma_1:=
\left(
\begin{array}{cc}
0 & 1 \\
1 & 0
\end{array}
\right), \\
\sigma_2&:=
\left(
\begin{array}{cc}
0 & -i \\
i & 0
\end{array}
\right),~
\sigma_3:=
\left(
\begin{array}{cc}
1 & 0 \\
0 & -1
\end{array}
\right).
\end{align}
We define the operator $a_{k,JW}$ as
\begin{align}
a_{k,JW}= \sigma_0^{\otimes k-1}\otimes\frac{1}{2}(\sigma_{1}-i \sigma_{2}) \otimes \sigma_{3}^{\otimes n-k}. 
\end{align}
Jordan–Wigner transformation \cite{JW1928} gives the operators
$a_{k,JW}$ and 
$a_{k,JW}^\dagger$ with $k=1, \ldots,n$
satisfy the conditions \eqref{XB1} and \eqref{XB2} 
Then, the operators
\begin{align}
\lambda_{2k-1,JW}:&= a_{k,JW}+a_{k,JW}^\dagger \nonumber\\
&=\sigma_0^{\otimes k-1}\otimes \sigma_1\otimes \sigma_3^{\otimes n-k}\\
\lambda_{2k,JW}:&= i(a_{k,JW}-a_{k,JW}^\dagger) \nonumber\\
&=\sigma_0^{\otimes k-1}\otimes \sigma_2\otimes \sigma_3^{\otimes n-k}
\end{align}
satisfy the condition \eqref{CCM}
for Jordan algebra with Lorentz cone.
Therefore, the following map $\phi_{2n+1}:
V_{2n+1} \to {\cal B}_H({\cal H})$
is a Jordan homomorphsm.
\begin{align}
\phi_{2n+1}(e_{j})= \lambda_{j,JW},\quad
\phi_{2n+1}(e_0)=I
\end{align}
for $j=1,\ldots, 2n$.
Then, we have
\begin{align}
{\cal Q}(\phi_{2n+1}({\cal V}_{2n+1})) 
= \left\{c_{-1}I+ \sum_{j=1}^{2n}c_j \lambda_{j,JW}\middle|
c_{-1}\ge \sqrt{\sum_{j=1}^{2n}c_j^2}\right\}.
\end{align}
Since
\begin{align}
\Tr \frac{1}{2^n}\lambda_{j,JW} \lambda_{j',JW}=\delta_{j,j'},
\Tr \frac{1}{2^n}I \lambda_{j',JW}=0,
\Tr \frac{1}{2^n}I =1\label{HJK1}
\end{align}
for $j,j'=1, \ldots, 2n$,
we have
\begin{align}
&{\cal Q}(\phi_{2n+1}({\cal V}_{2n+1})) \notag\\
=&
\{X \in \phi_{2n+1}({\cal V}_{2n+1}) | 
\Tr X Y \ge 0 ,\forall Y \in \cQ(\phi_{2n+1}({\cal V}_{2n+1}))\}.
\end{align}
Hence, the embedding map $\phi_{2n+1}$ is a canonical embedding map.

\subsubsection{Canonical embedding map with $d=2n+1$}\label{S4}
We choose $\lambda_{2n+1,JW}:= \sigma_{3}^{\otimes n+1}$.
Then, we have
\begin{align}
\lambda_{2n+1,JW} \circ \lambda_{2n+1,JW}&=I \\
\lambda_{2n+1,JW} \circ \lambda_j&=0 
\end{align}
for $j= 1, \ldots, 2n$.
Therefore, the following map $\phi_{2n+2}:
V_{2n+2} \to {\cal B}_H({\cal H})$
is a Jordan homomorphsm.
\begin{align}
\phi_{2n+2}(e_{j})= \lambda_{j,EJW},\quad
\phi_{2n+2}(e_0)=I
\end{align}
for $j=1,\ldots, 2n+1$.
Since \eqref{HJK1} holds for $j,j'=1, \ldots, 2n+1$,
the embedding map $\phi_{2n+2}$ is a canonical embedding map.

Here, we compare the discussion by Barnum et.al. \cite{BGW2020}.
In the case of $d=2n$, our embedding map is essentially same as the equations (8-10) in \cite{BGW2020}.
In the case of $d=2n+1$, Ref. \cite{BGW2020} embedds the Lorentz type into the set of Hermitian matrices on 
Hilbert space with twice dimension
instead of taking an additional element $\sigma_3^{\otimes n+1}$ 
without considering the relation with fermion.

As a result, we embed Lorentz type into low dimensional quantum theory.

\subsection{Quaternion type}
\label{section:QuaternionType}

Next, we show that 
a quaternion type, i.e., a Jordan algebra with Hermitian matrices on quaternion, satisfies the conditions of Theorem \ref{TT3}. 
Actually, it has already known in \cite{BGW2020}.

\subsubsection{Formulation}

We denote $\mathbb{H}$ as the quaternion.
For a matrix $X$ with $\mathbb{H}$-valued entries, we say that $X$ is Hermitian
if $X^\dag=X$, where $X^\dag$ denotes Hermitian conjugate with the conjugation on $\mathbb{H}$.
Let $\cV$ be the vector space of $d\times d$ Hermitian matrices on $\mathbb{H}$,
and we define a Jordan product $\circ$ for $X,Y\in\cV$ as follows:
\begin{align}
	X\circ Y:=\frac{1}{2}(XY+YX).
\end{align}
This algebra composes an EJA \cite{Faraut},
and we call it \textit{quaternion type with dimension $d$}.
We denote quaternion type with dimension $d$ as $\mathrm{Herm}(d,\mathbb{H})$.

\subsubsection{Canonical Embedding map}

To find a canonical embedding map of a quaternion type,
we define a map $\phi_0:\mathbb{H}\to\mathrm{M}_2(\mathbb{C})$ as
\begin{align}
	\phi_0(a+b\bm{i}+c\bm{j}+d\bm{k}):=
	\begin{pmatrix}
		a+b\bm{i} & c+d\bm{i}\\
		-c+d\bm{i} & a-b\bm{i}\\
	\end{pmatrix}.
\end{align}
By definition, the map $\phi_0$ is linear map,
and moreover, $\phi_0$ preserve the matrix product, i.e., $\phi_0(x_1x_2)=\phi_0(x_1)\phi_0(x_2)$ holds.

By using $\phi_0$ on each entry, we define a map $\phi_{\mathbb{H}}:\mathrm{Herm}(d,\mathbb{H})\to\cB_H(\mathbb{C}^{2d})$ as follows.
Let $A=(a_{ij})_{ij}$ be a $d\times d$ Hermitian matrix with $\mathbb{H}$-valued entries. We define $\phi_{\mathbb{H}}(A)$ as the $2d\times 2d$ $\mathbb{C}$-vauled Block matrix $X=(X_{ij})_{ij}$ with $\phi_0(a_{ij})=X_{ij}$.
It is easy to show from the above definition of $\phi_{\mathbb{H}}$ that $\phi_{\mathbb{H}}(X)$ is a Hermitian matrix,
which implies that the range of $\phi_{\mathbb{H}}$ is contained by $\cB_H(\mathbb{C}^{2d})$.

Besides,
because $\phi_0$ is linear map preserving the matrix products,
the map $\phi_{\mathbb{H}}$ satisfies
\begin{align}
	\phi_{\mathbb{H}}(AB)=\left(\sum_k \phi_0(a_{ik}b_{kj})\right)_{ij}
	=\left(\sum_k X_{ik}Y_{kj}\right)_{ij}=XY=\phi_{\mathbb{H}}(A)\phi_{\mathbb{H}}(B),
\end{align}
for $A:=(a_{ij})_{ij},B:=(b_{ij})_{ij}\in\mathbb{H}^{d}$ and $X:=(X_{ij})_{ij},Y:=(Y_{ij})_{ij}\in\cB_H(\mathbb{C}^{2d})$.
In other words, 
the map $\phi_{\mathbb{H}}$ preserves the matrix products.
Because both the Jordan products and the inner produt induced by the trace is defined by the matrix product,
the map $\phi_{\mathbb{H}}$ is a Jordan homomorphism from $\mathrm{Herm}(d,\mathbb{H})$ to $\cB_H(\mathbb{C}^{2d})$ and the trace is a cannonical inner product.
As a result, $\phi_\mathbb{H}$ is a cannonical embedding map of the quaternion type.

\section{Conclusion}\label{section:conclusion}

In this paper,
we have dealt with EJAs and models of GPTs associated with EJAs.
Through mathematical properties of EJAs, we have established information quantities and information theoretical tools in the associated models.
By analyzing informtion quantities by information theoretical tools, we have obtained important inequailties for the proof of Stein's Lemma.
As a result, we have proven Stein's Lemma in the model associated with any EJA as the same statement as that of quantum and classical theories.
This result implies that the structure of EJAs is the mathematically essential structure for the relation between the exponent of hypothesis testing and relative entropy.
Moreover, we have discussed embedding from EJAs into quantum theory, which have given another proof of Stein's Lemma through the inequalities of information quantities that we have established.

Finally, we give two open problems.
The first problem is to prove other results of typical topics of quantum information theory even in EJAs.
For example, we can consider a generalization of C-Q and Q-Q channels and information transsmission with such channels.
Even for the task and even in EJAs, can we obtain the same results, the relation between the limit performance and informaiton quantities.
The second problem is to prove Stein's Lemma for any compositions other than the canonical composition in this paper. Even we assume the structure of EJAs for composition of GPTs, there are other compositions \cite{BGW2020}. It is still open whether Stein's Lemma holds in any composition.

\acknowledgments
HA was supported by JSPS KAKENHI Grant Number 25KJ0043.
M.H. was supported in part by the National Natural Science Foundation of China (Grants no. 62171212), and the General R$\&$D Projects of $1+1+1$ CUHK-CUHK(SZ)-GDST Joint Collaboration Fund (Grant No. GRDP2025-022).

\bibliographystyle{unsrt}
\bibliography{bibtex_EJA}

\appendix
\section{Appendix}
\renewcommand{\thesubsection}{\Alph{section}.\arabic{subsection}}

\subsection{Proof of concepts in Euclidean Jordan algebra}\label{appendix:conceptsinEJA}
At first, we introduce a homomorphism and an isomorphism in an Euclidean Jordan algebra.
We use these morphisms in order to show that a Classical system is isomorphic
to an Euclidean Jordan algebra where its all elements are simultaneous spectral decomposable.

\begin{definition}[Homomorphism and Isomorphism\cite{Mccrimmon}{[Definition1.2.2]}]
    Let $\cV,\cV'$ be Euclidean Jordan algebras.
    A linear map $f:\cV\to \cV'$ is called a Jordan homomorphism if $f$ satisfies the following condition for all $x,y\in \cV$,
    \begin{align}
        f(x\circ y)=f(x)\circ'f(y),
    \end{align}
    where $\circ'$ is the Jordan product in $\cV'$.
    In addition, if this map $f$ is a bijection, $f$ is called as a Jordan isomorphism.
    Moreover, if there exist a Jordan isomorphism $f:\cV\to \cV'$,
    we call that $\cV$ is isomorphic to $\cV'$.
\end{definition}

\begin{remark}
    A linear function $f:\cV\to \cV'$ is Jordan homomorphism if and only if the linear function $f$ satisfies $f(x^2)=f(x)^2$.
    This is shown by calculating $f((x+y)^2-x^2-y^2)$ using linearity of $f$.\cite{Mccrimmon}
    Originally, these morphisms are given as morphisms between two (non associative) commutative rings with $R$ modules
    because it does not need the conditions of (J2) and (J3) of Definition~\ref{definition:EuclideanJordanalgebra}
    ,where $R$ is a ring.
\end{remark}

The following Lemma is important for us to consider the correspondence of the space $\cV$ to the classical system.

\begin{lemma}[characterization of Classical system]\label{lemma:characterizationofClassicalsystem}
    If the all of elements in $\cV$ are classically,
    $\cV$ is isomorphic to the classical system.
\end{lemma}

\begin{proof}[Proof of characterization of Classical system Lemma~\ref{lemma:characterizationofClassicalsystem}]
    If all elements $x,y\in \cV$ are classically, from Theorem~\ref{theorem:simultaneousspectral}, all of elements $x,y\in \cV$ have a simultaneous spectral decomposition.
    We fix $x$ as $x=\sum_i\lambda_i' e_i$ where $\{e_i\}_{i=1}^n$ is the complete system of orthogonal primitive idempotents, $\lambda_i'$ are all distinct. Then, all elements $y\in \cV$ are decomposed as $y=\sum_{i=1}^n \mu_i' e_i$.
    Then, we construct a following homomorphism between an Euclidean Jordan algebra $\cV$ and Classical system, that is $f:V\to \mathbb{R}^n$,
    where $u_i$ takes $1$ in $i$th element and $0$ in others.
    \begin{align}
        f(e_i)=u_i\quad (i=1,\ldots,n).
    \end{align}
    Then,
    \begin{align}
        f(x\circ y)=\sum_{i=1}^n\lambda_i'\mu_i'u_i=f(x)\circ f(y)
    \end{align}
    Therefore, $f:V\to \mathbb{R}^n$ is a homomorphism. In addition, from $f(e_i)=u_i$, this is surjective. And from $\mathrm{Ker}f=\{0\}$, this is injective. Therefore this homomorphism is bijective, so this is isomorphism and $\cV$ is isomorphic to the Classical system.
\end{proof}

\begin{proof}[Proof of Lemma~\ref{lemma:projectionPeirce}]
    Let $x=\sum_i x_i+\sum_{i<j}x_{i,j}$ be the Peirce decomposition with CSOI $\{c_i\}$.
    Then, we calculate as follows:
    \begin{align}
        P_{c_i}(x)&=2L_{c_i}^2(x)-L_{c_i^2}(x)\\
        &=2L_{c_i}(x_i+\sum_{j:i<j}\frac{1}{2}x_{i,j})-(x_i+\sum_{j:i<j}\frac{1}{2}x_{i,j})=x_i
    \end{align}
\end{proof}

\begin{proof}[Proof of Lemma~\ref{lemma:orthogonalPeirce}]
    Applying identity $L_{x^2\circ y}-L_{x^2}L_{y}=2(L_{x\circ y}-L_xL_y)L_x$ for $x,y\in \cV$.
    The element $c_i$ in place of $x$ and the element $x$ in place of $y$, then applying this identity to the element $y$ we obtain
    \begin{align}
        L_{c_i}(x\circ y)-x\circ y=0,
    \end{align}
    where $x\in \cV(c_i,1), y\in \cV(c_i,0),i\neq j$.
    Moreover, exchange $x$ and $y$ as $x\in \cV(c_i,0)$ and $y\in \cV(c_i,1)$, we obtain
    \begin{align}
        L_{c_i}(xy)=0.
    \end{align}
    Therefore, we obtain $\cV(c_i,1)\circ \cV(c_j,1)\subset \cV(c_i,1)\circ \cV(c_i,0)=0$.
\end{proof}

\begin{proof}[Proof of Lemma~\ref{lemma:newCSOI}]
    Let $\{c_i\}$ be a CSOI.
    Considering Peirce decomposition by $\{c_i\}$, the space $\cV(i,1)$ is subalgebra of $\cV$ because $L_{c_i}(x\circ y)-x\circ y=0$ hols for $x,y\in \cV(i,1)$.
    Then, applying Theorem~\ref{theorem:Spectraldecomposition} to an element $x_i\in \cV(i,1)$, we obtain a family $\{c_{i,j}\}$ and coefficiences $\{\mu_j\}$ such that
    \begin{align}
        &\sum_{j=1}c_{i,j}=c_i,\\
        &\sum_{j=1}\mu_j c_{i,j}=x_i.
    \end{align}
\end{proof}


Now we prove the important lemma of quadratic form (Lemma~\ref{lemma:positivityofQradratic}).
\begin{proof}[Proof of \ref{lemma:positivityofQradratic}]
    At first, we will show this statement for $y\in \cQ_\cV$ and invertible $x\in \cV$.\\
    Suppose to $P_x(y)\notin \cQ_\cV$,we show by contradiction.\\
    For the element $y(t):=t y+(1-t)u\quad (t\in [0,1])$, $y(t)\in Q$ because of the convexity of $\cQ$. In particular, $y(t)$ is invertible in $t\in (0,1)$.
    We put on $z(t)=P_x(y(t))$. From $z(0)=x^2\in \cQ_\cV,z(1)=P_x(y)\notin \cQ_\cV$, $z(t)$ has a negative eigenvalue in $t=1$ and is positive in $t=0$.Hence, there exist $s,w\in \cV$ such that $L_{z(s)}(w)=0$ in $0<s\le 1$.
    Now we observe $P_{z(t)}$.$P_{z(t)}=2L_{z(t)}^2-L_{z(t)^2}$.The quantity $\langle P_{z(s)}(w),w\rangle <0$ because of $L_{z(t)^2}>0$ in $t\in (0,s)$.
    On the other hand, $P_{z(0)}=P_{x^2}=P_{x^2}>0$ holds. Hence, $P_{z(t)}$ has 0 eigenvalues in $0<t<s$.
This is a contradiction to invertibility of $z(t)$ in $0<t<s$.
    Therefore, $P_{x}(y)\in \cQ_\cV$.

    Next, we will show the statement for $y\in \cQ_\cV$ and $x\in \cV$.
    The element $x$ has finite eigenvalues because $\cV$ is a finite dimensional Euclidean Jordan algebra.
    Therefore, there exists $s,\epsilon$ such that $x+\epsilon u$ is invertible in $0<\epsilon <s$. \\
    So, if we take $\epsilon \to 0$ for $P_{x+\epsilon u}(y)\in \cQ_\cV$, then we obtain $P_{x}(y)\in \cQ_\cV$.
\end{proof}
By applying Lemma~\ref{lemma:positivityofQradratic}, now we prove self-duality of $\cQ_\cV$ (Lemma~\ref{lemma:selfduality}).
\begin{proof}[Proof of Lemma~\ref{lemma:selfduality}]
    We show that $\cQ_\cV\subset \cQ_\cV^*$.\\
    If $x=z^2\in \cQ_\cV$, $\langle y^2,z^2\rangle=\langle P_z(y^2),u\rangle$ for the $z\in \cV,y^2\in \cQ_\cV$.From Lemma~\ref{lemma:positivityofQradratic}, $P_z(y^2)\in \cQ_\cV$.Hence, there exists $w\in \cV$ such that $P_z(y^2)=w^2$.
    Therefore, $\langle y^2,z^2\rangle=\langle w^2,u\rangle=||w||^2\ge 0$.This means $x=z^2\in \cQ_\cV^*$.\\
    Next, we show that $\cQ_\cV^*\subset \cQ_\cV$.\\
    The quantity $\langle x,y^2\rangle\ge 0,\quad \forall y\in \cQ_\cV$ for $x\in \cQ_\cV^*$. In particular, for the spectral decomposition of $x$ as $x=\sum_i\lambda_i c_i$, $\lambda_i\ge 0 $ because of $c_i\in \cQ_\cV$.
    We can define $x=z^2$ as $z=\sum_i \sqrt{\lambda_i}c_i$.This means $x\in \cQ_\cV$.
\end{proof}

\subsection{Proofs of fundamental properties of entropies}
\label{appendix:fundamentalproperties}

\begin{proof}[Proof of Lemma~\ref{lemma:classicallyRenyies}]
    We will show that
    \begin{align}
        (P_{\sigma^{\frac{-s}{2(1+s)}}}(\rho))^{1+s}=\rho^{1+s}\circ \sigma^{-s}
    \end{align}
    Let $\rho=\sum_i\lambda_i c_i,\sigma=\sum_i \mu_i d_i$ be a simultaneous spectral decomposition by Theorem~\ref{theorem:simultaneousspectral}, where $d_i\in \oplus_j V(c_j,1)$.
    Then, by applying $\{d_j\}$ to Theorem~\ref{theorem:simultaneousspectral}, we obtain
    and we write down $\rho,\sigma$ as
    \begin{align}
        \rho   & =\sum_{i,j} \lambda'_{j} d_j. \\
        \sigma & =\sum_{i,j} \mu_{j}d_j.
    \end{align}
    \begin{align}
        P_{\sigma^{\frac{-s}{2(1+s)}}}(\rho) & =2L_{\sigma^{\frac{-s}{2(1+s)}}}L_{\sigma^{\frac{-s}{2(1+s)}}}(\rho)-L_{\sigma^{\frac{-s}{1+s}}}(\rho)                           \\
                                            & =\sum_{j}2\lambda'_{j}L_{\sigma^{\frac{-s}{2(1+s)}}}L_{\sigma^{\frac{-s}{2(1+s)}}}(d_j)-L_{\sigma^{\frac{-s}{1+s}}}(d_j) \\
                                            & =\sum_{j} \lambda'_{j} {\mu}_{j}^{\frac{-s}{1+s}}d_j
    \end{align}
    Hence,
    \begin{align}
        (P_{\sigma^{\frac{-s}{2(1+s)}}}(\rho))^{1+s}=\sum_{j}{\lambda'}_{j}^{1+s}{\mu}_{j}^{-s}d_i=\rho^{1+s}\circ \sigma^{-s}
    \end{align}
\end{proof}

\begin{proof}[Proof of Lemma~\ref{theorem:additivity}]
    In the case of the Relative entropy, we can show as follows:
    \begin{align}
        D(\rho_1\otimes \rho_2||\sigma_1\otimes \sigma_2)=\mathrm{tr}\rho_1\otimes \rho_2(\log\rho_1\otimes \rho_2-\log\sigma_1\otimes \sigma_2).
    \end{align}
    Here, we apply the relations $\log\rho_1\otimes \rho_2=\log(\rho_1\otimes u\circ u\otimes \rho_2)=\log (\rho_1\otimes u)+\log (u\otimes \rho_2)$, $\log u\otimes \rho=u\otimes \log\rho$, then we obtain
    \begin{align}
          & D(\rho_1\otimes\rho_2 ||\sigma_1\otimes \sigma_2) =\mathrm{tr}\rho_1\otimes \rho_2(\log\rho_1\otimes u+u\otimes \log \rho_2-\log\sigma_1\otimes u-u\otimes \log\sigma_2) \\
        = & \mathrm{tr}(\rho_1\circ\log\rho_1-\log\rho_1\circ\sigma_1)+\mathrm{tr}(\rho_2\circ\log\rho_2-\log\rho_2\circ\sigma_2)=D(\rho_1||\sigma_1)+D(\rho_2||\sigma_2).
    \end{align}

    In the case of Petz Relative R\'{e}nyi entropy, we can show as follows:
    \begin{align}
        \mathrm{tr}(\rho_1\otimes\rho_2)^{1+s}\circ(\sigma_1\otimes \sigma_2)^{-s}=\mathrm{tr}(\rho_1^{1+s}\otimes\rho_1^{1+s})\circ(\sigma_1^{-s}\otimes\sigma_2^{-s})=(\mathrm{tr}\rho_1^{1+s}\sigma_1^{-s})(\mathrm{tr}\rho_2^{1+s}\sigma_2^{-s}).
    \end{align}
    holds. Hence, we take the log both sides, then
    \begin{align}
         & \phi(-s|\rho_1\otimes \rho_2||\sigma_1\otimes \sigma_2)=\log(\mathrm{tr}\rho_1^{1+s}\sigma_1^{-s})(\mathrm{tr}\rho_2^{1+s}\sigma_2^{-s})=\phi(-s|\rho_1||\sigma_1)+\phi(-s|\rho_2||\sigma_2).
    \end{align}
    Therefore, we divide the both sides by $s$, then we obtain
    \begin{align}
        D_{1+s}(\rho_1\otimes \rho_2||\sigma_1\otimes \sigma_2)=D_{1+s}(\rho_1||\sigma_1)+D_{1+s}(\rho_2||\sigma_2).
    \end{align}

    In the case of Sandwiched Relative R\'{e}nyi entropy, we can show as follows:
    \begin{align}
        P_{(\sigma_1\otimes \sigma_2)^{\frac{-s}{2(1+s)}}}(\rho_1\otimes \rho_2)=P_{\sigma_1^{\frac{-s}{2(1+s)}}\otimes \sigma_2^{\frac{-s}{2(1+s)}}}(\rho_1\otimes \rho_2)=P_{\sigma_1^{\frac{-s}{2(1+s)}}}(\rho_1)\otimes P_{\sigma_2^{\frac{-s}{2(1+s)}}}(\rho_2).
    \end{align}
    We take the power both sides by $1+s$, then we obtain
    \begin{align}
        \left(P_{\sigma_1^{\frac{-s}{2(1+s)}}}(\rho_1)\otimes P_{\sigma_2^{\frac{-s}{2(1+s)}}}(\rho_2)\right)^{1+s}=\left(P_{\sigma_1^{\frac{-s}{2(1+s)}}}(\rho_1)\right)^{1+s}\otimes \left(P_{\sigma_2^{\frac{-s}{2(1+s)}}}(\rho_2)\right)^{1+s}.
    \end{align}
    We take the trace and the log, then we ontain
    \begin{align}
          & \tilde{\phi}(-s|\rho_1\otimes \rho_2||\sigma_1\otimes \sigma_2)                                                           \\
        = & \log\mathrm{tr}\left(P_{\sigma_1^{\frac{-s}{2(1+s)}}}(\rho_1)\right)^{1+s}\otimes \left(P_{\sigma_2^{\frac{-s}{2(1+s)}}}(\rho_2)\right)^{1+s}         \\
        = & \log\mathrm{tr}\left(P_{\sigma_1^{\frac{-s}{2(1+s)}}}(\rho_1)\right)^{1+s}+\log\mathrm{tr}\left(P_{\sigma_2^{\frac{-s}{2(1+s)}}}(\rho_2)\right)^{1+s} \\
        = & \tilde{\phi}(-s|\rho_1||\sigma_1)+\tilde{\phi}(-s|\rho_2||\sigma_2).
    \end{align}
    Then we divide both sides by $s$, we obtain
    \begin{align}
        \underline{D}_{1+s}(\rho_1\otimes \rho_2||\sigma_1\otimes \sigma_2)=\underline{D}_{1+s}(\rho_1||\sigma_1)+\underline{D}_{1+s}(\rho_2||\sigma_2).
    \end{align}
\end{proof}

\begin{proof}[Proof of Lemma~\ref{lemma:limitofRenyies}]
    In the case of Petz Relative R\'{e}nyi entropy, at first, we check the differential of $\rho^s$.\\
    Let $\rho=\sum_i\lambda_ic_i$ be the spectral decomposition of $\rho$.Then,
    \begin{align}
        \frac{d}{ds}\rho^s=\frac{d}{ds}\sum_i\lambda_i^sc_i=\sum_i \log\lambda_i \lambda_i^s c_i =\rho^s\circ\log\rho.
    \end{align}
    Hence, the differential of $\rho^{1+s}\circ \sigma^{-s}$ is
    \begin{align}
        \frac{d}{ds}\rho^{1+s}\circ \sigma^{-s}=(\rho^{1+s}\circ\log \rho)\circ \sigma^{-s}-\rho^{1+s}\circ(\sigma^{-s}\circ\log\sigma).
    \end{align}
    Then, the following equality holds.
    \begin{align}
          & \lim_{s\to 0}D_{1+s}(\rho||\sigma)=\frac{d}{ds}\phi(-s|\rho||\sigma)|_{s=0}                                                                              \\
        = & \frac{d}{ds}\log\mathrm{tr}\rho^{1+s}\circ\sigma^{-s}|_{s=0}                                                                                             \\
        = & \frac{1}{\mathrm{tr}\rho^{1+s}\circ\sigma^{-s}}\frac{d}{ds}\mathrm{tr}\rho^{1+s}\circ\sigma^{-s}|_{s=0}                                                  \\
        = & \frac{1}{\mathrm{tr}\rho^{1+s}\circ\sigma^{-s}}\mathrm{tr}(\rho^{1+s}\circ\log \rho)\circ \sigma^{-s}-\rho^{1+s}\circ(\sigma^{-s}\circ\log\sigma)|_{s=0} \\
        = & \mathrm{tr}\rho\log\rho-\rho\log\sigma=D(\rho||\sigma).
    \end{align}
    In the case of Sandwiched Relative R\'{e}nyi entropy, the following equality holds.
    \begin{align}
          & \lim_{s\to 0}\underline{D}_{1+s}(\rho||\sigma)=\frac{d}{ds}\tilde{\phi}(-s|\rho||\sigma)|_{s=0}                                                                                                                                                      \\
        = & \left.\frac{d}{ds}\log\mathrm{tr}\left(P_{\sigma^{\frac{-s}{2(1+s)}}}(\rho)\right)^{1+s}\right|_{s=0}
=\left. \frac{\frac{d}{ds}\mathrm{tr}\left(P_{\sigma^{\frac{-s}{2(1+s)}}}(\rho)\right)^{1+s}}{\mathrm{tr}\left(P_{\sigma^{\frac{-s}{2(1+s)}}}(\rho)\right)^{1+s}}\right|_{s=0}.\label{equation:sandwichlimitproof1}
    \end{align}
    Now, we consider the differential of $\left(P_{\sigma^{\frac{-s}{2(1+s)}}}(\rho)\right)^{1+s}$
    \begin{align}
          & \frac{d}{ds}\left(P_{\sigma^{\frac{-s}{2(1+s)}}}(\rho)\right)^{1+s} \\
        = & \left(\left(P_{\sigma^{\frac{-s}{2(1+s)}}}(\rho)\right)^{1+s}\circ \log \left(P_{\sigma^{\frac{-s}{2(1+s)}}}(\rho)\right)\right)
+(1+s)\left(P_{\sigma^{\frac{-s}{2(1+s)}}}(\rho)\right)^{s}\circ\frac{d}{ds}\left(P_{\sigma^{\frac{-s}{2(1+s)}}}(\rho)\right).\label{equation:sandwichlimitproof2}
    \end{align}
    Here,
    \begin{align}
          & \frac{d}{ds}P_{\sigma^{\frac{-s}{2(1+s)}}}(\rho)=\frac{d}{ds}2L_{\sigma^{\frac{-s}{2(1+s)}}}L_{\sigma^{\frac{-s}{2(1+s)}}}(\rho)-L_{\sigma^{\frac{-s}{1+s}}}(\rho)                                                  \\
        = & 2\frac{-1}{2(1+s)^2}\sigma^{\frac{-s}{2(1+s)}}\circ\left(\left(\log \sigma\circ \sigma^{\frac{-s}{2(1+s)}}\right)\circ \rho\right)+2\left(\sigma^{\frac{-s}{2(1+s)}}\circ\frac{-1}{2(1+s)^2}\sigma^{\frac{-s}{2(1+s)}}\right)\circ\left(\log \sigma\circ\rho\right) \nonumber \\ 
       & -  \frac{-1}{(1+s)^2}\left(\sigma^{\frac{-s}{1+s}}\circ\log \sigma\right)\circ\rho.\label{equation:sandwichlimitproof3}
    \end{align}
    Hence, from \eqref{equation:sandwichlimitproof1}, \eqref{equation:sandwichlimitproof2}, \eqref{equation:sandwichlimitproof3}, we obtain
    \begin{align}
        \lim_{s\to 0}\underline{D}_{1+s}(\rho||\sigma)=D(\rho||\sigma).
    \end{align}
\end{proof}

\begin{proof}[Proof of Lemma~\ref{lemma:monotonicityins}]
    Let $\rho=\sum_i \lambda_i c_i$ and $\sigma=\sum_i \mu_i e_i$ be the spectral decompositions of $\rho,\sigma$.
    Now we focus on $\log \mathrm{tr}\rho^{1+s}\circ\sigma^{-s}$.
    We calculate differentiation of $\log \mathrm{tr}\rho^{1+s}\circ \sigma^{-s}$ as follows:
    \begin{align}
        \phi''(-s|\rho||\sigma)&=\frac{d}{ds}\frac{\mathrm{tr}(\rho^{1+s}\circ\log\rho)\circ \sigma^{-s}-\rho^{1+s}\circ(\sigma^{-s}\circ\log\sigma)}{\mathrm{tr}\rho^{1+s}\circ\sigma^{-s}}\\
&=\frac{d}{ds}\frac{\sum_{i,j}\lambda_i^{1+s}\mu_j^{-s}(\log\lambda_i-\log\mu_j)\mathrm{tr}c_i\circ e_j}{\mathrm{tr}\rho^{1+s}\circ\sigma^{-s}}\\
&=\frac{\sum_{i,j}\lambda_i^{1+s}\log\lambda_i\mu_j^{-s}(\log\lambda_i-\log\mu_j)\mathrm{tr}c_i\circ e_j-\sum_{i,j}\lambda_i^{1+s}\mu_j^{-s}\log\mu_j(\log\lambda_i-\log\mu_j)\mathrm{tr}c_i\circ e_j}{\mathrm{tr}\rho^{1+s}\circ\sigma^{-s}}\nonumber\\
&-\left(\frac{\sum_{i,j}\lambda_i^{1+s}\mu_j^{-s}(\log\lambda_i-\log\mu_j)\mathrm{tr}c_i\circ e_j}{\mathrm{tr}\rho^{1+s}\circ\sigma^{-s}}\right)^2.
    \end{align}
    By applying Schwarz inequality to the vector $(\lambda_i^{\frac{1+s}{2}}\mu_j^{\frac{-s}{2}}(\mathrm{tr}c_i\circ e_j)^{\frac{1}{2}})_{i,j}$ and $(\lambda_i^{\frac{1+s}{2}}\mu_j^{\frac{-s}{2}}(\log\lambda_i-\log\mu_j)(\mathrm{tr}c_i\circ e_j)^{\frac{1}{2}})_{i,j}$, we obtain
    $\phi''(-s|\rho||\sigma)\ge 0$.
    Therefore, $\phi(-s|\rho||\sigma)$ is convex and $D_{1+s}(\rho||\sigma)=\frac{\phi(-s|\rho||\sigma)}{s}$ is monotone increasing.
\end{proof}

\begin{proof}[Lemma~\ref{lemma:Jenseninequality}]
    Let $x=\sum_i \lambda_i c_i$ be a spectral decomposition of $x$. From the Lemma~\ref{lemma:CSOIisMeasurement}, $\{c_i\}$ is a Measurement.\\
    Let $f(x)=\sum_i f(\lambda_i)c_i$ be from Definition~\ref{definition:fofEJAelement}, and $p_i=\mathrm{tr}\rho\circ c_i$  be a probability distribution defined by Definition~\ref{definition:distribution}.
    From an original Jensen's inequality, we obtain
    \begin{align}
        \mathrm{tr}\rho\circ f(x)=\sum_i p_i f(\lambda_i)\ge f(\sum_i p_i\lambda_i)=f(\sum_i\mathrm{tr}\rho\circ c_i\lambda_i)=f(\mathrm{tr}\rho\circ x).
    \end{align}
\end{proof}

\begin{proof}[Proof of Lemma~\ref{lemma:numberofeigenvalues}]
    From the spectral decomposition of $x$ ,we can write down as $x^{\otimes n}=\sum_{j_1,\ldots j_d}\lambda_1^{j_1}\ldots\lambda_d^{j_d}c_1^{j_1}\otimes \cdots \otimes c_d^{j_d}\quad(j_1+\cdots+j_d=n)$.\\
    The numbers $j_1,\ldots, j_{d-1}$ take the values from $0$ to $n$ but $j_d$ is decided by $j_1,\ldots,j_{d-1}$ because of the relation $j_1+\cdots+j_d=n$.
    Therefore, the eigenvalues of $x^{\otimes n}$, $\lambda_1^{j_1}\cdots \lambda_d^{j_d}$ takes at most $(n+1)^{d-1}$ values.\\
    The spectral decomposition of $x^{\otimes n}$ has all distinct eigenvalues and a complete system of orthogonal idempotents. Hence, the number of eigenvalues and elements of the set of a complete system of orthogonal idempotents are bounded by $(n+1)^{d-1}$
\end{proof}

\subsection{Proofs about Petz Relative R\'{e}nyi entropy}\label{appendix:proofofRelativeRenyi}

\begin{proof}[Proof of Lemma~\ref{lemma:Monotonicity of Relativerenyi1}]
    Simillary to the proof of Lemma~\ref{lemma:RepresentEntropies}, we define a new CSOI $\{c_{i,j}\}$.
    The spectral decomposition of $\sigma$ is given as $\sigma=\sum_i \mu_i e_i$.
    Now we define another CSOI $\{c_{i,j}\}$ by $c_{i,j}:=P(e_i)c'_{i,j}$ by Lemma~\ref{lemma:newCSOI} ,where the spectral decomposition $\rho=\sum_j\lambda_j c_j$ and $\kappa_\sigma(\rho)=\sum_{i,j}\lambda_{i,j}c_{i,j}$ hold.
    Then. we have the following relations similarly to the proof of Lemma~\ref{lemma:RepresentEntropies}.
    \begin{align}
        P(e_i)\rho&=\sum_{j}\lambda_{i,j}c_{i,j}.\label{equation:proofofRRE1}\\
        \kappa_{\sigma}(\rho)&=\sum_i P(e_i)\rho=\sum_{i,j}\lambda_{i,j}c_{i,j}.\label{equation:proofofRRE2}\\
        \sum_{j}c_{i,j}&=e_i.\label{equation:proofofRRE3}
    \end{align}
    Then, the following relation holds:
    \begin{align}
        \mathrm{tr}(\rho\circ c_{i,j})\stackrel{(a)}{=}\mathrm{tr}(P(c_{i,j})\rho)\stackrel{(b)}{=}\mathrm{tr}(\lambda_{i,j}c_{i,j}).\label{equation:proofofRRE5}
    \end{align}
    The equation (a) is shown by the Euclidean condition.
    The equation (b) is shown by the condition \eqref{equation:proofofRRE1} and \eqref{equation:proofofRRE3} similarly to the proof in Lemma~\ref{lemma:RepresentEntropies}.
    Now, we focus on $\mathrm{tr}\rho\circ\frac{c_{i,j}}{\mathrm{tr}c_{i,j}}$.
    We apply Jensen inequality in EJAs (Lemma~\ref{lemma:Jenseninequality}).
    Then we obtain
    \begin{align}
        (\mathrm{tr}\rho\circ\frac{c_{i,j}}{\mathrm{tr}c_{i,j}})^{1+s}\le \mathrm{tr}\rho^{1+s}\circ\frac{c_{i,j}}{\mathrm{tr}c_{i,j}},\quad s\ge 0.\label{equation:proofofRRE4}
    \end{align}
    Therefore, the following relation holds:
    \begin{align}
        \mathrm{tr}\sigma^{-s}\circ\kappa_\sigma(\rho)^{1+s}&\stackrel{(a)}{=}\sum_{i,j}\mathrm{tr}\mu_j^{-s}\lambda_{i,j}^{1+s}c_{i,j}\\
&\stackrel{(b)}{=}\sum_{i,j}\mu_i^{-s}\mathrm{tr}c_{i,j}(\mathrm{tr}\rho\circ\frac{c_{i,j}}{\mathrm{tr}c_{i,j}})^{1+s}\\
&\stackrel{(c)}{\le}\sum_{i,j}\mu_i^{-s}\mathrm{tr}c_{i,j}(\mathrm{tr}\rho^{1+s}\circ \frac{c_{i,j}}{\mathrm{tr}c_{i,j}})\\
&=\sum_{i,j}\mathrm{tr}\mu_i^{-s}\rho^{1+s}\circ c_{i,j}\stackrel{(d)}{=}\mathrm{tr}\rho^{1+s}\circ \sigma^{-s}\label{equation:proofofRRE6}
    \end{align}
    The equality (a) is shown by the relation \eqref{equation:proofofRRE2} ,spectral decomposition of $\sigma$ and \eqref{equation:proofofRRE3}.
The equality (b) is shown by the relation \eqref{equation:proofofRRE5}.
The inequality (c) is shown by Lemma~\ref{lemma:Jenseninequality}.
The equality (d) is shown by \eqref{equation:proofofRRE3} and spectral decomposition of $\sigma$.
Therefore, we divide \eqref{equation:proofofRRE6} by $s>0$, and then we obtain the conclusion.
\end{proof}

\begin{proof}[Proof of Lemma~\ref{lemma:Monotonicity of Relativerenyi2}]
    The spectral decomposition of $\sigma$ and CSOI $\{c_{i,j}\}$ are given similarly to the proof of Lemma~\ref{lemma:Monotonicity of Relativerenyi1}, i.e., 
    $\sigma=\sum_{i}\mu_i e_i$ and $\kappa_\sigma(\rho)=\sum_{i,j}\lambda_{i,j}c_{i,j}$.
    In addition, \eqref{equation:proofofRRE1}, \eqref{equation:proofofRRE2} and \eqref{equation:proofofRRE3} hold.
	Now, for a measurement $\bm{M}=\{M_i\}$, we define new measurement ${\bm{M}'}^\rho_\sigma$ as follows:
	
    \begin{align}
        {\bm{M}'}^\rho_\sigma&:=\{\sum_{j,k}\bm{M}^\rho_\sigma(i,j,k)\}=\{{\bm{M}'}^\rho_\sigma(i)\},
    \end{align}
    where $\bm{M}^\rho_\sigma(i,j,k)=P_{c_{j,k}}(M_i)$ defined in Definition~\ref{definition:pinchedmeasurement}.

    Then, we obtain the following inequality:
    \begin{align}
        \mathrm{tr}\rho^{1+s}\circ \sigma^{-s}&\stackrel{(a)}{\ge}\mathrm{tr}\kappa_\sigma(\rho)^{1+s}\circ\sigma^{-s}\\
        &\stackrel{(b)}{=}\sum_{j,k}\mu_j^{-s}\lambda^{1+s}_{j,k}\mathrm{tr}c_{j,k}\\
        &=\sum_{j,k}\mu_j^{-s}\lambda^{1+s}_{j,k}\mathrm{tr}c_{j,k}\circ\sum_iM_i\\
        &\stackrel{(c)}{=}\sum_{i,j,k}(\mathrm{tr}\rho\circ \bm{M}^\rho_\sigma(i,j,k))^{1+s}(\mathrm{tr}\sigma\circ M^\rho_\sigma(i,j,k))^{-s}\\
        &\stackrel{(d)}{\ge}\sum_i(\mathrm{tr}\rho\circ {\bm{M}'}^\rho_\sigma(i))^{1+s}(\mathrm{tr}\sigma\circ {\bm{M}'}^\rho_\sigma(i))^{-s}\\
        &\stackrel{(e)}{=}\sum_i(\mathrm{tr}\kappa_\sigma(\rho)\circ M_i)^{1+s}(\mathrm{tr}\sigma\circ M_i)^{-s}\\
        &\stackrel{(f)}{\ge}|\bm{C}_\sigma|^{-(1+s)}\sum_i(\mathrm{tr}\rho\circ M_i)^{1+s}(\mathrm{tr}\sigma\circ M_i)^{-s}\label{equation:proofofRRE9}
    \end{align}
The inequality (a) is shown in the proof of Lemma~\ref{lemma:Monotonicity of Relativerenyi1}.
The equality (b)  is shown by the condition \eqref{equation:proofofRRE1} and \eqref{equation:proofofRRE3} similarly to the proof in Lemma~\ref{lemma:RepresentEntropies}.
The equality (c) is shown by the following relations of $\rho,\sigma$:
\begin{align}
    \mathrm{tr}\rho\circ \bm{M}^\rho_\sigma(i,j,k)&=\mathrm{tr}\rho\circ P_{c_{j,k}}(M_i)\\
    &\stackrel{(g)}{=}\mathrm{tr}(P_{c_{j,k}}(\rho))\circ M_i\\
    &\stackrel{(h)}{=}\lambda_{j,k}\mathrm{tr}c_{j,k}\circ M_i.\label{equation:proofofRRE7}
\end{align}
\begin{align}
    \mathrm{tr}\sigma\circ \bm{M}^\rho_\sigma(i,j,k)&=\mathrm{tr}\sigma\circ P_{c_{j,k}}(M_i)\\
    &=\mathrm{tr}(P(c_{j,k})\sigma)\circ M_i\\
    &=\mu_j\mathrm{tr}c_{j,k}\circ M_i.\label{equation:proofofRRE8}
\end{align}
The equality (g) is shown by the Euclidean condition.
The equality (h) is shown by \eqref{equation:proofofRRE5}.

The inequality (d) is shown by the monotonicity of classical Relative R\'{e}nyi entropy. 
The equality (e) is shown by taking sum with respect to $j,k$ in \eqref{equation:proofofRRE7} and \eqref{equation:proofofRRE8}.
The inequality (f) is shown as follows:
First, we apply the pinching inequality (Lemma~\ref{lemma:PinchingInequality}).
\begin{align}
    |\bm{C}_\sigma|\kappa_\sigma(\rho)\ge \rho.
\end{align}
In addition, $P_{\sqrt{M_i}}(|\bm{C}_\sigma|\kappa_\sigma(\rho)-\rho)\ge 0$ because of Lemma~\ref{lemma:positivityofQradratic}.
Finally, we take trace of $P_{\sqrt{M_i}}(|\bm{C}_\sigma|\kappa_\sigma(\rho)-\rho)$, we obtain
\begin{align}
    \mathrm{tr}P_{\sqrt{M_i}}(|\bm{C}_\sigma|\kappa_\sigma(\rho)-\rho)&=\left\langle|\bm{C}_\sigma|\kappa_\sigma(\rho)-\rho,P_{\sqrt{M_i}}(u)\right\rangle\\
    &=\left\langle|\bm{C}_\sigma|\kappa_\sigma(\rho)-\rho,M_i\right\rangle\\
    &=|\bm{C}_\sigma|\mathrm{tr}\kappa_\sigma(\rho)\circ M_i-\mathrm{tr}\rho\circ M_i \ge 0.
\end{align}
Therefore, $|\bm{C}_\sigma|^{1+s}(\mathrm{tr}\kappa_\sigma(\rho)\circ M_i)^{1+s}\ge(\mathrm{tr}\rho\circ M_i)^{1+s}$ for $s>0$.

By taking logarithm in \eqref{equation:proofofRRE9} and divide by $s>0$, then we obtain the conclusion.
\end{proof}

\subsection{Proofs about Sandwiched Relative R\'{e}nyi entropy}\label{appendix:proofofSandwich}
\begin{proof}[Proof of Lemma~\ref{lemma:inequalityofSandwich}]
    At first, we show the following inequality:
    \begin{align}
        \underline{D}_{1+s}(\rho||\sigma)\ge D_{1+s}(\kappa_\sigma(\rho)||\sigma).\label{equation:proofofSRR1}
    \end{align}
    Similarly to the proof of Lemma~\ref{lemma:Monotonicity of Relativerenyi1}, we give 
    the spectral decomposition $\sigma=\sum_i \mu_ie_i$ and CSOI $\{c_{i,j}\}$ satisfying \eqref{equation:proofofRRE1}, \eqref{equation:proofofRRE2} and \eqref{equation:proofofRRE3}.
    Then, we calculate as follows:
    \begin{align}
        \mathrm{tr}\sigma^{-s}\circ \kappa_\sigma(\rho)^{1+s}&=\sum_{i,j}\mathrm{tr}\mu_j^{-s}\lambda_{i,j}^{1+s}c_{i,j}\\
        &=\sum_{i,j}\mathrm{tr}c_{i,j}\left(\mathrm{tr}\frac{\rho\circ (\mu_j^{\frac{-s}{1+s}}c_{i,j})}{\mathrm{tr}c_{i,j}}\right)^{1+s}\\
        &=\sum_{i,j}\mathrm{tr}c_{i,j}\left(\mathrm{tr}P_{\sigma^{\frac{-s}{2(1+s)}}}(\rho)\circ \frac{c_{i,j}}{\mathrm{tr}c_{i,j}}\right)^{1+s}\\
        &\stackrel{(a)}{\le} \sum_{i,j}\mathrm{tr}c_{i,j}\mathrm{tr}\left(P_{\sigma^{\frac{-s}{2(1+s)}}}(\rho)\right)^{1+s}\circ\frac{c_{i,j}}{\mathrm{tr}c_{i,j}}\\
        &=\mathrm{tr}\left(P_{\sigma^{\frac{-s}{2(1+s)}}}(\rho)\right)^{1+s}.\label{equation:proofofSRR3}
    \end{align}
    The inequality (a) is shown by Jensen's inequality with EJAs (Lemma~\ref{lemma:Jenseninequality}) for the state $\frac{c_{i,j}}{\mathrm{tr}c_{i,j}}$.
    Therefore, by taking logarithm in \eqref{equation:proofofSRR3} and dividing by $s$ we obtain the conclusion.

    Next, we show the following inequality:
    \begin{align}
        D_{1+s}(\kappa_\sigma(\rho)||\sigma)+\frac{1+s}{s}\log|\bm{C}_\sigma|\ge \underline{D}_{1+s}(\rho||\sigma).\label{equation:proofofSRR2}
    \end{align}
    By pinching inequality, we have $\rho\le |\bm{C}_\sigma|\kappa_\sigma(\rho)$.
    By applying Lemma~\ref{lemma:positivityofQradratic}, the following relation holds:
    \begin{align}
        P_{\sigma^{\frac{-s}{2(1+s)}}}(\rho)\le |\bm{C}_\sigma|P_{\sigma^{\frac{-s}{2(1+s)}}}(\kappa_\sigma(\rho)).
    \end{align}
    From next Lemma~\ref{lemma:traceineq}, we obtain 
    \begin{align}
        \mathrm{tr}\left(P_{\sigma^{\frac{-s}{2(1+s)}}}(\rho)\right)^{1+s}&\le |\bm{C}_\sigma|^{1+s}\mathrm{tr}P_{\sigma^{\frac{-s}{2(1+s)}}}\left(\kappa_\sigma(\rho)\right)^{1+s}.\label{equation:proofofSRR4}\\
&=|\bm{C}_\sigma|^{1+s}\mathrm{tr}\kappa_\sigma(\rho)^{1+s}\circ \sigma^{-s}
    \end{align}
    By taking logarithm in \eqref{equation:proofofSRR4} and divide by $s$, we obtain conclusion.
    
    Finally, conbining \eqref{equation:proofofSRR1} and \eqref{equation:proofofSRR2}, we obtain Lemma~\ref{lemma:inequalityofSandwich}.
\end{proof}

\begin{lemma}\label{lemma:traceineq}
    Let $x,y$ be elements in EJAs satisfying $0\le x\le y$.
    Then, $\mathrm{tr}x^{1+s}\le \mathrm{tr}y^{1+s}$ for $s\ge 0$.
\end{lemma}
\color{black}
\begin{proof}
    The spectral decompositions of $x,y$ are given as $x=\sum_i x_i c_i,y=\sum_i y_i d_i$.
    Then,
    \begin{align}
        \mathrm{tr}y^{1+s}&=\sum_i y_i^{1+s}\langle u,d_i\rangle=\sum_{i,j}y_i^{1+s}\langle c_j,d_i\rangle\\
        &\stackrel{(a)}{\ge} \sum_j(\sum_i y_i\langle c_j,d_i\rangle)^{1+s}=\sum_j\langle c_j, y\rangle^{1+s}\stackrel{(b)}{\ge} \sum_j\langle c_j,x\rangle^{1+s}\\
        &\stackrel{(c)}{=}\sum_j x_j^{1+s}\langle c_j,u\rangle=\mathrm{tr}x^{1+s}.
    \end{align}
    The inequality (a) is shown by Jensen's inequality (Lemma~\ref{lemma:Jenseninequality}) for the probability distribution $\{\langle c_j,d_i\rangle\}_i$, where
    $\langle c_j, d_i\rangle=\langle c_j, P_{d_i}(u)\rangle=\mathrm{tr}P_{d_i}(c_j)\ge 0$ and $\langle c_j,u\rangle=1$ hold by normalization of inner product in Section~\ref{subsection:EuclideanJordanalgebra} Lemma~\ref{lemma:positivityofQradratic}.
    The inequality (b) is shown by the condition $0\le x\le y$.
    The equality (c) is shown by normalization of the norm $\langle c_j,u\rangle=1$ discussed in Section~\ref{subsection:EuclideanJordanalgebra}.
Therefore, we obtain the conclusion.
\end{proof}

\begin{proof}[Proof of Lemma~\ref{lemma:observationinequalityofSandwich}]
    By applying Lemma~\ref{lemma:inequalityofSandwich} to states $\rho^{\otimes n},\sigma^{\otimes n}$ in $V^{\otimes n}$, we obtain
    \begin{align}
        D_{1+s}(\kappa_{\sigma^{\otimes n}}(\rho^{\otimes n})||\sigma^{\otimes n})+\frac{1+s}{s}\log|\bm{C}_{\sigma^{\otimes n}}|\ge n\underline{D}_{1+s}(\rho||\sigma)\stackrel{(a)}{=} \underline{D}_{1+s}(\rho^{\otimes n}||\sigma^{\otimes n})\ge D_{1+s}(\kappa_{\sigma^{\otimes n}}(\rho^{\otimes n})||\sigma^{\otimes n}).\label{equation:proofofSRR5}
    \end{align}
    The equality (a) is shown by additivity (Lemma~\ref{theorem:additivity}).
    By deviding \eqref{equation:proofofSRR5} by $n$ and applying Lemma~\ref{lemma:numberofeigenvalues}, we obtain
    \begin{align}
       \frac{1}{n} D_{1+s}(\kappa_{\sigma^{\otimes n}}(\rho^{\otimes n})||\sigma^{\otimes n})+\frac{1+s}{ns}\log(n+1)^{d-1}\ge \underline{D}_{1+s}(\rho||\sigma)\ge\frac{1}{n} D_{1+s}(\kappa_{\sigma^{\otimes n}}(\rho^{\otimes n})||\sigma^{\otimes n}),\label{equation:proofofSRR6}
    \end{align}
    where $d:=|\bm{C}_\sigma|$.
    Now, we take a limit of $n$ in \eqref{equation:proofofSRR6}, we obtain
    \begin{align}
        \underline{D}_{1+s}(\rho||\sigma)=\lim_{n\to\infty} \frac{1}{n}D_{1+s}(\kappa_{\sigma^{\otimes n}}(\rho^{\otimes n})||\sigma^{\otimes n})\label{equation:proofofSRR7}
    \end{align}
    Combining \eqref{equation:proofofSRR7} and  Theorem~\ref{theorem:MonotonicityofRelativerenyi}, we obtain the conclusion.
\end{proof}

\subsection{Proofs about Relative entropy}\label{appendix:proofofRelative}
\begin{proof}[Proof of Lemma~\ref{lemma:DirectpartofRelative1}]
    By definition of Relative entropy, the following relation holds:
    \begin{align}
        D(\kappa_\sigma(\rho)||\sigma)-D(\rho||\sigma)=\mathrm{tr}\kappa_\sigma(\rho)\circ \log\kappa_\sigma(\rho)-\mathrm{tr}\kappa_\sigma(\rho)\circ \log\sigma-(\mathrm{tr}\rho\circ\log\rho-\mathrm{tr}\rho\circ\log\sigma).\label{equation:proofofR1}
    \end{align}
    Now, the following two relations holds:
    \begin{align}
        \mathrm{tr}\kappa_\sigma(\rho)\circ \log\sigma&\stackrel{(a)}{=}\mathrm{tr}\rho\circ\log\sigma.\label{equation:proofofR2}\\
        \mathrm{tr}\kappa_\sigma(\rho)\circ\log\kappa_\sigma(\rho)&\stackrel{(b)}{=}\mathrm{tr}\rho\circ \log\kappa_\sigma(\rho).\label{equation:proofofR3}
    \end{align}
    The equality (a) and (b) is shown by Eualidean condition.
    Therefore, applying \eqref{equation:proofofR1} to \eqref{equation:proofofR2} and \eqref{equation:proofofR3}, we obtain
    \begin{align}
        D(\kappa_\sigma(\rho)||\sigma)-D(\rho||\sigma)=\mathrm{tr}\rho\circ\kappa_\sigma(\rho)-\mathrm{tr}\rho\circ\log\rho=-D(\rho||\kappa_\sigma(\rho)).
    \end{align}
    Hence, we obtain conclusion.
\end{proof}

\begin{proof}[Proof of Lemma~\ref{lemma:DirectpartofRelative2}]
    At first, we consider the case which $\rho$ is an external point of state space of $\cV$, i.e., $\rho$ is an element of a jordan frame.
    \begin{align}
        H(\rho)=-\mathrm{tr}\rho\circ \log\rho=0.
    \end{align}   
    For a jordan frame $\{c_i\}_{i=1}^d$, there exists $\lambda_i$ satisfying $P_{c_i}(\rho)=\lambda_i\rho$.
    This equality is derived from $\cV(i,1)=\mathbb{R}c_i$, where $\cV(i,1)$ is a direct sum factor on Peirce decomposition by $\{c_i\}_{i=1}^d$.
    Therefore, we obtain
    \begin{align}
        H(\kappa_{\bm{C}}(\rho))&=-\mathrm{tr}\sum_{i=1}^dP_{c_i}(\rho)\circ \log\sum_{j=1}^dP_{c_j}(\rho).\\
&=-\mathrm{tr}(\sum_{i=1}^d\lambda_ic_i)\circ\log(\sum_{j=1}^d\lambda_j c_j).\\
&=-\mathrm{tr}(\sum_{i=1}^d\lambda_ic_i)\circ(\sum_{j=1}^d\log\lambda_j c_j).\\
&=-\mathrm{tr}\sum_{i=1}^d \lambda_i\log\lambda_i c_i.\\
&\stackrel{(a)}{=}-\sum_{i=1}^d\lambda_i\log\lambda_i\stackrel{(b)}{\le}\log|\bm{C}|.
    \end{align} 
    The equation (a) is shown by a normalization of a norm.
    The inequality (b) is shown as follows:
    The inequality $\lambda_i\ge 0$ is shown by $\rho\ge 0$ and $P_{c_i}(\rho)\ge 0$.
    The inequlity $\rho\ge 0$ is shown by a spectral decomposition by a CSOI $\{\rho,u-\rho\}$.
    The inequality $P_{c_i}(\rho)\ge 0$ is shown by Lemma~\ref{lemma:positivityofQradratic}.
    In addition, $\lambda_i\le 1$ is shown by $\kappa_{\bm{C}}(\rho)=\sum_i P_{c_i}(\rho)=\sum_i\lambda_i\rho$ is a state.
    
    Next we consider the case which $\rho$ is a convex conbination of external points.
    For states $\rho,\rho_i$ and a probability distribution $\{p_i\}_{i=1}^d$,
    we obtain
    \begin{align}
        D(\sum_{i=1}^dp_i\rho_i||\kappa_{\bm{C}}(\sum_{j=1}^dp_j\rho_j))&=D(\sum_{i=1}^dp_i\rho_i||\sum_{j=1}^dp_j\kappa_{\bm{C}}(\rho_j))\\
        &\stackrel{(a)}{\le} \sum_{i=1}^dp_iD(\rho_i||\kappa_{\bm{C}}(\rho_i))\stackrel{(b)}{\le}\sum_{i=1}^dp_i\log|\bm{C}|=\log|\bm{C}|.
    \end{align}
    The inequality (a) is shown by joint convexity (Theorem~\ref{theorem:jointconvexity}).
    The inequality (b) is shown by an external point case.
    Therefore, we obtain the conclusion.
\end{proof}
\color{black}
\subsection{Proof from Theorem~\ref{theorem:Stein2} to Theorem~\ref{theorem:Stein} and from Theorem~\ref{theorem:Stein} to Theorem~\ref{theorem:Stein2}}\label{appendix:proofofStein}
We fix $0<\epsilon<1$.
(1)When $B^\dag(\rho||\sigma)=B(\rho||\sigma)=D(\rho||\sigma)$ holds, for arbitrary $\delta>0$
there exists a family $\{T_n\}$ satisfying the following condition because $B(\rho||\sigma)>D(\rho||\sigma)-\delta$ holds.

\begin{align}
    \varliminf_{n\to\infty}-\frac{1}{n}\log\mathrm{tr}\sigma^{\otimes n}\circ T_n\ge D(\rho||\sigma)-\delta,\quad\lim_{n\to\infty}\mathrm{tr}\rho^{\otimes n}\circ (u-T_n)=0.
\end{align}

Because $\lim_{n\to\infty}\mathrm{tr}\rho^{\otimes n}\circ(u-T_n)=0$ holds, for $\epsilon$ there exist $N$ such that $\mathrm{tr}\rho^{\otimes n}\circ(u-T_n)<\epsilon$ holds for every $n\ge N$.
For $n\ge N$, we obtain
\begin{align}
    \mathrm{tr}\sigma^{\otimes n}\circ T_n \ge \beta_\epsilon^n(\rho||\sigma).
\end{align}

Therefore,
\begin{align}
    -\frac{1}{n}\log\beta^n_\epsilon(\rho||\sigma)\ge -\frac{1}{n}\log\mathrm{tr}\sigma^{\otimes n}\circ T_n,\quad n\ge N,\quad 0<\epsilon<1.
\end{align}
Taking limit inferior, we obtain
\begin{align}
    \varliminf_{n\to \infty}-\frac{1}{n}\log\beta^n_\epsilon(\rho||\sigma)\ge\varliminf_{n\to\infty}-\frac{1}{n}\log\mathrm{tr}\sigma^{\otimes n}\circ T_n\ge D(\rho||\sigma)-\delta.
\end{align}
We take $\delta\to 0$, then we obtain
\begin{align}
    \varliminf_{n\to\infty}-\frac{1}{n}\log\beta^n_\epsilon(\rho||\sigma)\ge D(\rho||\sigma).\label{equation:proofofstein1}
\end{align}

(2)We suppose that there exists the family $\{T_n\}$ satisfying following conditions:
\begin{align}
    &\varlimsup_{n\to\infty}-\frac{1}{n}\log\beta_\epsilon^n(\rho||\sigma)<D(\rho||\sigma),\\
&\beta^n_\epsilon(\rho||\sigma)=\mathrm{tr}\sigma^{\otimes n}\circ T_n,\quad \mathrm{tr}\rho^{\otimes n}\circ(u-T_n)\le \epsilon,\quad \forall n
\end{align}
Then, $B^\dag(\rho||\sigma)<D(\rho||\sigma)$ holds and this is contradiction.
Therefore, in order to satisfy $B^\dag(\rho||\sigma)=D(\rho||\sigma)$,
it is necessarily to satisfy 
\begin{align}
    \varlimsup_{n\to\infty}-\frac{1}{n}\log\beta_\epsilon^n(\rho||\sigma)\ge D(\rho||\sigma).\label{equation:proofofstein2}
\end{align}

(3)Combining \eqref{equation:proofofstein1} and \eqref{equation:proofofstein2}, we obtain
\begin{align}
    \lim_{n\to\infty}-\frac{1}{n}\log\beta^n_\epsilon(\rho||\sigma)=D(\rho||\sigma).
\end{align}

(4) Next, under the condition of Stein's Lemma, we show $B^\dag(\rho||\sigma)\le D(\rho||\sigma)\le B(\rho||\sigma)$.
From Stein's Lemma, we obtain
\begin{align}
    \lim_{n\to\infty}-\frac{1}{n}\log\beta^n_\epsilon(\rho||\sigma)=D(\rho||\sigma),
\end{align}
where $0<\epsilon<1$.
For arbitrary $0<\epsilon<1$, there exists a number $N$ and a family $\{T_n\}$ such that
\begin{align}
    &\mathrm{tr}\rho^{\otimes n}\circ (u-T_n)<\epsilon,\label{equation:A.51}\\
    &-\frac{1}{n}\log\mathrm{tr}\sigma^{\otimes n}\circ T_n= -\frac{1}{n}\log\beta^n_{\epsilon}(\rho||\sigma),
\end{align}
for $n\ge N$.
Therefore, by definition of $B(\rho||\sigma)$, we obtain $B(\rho||\sigma)\ge D(\rho||\sigma)$.
Similarly to obtaining $B(\rho||\sigma)\ge D(\rho||\sigma)$, we obtain $D(\rho||\sigma)\ge B^\dag(\rho||\sigma)$ by following way:
For a family $\{T_n\}$ satisfying $\varliminf_{n\to\infty}\mathrm{tr}\rho^{\otimes n}\circ (u-T_n)<1$ and arbitrary $0<\epsilon<1$, there exist a number $N$
\begin{align}
&\inf_{n\ge k}\mathrm{tr}\rho^{\otimes n}\circ (u-T_n)\le \epsilon,\\
    &\varliminf_{n\to\infty}-\frac{1}{n}\log\mathrm{tr}\sigma^{\otimes n}\circ T_n\le \lim_{n\to\infty}-\frac{1}{n}\log\beta^n_\epsilon(\rho||\sigma)=D(\rho||\sigma),
\end{align}
where $k\ge N$.
Then, $\varliminf_{n\to\infty}\mathrm{tr}\rho^{\otimes n}\circ (u-T_n)<1$ and $\varliminf_{n\to\infty}-\frac{1}{n}\log\mathrm{tr}\sigma^{\otimes n}\circ T_n=\varliminf_{n\to\infty}-\frac{1}{n}\log\beta^n_\epsilon(\rho||\sigma)=D(\rho||\sigma)$ holds.
Therefore, we obtain $D(\rho||\sigma)\ge B^\dag(\rho||\sigma)$.

\subsection{Classical Stein's Lemma}\label{appendix:classicalStein}
We consider the case of simple hypothesis testing.
We put the element of null hypothesis as $\rho$ and the element of alternative hypothesis as $\sigma$.
Also, we consider the probability distributions $p,q$ on the sample space $\Omega=\{1,\ldots,m\}$, where 
the distribution of null hypothesis is $p$ and the distribution of alternative hypothesis is $q$.
Now, we proceed the procedure which we obtain set of $n$ events $A_n\subset \Omega_n:=\{1,\ldots,m\}^n$.
Then, the first type error is described as
\begin{align}
    \beta(A_n):=1-\sum_{(i_1,\ldots,i_n)\in A_n}p_{i_1}\cdots p_{i_n}.\label{equation:classicalstein1}
\end{align}
The second type error is described as
\begin{align}
    \alpha(A_n):=\sum_{((i_1,\ldots,i_n)\in A_n)}q_{i_1}\cdots q_{i_n}.\label{equation:classicalstein2}
\end{align}
We define the following error probability.
\begin{definition}\label{definition:classicalsteinbeta}
    Let $p,q$ be the probability distribution on the sample space $\Omega$.
    We fix a set $A_n\subset \Omega_n$.
    The first and second type errors are defined as \eqref{equation:classicalstein1} and \eqref{equation:classicalstein2} for the set $A_n$.
    Then, we define a error probability as follows:
    \begin{align}
        \beta^n_\epsilon(p||q):=\min_{A_n}\{\alpha(A_n)|\beta(A_n)\le \epsilon\},\quad 0<\epsilon<1.
    \end{align}
\end{definition}
The classical Stein's Lemma is represented as follows:
\begin{theorem}\label{theorem:classicalstein}
    Let $p,q$ be probability distributions on the sample space $\Omega$.
    Then the following relation holds for the error probability defined in Definition~\ref{definition:classicalsteinbeta}.
    \begin{align}
        \lim_{n\to\infty}-\frac{1}{n}\log\beta^n_{\epsilon}(p||q)=D(p||q),\quad 0<\forall \epsilon <1.
    \end{align} 
\end{theorem}
Similarly to the proof of Quantum Stein's Lemma, we define the following quantities:
\begin{definition}
    Let $p,q$ be the probability distributions on sample space $\Omega$.
    Then, for the family $\{A_n\subset \Omega_n\}$, we define the following quantities:
    \begin{align}
        B(p||q)&:=\sup_{\{A_n\}}\{\varliminf_{n\to\infty}-\frac{1}{n}\log\alpha(A_n)|\lim_{n\to\infty}\beta(A_n)=0\}.\\
B^\dag(p||q)&:=\sup_{\{A_n\}}\{\varliminf_{n\to\infty}-\frac{1}{n}\log\alpha(A_n)|\varliminf_{n\to\infty}\beta(A_n)<1\}.
    \end{align}
\end{definition}
Similarly to Appendix~\ref{appendix:proofofStein}, Classical Stein's Lemma implies the following theorem.
\begin{theorem}
    For probability distributions on a sample space $\Omega$, the following equality holds.
    \begin{align}
        B^\dag(p||q)=B(p||q)=D(p||q).
    \end{align}
\end{theorem}

\end{document}